\documentclass[letterpaper,onecolumn]{IEEEtran}

\usepackage{amsmath,amssymb,amsfonts,amsthm,mathtools}
\usepackage{bm}
\usepackage{cite}
\usepackage{url}
\usepackage{booktabs}
\usepackage{array}
\usepackage{tabularx}
\usepackage{float}
\usepackage{xcolor}
\usepackage{enumitem}
\usepackage[hidelinks]{hyperref}

\DeclareRobustCommand{\doi}[1]{doi: \href{https://doi.org/#1}{#1}}

\allowdisplaybreaks

\newtheorem{theorem}{Theorem}
\newtheorem{proposition}{Proposition}
\newtheorem{lemma}{Lemma}
\newtheorem{corollary}{Corollary}
\newtheorem{assumption}{Assumption}
\newtheorem{definition}{Definition}
\newtheorem{remark}{Remark}

\DeclareMathOperator{\diag}{diag}

\newcommand{\PL}{\mathrm{PL}}

\newcommand{\plim}{\xrightarrow{p}}

\title{Gaussian-Process Dynamics of Diagonal Expectation Propagation under Variance-Profile Gaussian Measurements}
\author{Fangqing~Xiao,~\IEEEmembership{Member,~IEEE}, and Dirk~T.~M.~Slock,~\IEEEmembership{Life~Fellow,~IEEE}%
\thanks{Fangqing Xiao is with the School of Information, Yunnan University, Kunming 650500, China (e-mail: fangqing.xiao@ynu.edu.cn).}%
\thanks{Dirk T. M. Slock is with EURECOM, Sophia Antipolis, 06410 Biot, France (e-mail: dirk.slock@eurecom.fr).}}

\begin{document}
\maketitle

\begin{abstract}
State-evolution analyses of approximate-message-passing and expectation-propagation-type algorithms rely on an effective-channel principle: after a suitable Onsager, orthogonal, or extrinsic correction, the nonlinear module receives a fresh scalar Gaussian observation. This paper studies this principle for diagonal expectation propagation under variance-profile Gaussian sensing matrices. The model preserves Gaussian conditioning, but removes the isotropy that supports the usual scalar decoupling arguments. We prove a finite-time large-system description in which the linear EP module remains Gaussian at the coordinate level, but is generally not a fresh scalar channel. Instead, the residuals form a coordinate-dependent Gaussian process whose covariance is shaped by the variance profile and by the finite linear history of the algorithm. The standard diagonal EP cavity cancels the instantaneous response of the incoming message, but may leave a component predictable from past residuals. We characterize this process through a conditioned matrix-Dyson-equation deterministic equivalent and a Schur-complement representation of the linear module. A Gaussian-regression decomposition then separates the predictable memory from the orthogonal innovation and yields an oracle state-evolution-level correction. Thus, under variance-profile measurements, the limiting object for diagonal EP is a Gaussian-process dynamics with profile-dependent memory rather than the conventional fresh-noise scalar state evolution.
\end{abstract}


\section{Introduction}
\label{sec:introduction}

\subsection{Motivation}
High-dimensional linear inference is a basic model behind compressed sensing, multiuser detection, Bayesian linear estimation, low-rank estimation, coding, and signal recovery with discrete or structured priors \cite{BarbierKrzakala2017,RushGreigVenkataramanan2017,LelargeMiolane2019,MontanariVenkataramanan2021}.  In its simplest form, one observes
\(\boldsymbol y=\boldsymbol A\boldsymbol x+\boldsymbol v\), where \(\boldsymbol A\) is a known sensing matrix, \(\boldsymbol x\) is the unknown signal, and \(\boldsymbol v\) is noise.  Classical compressed sensing emphasizes sparse signals, and early message-passing algorithms were often motivated by this setting.  The large-system inference problem itself is broader: the prior on \(\boldsymbol x\) may be any separable distribution for which scalar posterior inference is well defined.  Sparse priors, such as Bernoulli--Gaussian priors, are then special cases rather than the defining assumption.

A major reason for the success of message-passing methods is that they can turn a high-dimensional inference problem into a sequence of effective scalar estimation problems.  This idea appears in statistical physics through Thouless--Anderson--Palmer-type equations and in belief-propagation-type algorithms for dense linear systems \cite{OpperWinther2001,Kabashima2003,MezardMontanari2009}.  In the compressed-sensing literature, approximate message passing (AMP) made this principle explicit: with an appropriate correction term, the input to the nonlinear denoiser behaves, in the large-system limit, like a scalar Gaussian observation.  The resulting scalar recursion is known as state evolution (SE) \cite{DonohoMalekiMontanari2009,BayatiMontanari2011,FengVenkataramananRushSamworth2022}.

The scalar Gaussian interpretation is not automatic.  Since the same sensing matrix is reused at every iteration, the current iterate and the matrix are statistically dependent.  For independent Gaussian matrices, the Onsager correction cancels the leading predictable dependence, and the SE can be made rigorous by conditioning on the past iterates \cite{BayatiMontanari2011}.  When the matrix departs from this homogeneous Gaussian setting, the same scalar SE need not remain valid, and the standard AMP recursion may require additional modifications.  This observation motivated a wide class of extensions, including generalized AMP, expectation propagation (EP), expectation consistent inference, vector AMP (VAMP), orthogonal AMP (OAMP), and turbo-type message-passing methods \cite{Rangan2011,Minka2001,OpperWinther2005,RanganSchniterFletcher2019,MaPing2017}.

Although these algorithms are derived from different principles, many of them enforce a common decoupling mechanism.  The linear and nonlinear modules should exchange information that is asymptotically extrinsic, orthogonal, or cavity-like relative to the information already used.  In independent Gaussian, right-orthogonally invariant, unitarily invariant, or rotationally invariant models, such mechanisms can lead again to scalar or finite-dimensional Gaussian descriptions \cite{RanganSchniterFletcher2019,Takeuchi2020,Fan2022,VenkataramananKoglerMondelli2022,ZhongWangFan2024}.  Diagonal EP is a natural framework in which to examine this mechanism because its cavity construction explicitly removes the instantaneous contribution of the incoming message before forming the effective observation sent to the prior module.

Variance-profile Gaussian measurements form a different regime.  The entries of the sensing matrix are still Gaussian, but their variances are prescribed by a deterministic profile.  Thus Gaussian conditioning remains available, while rotational invariance and coordinate homogeneity are lost.  Conditioning on past algorithmic information no longer leaves an isotropic residual matrix; it leaves a Gaussian field whose covariance depends on the variance profile and on the revealed history.  Consequently, the usual implication from an extrinsic correction to a fresh scalar Gaussian observation is no longer evident.

This leads to the basic question considered in this paper.  For diagonal EP under variance-profile Gaussian measurements, does the standard cavity construction still produce the fresh Gaussian information required by an ordinary scalar SE description?  Or does the variance profile force one to track a richer object than one or two scalar effective variances?  The issue is not merely whether a Gaussian approximation can be written at one iteration, but whether the Gaussian information delivered to the prior module is genuinely fresh relative to the information delivered in previous iterations.

\subsection{Proof Strategy and Main Technical Difficulty}
The question raised above cannot be answered by looking only at the marginal distribution of a single linear-module output.  At iteration \(t\), the effective observation passed to the prior module is generated by the same sensing matrix that has already produced all previous messages.  It is therefore statistically coupled with the past algorithmic history.  The central issue is not whether one can postulate a Gaussian approximation at one step, but whether the current effective observation contains a Gaussian innovation that is fresh relative to the information already used.

This is the same type of dependence that underlies rigorous analyses of AMP and EP-type algorithms.  For independent Gaussian matrices, the conditioning method represents the past dependence through finitely many linear observations of the sensing matrix, after which the Onsager correction leaves a fresh Gaussian component \cite{BayatiMontanari2011}.  For unitarily invariant measurements, a related conditioning program characterizes the residual Haar randomness after partial information has been revealed \cite{Takeuchi2020}.  More recent analyses for rotationally invariant models use spectral invariance, free-cumulant Onsager corrections, or reductions to orthogonal and long-memory message passing to obtain state evolutions \cite{Fan2022,ZhongWangFan2024,VenkataramananKoglerMondelli2022,LiuMa2024,DudejaLiuMa2026}.  These works show that, under suitable invariance structures, the algorithmic correction can restore an effective fresh Gaussian channel.

The variance-profile Gaussian setting requires a different route.  Since the entries of the sensing matrix are Gaussian, conditioning on finitely many linear observations still leaves a Gaussian matrix.  However, the conditional covariance is no longer isotropic.  It is deformed by the variance profile and by the history revealed by the algorithm.  Consequently, the remaining randomness is Gaussian, but it is not automatically exchangeable across coordinates, nor is it automatically fresh across iterations.  This is the main technical distinction from the independent Gaussian and unitarily invariant settings.

Our proof therefore does not attempt to derive a scalar SE directly.  It first identifies the Gaussian object produced by the linear module.  To keep the conditioning step exact, we formulate a predictable-precision version of diagonal EP in which the conditioning field contains only the finite linear history of the sensing matrix.  Nonlinear finite-dimensional resolvent quantities, such as the actual diagonal posterior variances and adaptive precisions generated by the current random matrix, are not included in this conditioning field.  They are connected back to the adaptive recursion later through a separate perturbation argument.

After conditioning on the linear history, the linear module contains the resolvent of a correlated Gaussian matrix with a bounded deterministic deformation.  We control this resolvent through a regularized matrix-Dyson-equation deterministic equivalent for correlated random matrices \cite{AjankiErdosKruger2019,Erdos2019MDE}.  This external random-matrix input is used only at the level needed for the EP-conditioned block linearization: the proof verifies the required bounded-deformation, flat-covariance, and positive-loading hypotheses, but does not reprove the underlying matrix-Dyson-equation local law.

The next step is a coordinate-wise Schur-complement expansion of the linear module.  This expansion separates the instantaneous response to the incoming message from the residual fluctuation left by the conditioned matrix.  The instantaneous response is the component that the diagonal EP cavity is designed to remove.  The remaining fluctuation is Gaussian, because it is a finite-dimensional projection of the conditioned Gaussian matrix, but its covariance is inherited from the variance profile and from the past linear history.  Thus, the linear module naturally produces a coordinate-dependent Gaussian process rather than a fresh scalar Gaussian innovation.

Once this Gaussian process has been identified, the role of the EP cavity can be examined precisely.  The standard diagonal cavity removes the instantaneous contribution of the incoming message.  What remains, however, may still contain the component of the current Gaussian residual that is predictable from previous residuals.  The final step is therefore a finite-dimensional Gaussian regression: the current residual is decomposed into its conditional mean given the past residual history and an orthogonal Gaussian innovation, using standard Gaussian conditioning identities \cite{Anderson2003}.  This regression is the mechanism by which the fresh part of the effective observation is identified.

The memory correction appearing in this decomposition is an oracle state-evolution-level construction: it characterizes the innovation part of the limiting Gaussian process, and is not proposed as a finite-sample algorithm.  Auxiliary concentration and interpolation tools used later in the proof follow standard high-dimensional Gaussian arguments \cite{Chatterjee2014}.  Designing practical estimators of the memory coefficients, or low-complexity approximations of the oracle correction, is left outside the present scope.

\subsection{Related Work}
Classical rigorous SE theory for AMP concerns independent Gaussian or universality classes close to it.  The original compressed-sensing AMP and its rigorous dense-graph analysis are now standard references \cite{DonohoMalekiMontanari2009,BayatiMontanari2011}.  Subsequent work extended this program to more general AMP recursions, spatial coupling, independent but non-identically distributed Gaussian matrices, universality for symmetric matrices, and non-separable nonlinearities \cite{JavanmardMontanari2013,BayatiLelargeMontanari2015,BerthierMontanariNguyen2020,FengVenkataramananRushSamworth2022}.  These results establish scalar or low-dimensional SEs for AMP-type recursions with carefully designed Onsager terms.  The present paper studies a different object: the resolvent-based linear module and cavity operation of diagonal EP under a deterministic variance profile.

Closest to the present matrix model are recent AMP results for variance-profile matrices.  Sparse symmetric variance profiles were analyzed in \cite{Hachem2024}, elliptic non-symmetric matrices were treated in \cite{GueddariHachemNajim2025Elliptic}, and a general non-symmetric model with variance and correlation profiles was developed in \cite{GueddariHachemNajim2026General}.  A leave-one-out approach to non-asymptotic AMP with Gaussian variance profiles was proposed in \cite{BaoHanXu2023}.  Related inhomogeneous-noise models also appear in low-rank matrix estimation with block or entrywise heterogeneity \cite{GuionnetKoKrzakalaZdeborova2022,PakKoKrzakala2023}.  These works show that variance profiles naturally lead to coordinate-dependent Gaussian descriptions, rather than a single scalar variance.  Their focus, however, is the construction of AMP Onsager terms and the corresponding density-evolution equations.  Here the issue is different: diagonal EP already prescribes an extrinsic cavity, and the question is whether that cavity leaves a fresh Gaussian channel once the LMMSE/resolvent linear module is driven by a variance-profile Gaussian matrix.

A second line of work studies message passing beyond independent Gaussian matrices by exploiting orthogonal, unitary, or rotational invariance.  OAMP and VAMP provide state-evolution descriptions for right-orthogonally invariant or related matrix ensembles \cite{MaPing2017,RanganSchniterFletcher2019}.  Rigorous EP dynamics for unitarily invariant measurements were established in \cite{Takeuchi2020}.  Long-memory and spatially coupled OAMP/VAMP-type constructions further clarify convergence and optimality mechanisms under right-orthogonally invariant models \cite{Takeuchi2022LongMemory,Takeuchi2024SpatialOAMP}.  These works are close to diagonal EP in algorithmic structure, but they rely on invariance of singular vectors or Haar-type conditioning.  A deterministic variance profile breaks this structure at the entrywise level.

Rotationally invariant AMP theory has recently become substantially more developed.  Free-cumulant Onsager corrections and spectral-invariance arguments yield AMP dynamics for rotationally invariant matrices \cite{Fan2022}.  Orthogonally invariant ensembles with multivariate nonlinearities and spectral initialization are treated in \cite{ZhongWangFan2024}.  Rotationally invariant generalized linear models are analyzed through RI-GAMP-type constructions \cite{VenkataramananKoglerMondelli2022}, and a unified construction of AMP algorithms for rotationally invariant models is obtained by reducing general iterative templates to long-memory OAMP \cite{LiuMa2024}.  Spiked matrix models with rotationally invariant noise have also been analyzed through OAMP-type dynamics and optimized denoisers \cite{DudejaLiuMa2026}.  These works reinforce the importance of matrix structure in determining the correct Onsager or extrinsic correction, but their invariance assumptions are fundamentally different from an entrywise deterministic variance profile.

Model mismatch, replica predictions, and the relation between algorithmic dynamics and statistical-mechanics free energies form another relevant direction.  VAMP in mismatched generalized linear models with rotation-invariant matrices has been analyzed through macroscopic SE, replica-symmetric free energy, and de Almeida--Thouless instability calculations \cite{TakahashiKabashima2022}.  The effect of sensing-matrix spectra on EP-type recovery for generalized linear inverse problems has also been studied in \cite{MaXuMaleki2024}.  These works are complementary to ours: they emphasize rotationally invariant spectra, replica or macroscopic performance predictions, and fixed-point stability, whereas the present paper studies the finite-time conditional law of the diagonal EP linear module under an entrywise variance profile.

Memory and universality results provide further context.  Memory AMP and sufficient-statistic memory AMP use previous messages as algorithmic variables to improve convergence or enforce state-evolution consistency for non-i.i.d. or right-unitarily invariant systems \cite{LiuHuangKurkoski2022,LiuHuangKurkoski2022SSMAMP}.  The word ``memory'' has a different meaning here.  We do not introduce a long-memory algorithm; memory is the predictable component of the Gaussian residual process left by the standard diagonal EP cavity.  Universality results show that certain AMP state evolutions extend beyond the Gaussian or rotationally invariant ensembles from which they were derived, including generalized Wigner and white-noise ensembles with heterogeneous entrywise variances \cite{WangZhongFan2024}.  Non-asymptotic concentration results are also available for generalized AMP/VAMP-type algorithms with right rotationally invariant designs \cite{CademartoriRush2024}.  These results ask whether a known AMP/VAMP recursion and its SE persist over broader ensembles or at finite sample sizes.  Our question is instead whether the standard diagonal EP cavity itself produces fresh Gaussian information under a variance-profile Gaussian ensemble.

Finally, the proof uses tools from random matrix theory for correlated matrices.  Matrix-Dyson-equation methods provide deterministic equivalents and stability theory for random matrices with non-identically distributed or correlated entries \cite{BaiSilverstein2010,CouilletDebbah2011,AjankiErdosKruger2017,AjankiErdosKruger2019,Erdos2019MDE}.  We use this framework as an external random-matrix input to control the resolvent of the EP-conditioned variance-profile matrix.  The novelty is not a new local law, but the way the conditioned matrix-Dyson-equation response enters the EP dynamics: it determines both the coordinate-dependent linear response and the covariance kernel of the Gaussian residual process.  This is the mechanism through which variance profiles replace the scalar fresh-noise picture by a profile-dependent Gaussian memory process.

\subsection{Contributions}
In contrast to existing scalar state evolutions for independent Gaussian, unitarily invariant, or rotationally invariant models, the results below show that a variance profile changes the nature of the effective Gaussian channel itself.

This paper develops a finite-time large-system analysis of diagonal EP under variance-profile Gaussian measurements.  The contribution is threefold.  First, we identify the conditional Gaussian object generated by the EP linear module.  Conditioning only on the finite linear history of the sensing matrix, we derive a coordinate-wise Schur representation in which the linear-module output is decomposed into an instantaneous response and a residual Gaussian process.  The covariance of this process is described by a conditioned matrix-Dyson-equation response and therefore depends on the variance profile and on the revealed history.

Second, we show that this Gaussian description does not in general reduce to the fresh scalar Gaussian channels that underlie the usual SE picture for independent Gaussian or rotationally invariant models.  The standard diagonal EP cavity removes the instantaneous response of the incoming message, but it need not remove the part of the current residual that is predictable from previous residuals.  Thus, under a general variance profile, the effective state is not characterized only by one or two scalar variances; it is governed by coordinate-dependent response vectors and a profile-dependent covariance kernel.

Third, we characterize the innovation component of this Gaussian process.  By applying finite-dimensional Gaussian regression, we decompose the Schur residual into a predictable memory term and a fresh Gaussian innovation.  This yields an oracle state-evolution-level correction that exposes the fresh channel hidden inside the residual process.  The correction is used as a theoretical decoupling device rather than as a finite-sample algorithm.  A separate perturbation argument connects the predictable-precision dynamics used in the conditioning proof to the adaptive precision recursion of diagonal EP under finite-time regularity conditions.

\subsection{Organization}
Section~\ref{sec:preliminaries} collects notation, empirical regularity, Gaussian conditioning and regression identities, and the regularized MDE input used throughout the proof.  Section~\ref{sec:system_ep} introduces the variance-profile measurement model, formulates the regularized diagonal EP recursion, and derives the error-domain identities that isolate the fresh-cavity question.  Section~\ref{sec:main_results} states the main results: the finite-time Gaussian-process dynamic theorem, the memory defect of the standard diagonal EP cavity, the oracle memory-corrected state evolution, and the precision-replacement theorem.  Section~\ref{sec:proof_general} proves the general dynamic theorem by finite-time induction.  Section~\ref{sec:proof_main_consequences} derives the memory-defect and corrected-state-evolution consequences.  Section~\ref{sec:precision_replacement} proves the branch-wise precision-replacement result connecting the predictable-precision proof dynamics to the adaptive diagonal EP recursion.  Section~\ref{sec:conclusion} concludes the paper.  The appendices contain the auxiliary probability tools, Gaussian conditioning of the EP history, conditioned MDE construction, Schur-kernel construction, empirical Gaussian-law arguments, memory decomposition, regularity closure, and perturbation estimates.

\section{Preliminaries}
\label{sec:preliminaries}

The purpose of this section is to collect the probabilistic and random-matrix
tools used throughout the proof.  The role of this section is analogous to the
preliminary section in rigorous EP analyses for unitarily invariant
measurements: before introducing the algorithmic recursion, we state the
definitions and limit results that will be repeatedly invoked later.  In the
present variance-profile Gaussian setting, the Haar-matrix tools are replaced by
Gaussian conditioning under finite linear observations, a weakly dependent
Gaussian empirical law, and a regularized correlated-Gaussian matrix Dyson
equation (MDE) input.

\subsection{Notation}
\label{subsec:notation}

We use bold lower-case letters for vectors and bold upper-case letters for
matrices.  Greek letters may denote either scalars or diagonal/vector
parameters; when the object is a vector or matrix, we use boldface, e.g.,
\(\boldsymbol\gamma_t\) and \(\boldsymbol\Gamma_t=\diag(\boldsymbol\gamma_t)\).
Scalar coordinates are written without boldface, e.g., \(x_j\),
\(A_{ij}\), and \(\gamma_{t,j}\).  Finite histories obtained by stacking
vectors columnwise are matrices and are therefore denoted by bold upper-case
letters, e.g., \(\boldsymbol Q_t=(\boldsymbol q_0,\ldots,\boldsymbol q_{t-1})\).
Upper-case scalar random variables, such as \(G_{t,j}\) or \(\mathcal Z_{t,j}\),
follow the standard probability convention and are not vector notation.
For a matrix \(\boldsymbol B\),
\(\boldsymbol B^{\mathsf T}\) denotes the transpose,
\(\boldsymbol B^{-1}\) the inverse when it exists,
\(\boldsymbol B^\dagger\) the Moore--Penrose inverse,
\(\|\boldsymbol B\|\) the operator norm, \(\|\boldsymbol B\|_F\) the
Frobenius norm, and \(\operatorname{Tr}(\boldsymbol B)\) the trace.  The
identity matrix of size \(n\) is denoted by \(\boldsymbol I_n\).  For
symmetric matrices, \(\boldsymbol B\preceq\boldsymbol C\) denotes the
Loewner order.  The map \(\operatorname{vec}(\cdot)\) stacks the columns of a
matrix, and \(\odot\) denotes the Hadamard product.  Inner products are
written as \(\boldsymbol a^{\mathsf T}\boldsymbol b\) in the real case
considered here; the complex analogue would use the Hermitian transpose.  For
vectors on the signal and measurement sides, we use the normalized norms
\begin{align*}
    \|\boldsymbol v\|_N^2
    &:=
    \frac1N\|\boldsymbol v\|^2,
    \qquad
    \boldsymbol v\in\mathbb R^N,
    \\
    \|\boldsymbol u\|_M^2
    &:=
    \frac1M\|\boldsymbol u\|^2,
    \qquad
    \boldsymbol u\in\mathbb R^M .
\end{align*}
The symbols \(O_p(1)\), \(o_p(1)\), and \(o_p^{\ell_2}(1)\) always refer to the
large-system limit \(M,N\to\infty\) with \(M/N\to\delta\in(0,\infty)\), while
the iteration horizon is fixed.  The notation \(o_p^{\ell_2}(1)\) is used only
for empirical vector errors.  More precisely, for signal-side arrays
\(\boldsymbol a_N,\boldsymbol b_N\in\mathbb R^N\),
\begin{equation*}
    \boldsymbol a_N=\boldsymbol b_N+o_p^{\ell_2}(1)
\end{equation*}
means
\begin{equation*}
    \frac1N\|\boldsymbol a_N-\boldsymbol b_N\|^2\overset{p}{\longrightarrow}0.
\end{equation*}
For measurement-side arrays, the same notation uses the normalization \(M^{-1}\).
When a coordinate formula such as
\(a_{t,j}=b_{t,j}+o_p^{\ell_2}(1)\) is displayed, it is always understood in
this empirical vector sense, i.e., the residual array
\(\{a_{t,j}-b_{t,j}\}_{j=1}^N\) is \(o_p^{\ell_2}(1)\).  All empirical smallness
statements are finite-time statements; no uniformity in \(t\to\infty\) is
claimed.

The diagonal EP recursion has two modules.  Module A is the linear Gaussian
module, and module B is the separable prior module.  The message
\(\boldsymbol r_{B\to A}^t\) and precision \(\boldsymbol\Gamma_t\) are the
input to the linear module at iteration \(t\), while
\(\boldsymbol r_{A\to B}^t\) and \(\boldsymbol\Pi_t\) are the output of the
linear module.  Unless a superscript is displayed explicitly, the variables in
the dynamic theorem refer to the predictable-precision recursion used in the
conditioning proof.  The superscript \(\mathrm{act}\) is reserved for the
finite-dimensional adaptive recursion, \(\mathrm{orc}\) for the corresponding
predictable recursion, \(\mathrm G\) for Gaussian reference variables, and
\(\mathrm{SE}\) for scalar state-evolution reference variables.

The paper uses three different conditioning environments.  The symbol
\(\mathcal F_t^{\rm lin}\) denotes the predictable linear-history filtration:
it contains the finite linear histories used to condition the Gaussian matrix,
but it does not contain retained diagonal resolvents or actual adaptive cavity
precisions.  The symbol \(\mathcal G_{r,j}\) denotes the column-wise revealed
history used to expose the scalar projections of column \(j\).  The symbol
\(\mathcal P_t\) denotes the MDE-generated reference environment, while
\(\mathcal R_t\) denotes the finite-time regularity event used in the
induction.  These objects serve different purposes and should not be
interchanged.

We use \(\mathcal Z_{t,j}\) for the coordinate of the Schur Gaussian process
generated by the linear module.  This variable is generally not fresh.  Its
Gaussian regression innovation is
\begin{equation*}
    G_{t,j}
    =
    \mathcal Z_{t,j}
    -
    \mathbb E[\mathcal Z_{t,j}\mid \mathcal Z_{<t,j},\mathcal P_t],
\end{equation*}
and \(W_{t,j}\) denotes the standardized version of \(G_{t,j}\) when the
innovation variance is positive.

For reference, Table~\ref{tab:main_notation} lists the main symbols used in the
paper.
\begin{table}[H]
    \centering
    \caption{Main notation.}
    \label{tab:main_notation}
    \begin{tabularx}{\linewidth}{@{}lX@{}}
        \toprule
        Symbol & Meaning \\
        \midrule
        \(\boldsymbol A\) & Variance-profile Gaussian measurement matrix. \\
        \(\boldsymbol S,\boldsymbol S_j\) & Variance profile and the profile matrix of column \(j\). \\
        \(\boldsymbol q_t,\boldsymbol m_t,\boldsymbol h_t,\boldsymbol p_t\) & Error variables of the prior-to-linear message, linear belief, linear-to-prior cavity, and prior belief. \\
        \(\boldsymbol u_t\) & Measurement residual \(\boldsymbol w-\boldsymbol A\boldsymbol m_t\). \\
        \(T_{t,j}\) & MDE deterministic equivalent of the retained linear variance \(d_{t,j}\). \\
        \(\Theta_j^{r,s}\) & Conditioned MDE two-resolvent response associated with the \(j\)th column covariance. \\
        \(\zeta_j^{r,s}\) & Covariance kernel of the Schur Gaussian process. \\
        \(\mu_{t,j},\tau_{t,j}\) & Predictable memory and fresh innovation variance of the EP cavity. \\
        \bottomrule
    \end{tabularx}
\end{table}

\begin{table}[H]
    \centering
    \caption{Conditioning environments and superscripts.}
    \label{tab:environment_notation}
    \begin{tabularx}{\linewidth}{@{}lX@{}}
        \toprule
        Symbol & Meaning \\
        \midrule
        \(\mathcal F_t^{\rm lin}\) & Predictable linear-history filtration used for Gaussian conditioning. \\
        \(\mathcal G_{r,j}\) & Column-wise revealed history for scalar projections of column \(j\). \\
        \(\mathcal P_t\) & MDE-generated reference environment. \\
        \(\mathcal R_t\) & Regularity event for the finite-time induction. \\
        \({\rm act}\), \({\rm orc}\) & Actual adaptive recursion and predictable/oracle recursion. \\
        \({\rm G}\), \({\rm SE}\) & Gaussian-kernel reference and scalar state-evolution reference. \\
        \bottomrule
    \end{tabularx}
\end{table}

\subsection{Definitions}
\label{subsec:prelim_definitions}

We first define the random matrix class studied in this paper.

\begin{definition}[Variance-profile Gaussian matrix]
\label{def:vpg}
For each pair $(M,N)$, let
\begin{equation*}
    \boldsymbol S=\boldsymbol S_{M,N}
    =
    (s_{ij})_{1\le i\le M,\,1\le j\le N}
    \in \mathbb R_+^{M\times N}
\end{equation*}
be a deterministic variance profile.  A random matrix
$\boldsymbol A\in\mathbb R^{M\times N}$ is called a
\emph{variance-profile Gaussian matrix} with profile $\boldsymbol S$ if
\begin{equation*}
    A_{ij}
    =
    \sqrt{\frac{s_{ij}}{M}}\,Z_{ij},
    \qquad
    Z_{ij}\overset{\mathrm{i.i.d.}}{\sim}\mathcal N(0,1).
\end{equation*}
The profile is called \emph{uniformly elliptic} if there exist constants
$0<s_{\min}\le s_{\max}<\infty$, independent of $M$ and $N$, such that
\begin{equation*}
    s_{\min}\le s_{ij}\le s_{\max}
\end{equation*}
for all $i,j,M,N$.  Throughout the paper, the proportional asymptotic regime is
assumed:
\begin{equation*}
    \frac{M}{N}\to \delta\in(0,\infty).
\end{equation*}
\end{definition}

The normalization $1/M$ in Definition~\ref{def:vpg} keeps the column norms of
$\boldsymbol A$ of constant order.  If $s_{ij}\equiv 1$, the model reduces to the
standard i.i.d. Gaussian sensing matrix.  In general, no separability, low-rank
structure, or isotropy is assumed for $\boldsymbol S$.

\begin{definition}[Profile matrices]
\label{def:profile_matrices}
For each column index $j$, define the measurement-side profile matrix
\begin{equation*}
    \boldsymbol S_j
    :=
    \operatorname{diag}(s_{1j},s_{2j},\ldots,s_{Mj})
    \in\mathbb R^{M\times M}.
\end{equation*}
Under the uniformly elliptic profile condition,
\begin{equation*}
    s_{\min}\boldsymbol I_M
    \preceq
    \boldsymbol S_j
    \preceq
    s_{\max}\boldsymbol I_M
\end{equation*}
for every $j$.
\end{definition}

The matrices $\boldsymbol S_j$ will be used to describe the profile-weighted
residual covariance
\begin{equation*}
    \frac{1}{M}\boldsymbol u_r^{\mathsf T}\boldsymbol S_j\boldsymbol u_s,
\end{equation*}
which replaces the ordinary residual covariance appearing in isotropic models.

We next recall pseudo-Lipschitz functions, following the convention commonly used in rigorous AMP state-evolution analyses \cite{BayatiMontanari2011}.

\begin{definition}[Pseudo-Lipschitz functions]
\label{def:pl}
For $k\ge 1$, a function $\psi:\mathbb R^d\to\mathbb R$ is said to be
pseudo-Lipschitz of order $k$, written $\psi\in \PL(k)$, if there exists a
constant $L>0$ such that
\begin{equation*}
    |\psi(\boldsymbol x)-\psi(\boldsymbol y)|
    \le
    L\|\boldsymbol x-\boldsymbol y\|
    \left(
        1+\|\boldsymbol x\|^{k-1}+\|\boldsymbol y\|^{k-1}
    \right)
\end{equation*}
for all $\boldsymbol x,\boldsymbol y\in\mathbb R^d$.
\end{definition}

A pseudo-Lipschitz function of order one is Lipschitz continuous.  Moreover,
every $\psi\in \PL(k)$ has polynomial growth of order $k$, as shown in
Proposition~\ref{prop:poly_growth} below.

We also introduce a convenient notation for empirical smallness.

\begin{definition}[Empirical $\ell_2$ smallness and admissible arrays]
\label{def:admissible}
For signal-side vectors $\boldsymbol r_N\in\mathbb R^N$, we write
\begin{equation*}
    \boldsymbol r_N=o_p^{\ell_2}(1)
\end{equation*}
if
\begin{equation*}
    \frac1N\|\boldsymbol r_N\|^2
    \overset{p}{\longrightarrow}0.
\end{equation*}
For measurement-side vectors $\boldsymbol r_M\in\mathbb R^M$, the same notation
means
\begin{equation*}
    \frac1M\|\boldsymbol r_M\|^2
    \overset{p}{\longrightarrow}0.
\end{equation*}
The side is always clear from the dimension of the vector under consideration.

A sequence of signal-side vectors $\boldsymbol v_N\in\mathbb R^N$ is called
\emph{admissible} if
\begin{equation*}
    \frac1N\|\boldsymbol v_N\|^2=O_p(1),
    \qquad
    \frac{\|\boldsymbol v_N\|_\infty}{\sqrt N}
    \overset{p}{\longrightarrow}0.
\end{equation*}
A sequence of measurement-side vectors $\boldsymbol u_M\in\mathbb R^M$ is called
\emph{admissible} if
\begin{equation*}
    \frac1M\|\boldsymbol u_M\|^2=O_p(1),
    \qquad
    \frac{\|\boldsymbol u_M\|_\infty}{\sqrt M}
    \overset{p}{\longrightarrow}0.
\end{equation*}
A finite collection of signal-side and measurement-side vectors is called
\emph{regular} if all its elements are admissible, its normalized Gram matrices
are tight, and asymptotically redundant directions are removed or treated by
Moore--Penrose inverses.
\end{definition}

The regularity notion in Definition~\ref{def:admissible} is used only for
finite-time histories.  No statement in this paper concerns the limit
$t\to\infty$.

\begin{definition}[Fresh Gaussian innovation and memory defect]
\label{def:freshness}
Fix a coordinate $j$ and an iteration $t$.  Let $\mathcal P_t$ be an auxiliary
environment and let
\begin{equation*}
    \mathcal Z_{<t,j}
    :=(\mathcal Z_{0,j},\ldots,\mathcal Z_{t-1,j})^{\mathsf T}
\end{equation*}
be the past Gaussian history.  A centered Gaussian residual $\mathcal Z_{t,j}$
is called \emph{fresh relative to the past} if, conditionally on $\mathcal P_t$,
\begin{equation*}
    \mathbb E[\mathcal Z_{t,j}\mid \mathcal Z_{<t,j},\mathcal P_t]=0
\end{equation*}
and $\mathcal Z_{t,j}$ is independent of $\mathcal Z_{<t,j}$.  Equivalently, for
a jointly Gaussian history, freshness is the conditional orthogonality condition
\begin{equation*}
    \operatorname{Cov}(\mathcal Z_{t,j},\mathcal Z_{<t,j}\mid\mathcal P_t)=0.
\end{equation*}
In general, the \emph{memory defect} is the predictable Gaussian-regression term
\begin{equation*}
    \mathfrak m_{t,j}
    :=
    \mathbb E[\mathcal Z_{t,j}\mid \mathcal Z_{<t,j},\mathcal P_t].
\end{equation*}
Thus the current residual is fresh if and only if its memory defect vanishes.
When the EP cavity is rescaled by a deterministic factor, the corresponding
memory defect is rescaled by the same factor; this is the origin of
$\mu_{t,j}$ in Section~\ref{subsec:fresh_awgn_principle}.
\end{definition}

\subsection{Results}
\label{subsec:prelim_results}

We now state the probabilistic and random-matrix results used later.  Standard Gaussian conditioning and regression identities are recalled in the exact form used by the proof; see, for example, \cite[Ch. 2]{Anderson2003}.

\begin{proposition}[Polynomial growth bound]
\label{prop:poly_growth}
For any $k\ge 1$, there exists a constant $C_k>0$ such that
\begin{equation*}
    (a+b)^k\le C_k(a^k+b^k),
    \qquad a,b\ge0.
\end{equation*}
Consequently, if $\psi\in \PL(k)$, then there exists a constant $C_\psi>0$ such
that
\begin{equation*}
    |\psi(\boldsymbol x)|
    \le
    C_\psi(1+\|\boldsymbol x\|^k)
\end{equation*}
for all $\boldsymbol x$.
\end{proposition}

\begin{IEEEproof}
For $k\ge 1$, the convexity of $x\mapsto x^k$ on $\mathbb R_+$ gives
\begin{equation*}
    \left(\frac{a+b}{2}\right)^k
    \le
    \frac{a^k+b^k}{2}.
\end{equation*}
Hence
\begin{equation*}
    (a+b)^k
    \le
    2^{k-1}(a^k+b^k).
\end{equation*}
Thus the first claim holds with $C_k=2^{k-1}$.

For the second claim, apply Definition~\ref{def:pl} with $\boldsymbol y=\boldsymbol 0$:
\begin{equation*}
    |\psi(\boldsymbol x)|
    \le
    |\psi(\boldsymbol 0)|
    +
    L\|\boldsymbol x\|
    \left(
        1+\|\boldsymbol x\|^{k-1}
    \right).
\end{equation*}
Since $\|\boldsymbol x\|\le 1+\|\boldsymbol x\|^k$ for all $k\ge1$, there is a constant
$C_\psi>0$ such that
\begin{equation*}
    |\psi(\boldsymbol x)|
    \le
    C_\psi(1+\|\boldsymbol x\|^k).
\end{equation*}
\end{IEEEproof}

\begin{lemma}[Gaussian conditioning under finite linear observations]
\label{lem:gaussian_conditioning}
Let
\begin{equation*}
    \boldsymbol A\in\mathbb R^{M\times N}
\end{equation*}
be a variance-profile Gaussian matrix.  Define
\begin{equation*}
    \boldsymbol A_{\mathrm{vec}}
    :=
    \operatorname{vec}(\boldsymbol A)
    \in\mathbb R^{MN}.
\end{equation*}
Then
\begin{equation*}
    \boldsymbol A_{\mathrm{vec}}
    \sim
    \mathcal N(\boldsymbol 0,\boldsymbol\Sigma_A),
\end{equation*}
where
\begin{equation*}
    \boldsymbol\Sigma_A
    =
    \operatorname{diag}\left(\frac{s_{ij}}{M}\right)_{i,j}.
\end{equation*}
Let $\boldsymbol L$ be any deterministic matrix with finitely many rows, and
consider the finite linear observation
\begin{equation*}
    \boldsymbol L\boldsymbol A_{\mathrm{vec}}=\boldsymbol B.
\end{equation*}
Then the conditional distribution of $\boldsymbol A_{\mathrm{vec}}$ given
$\boldsymbol L\boldsymbol A_{\mathrm{vec}}=\boldsymbol B$ is Gaussian and can be written as
\begin{equation*}
    \boldsymbol A_{\mathrm{vec}}\mid
    \boldsymbol L\boldsymbol A_{\mathrm{vec}}=\boldsymbol B
    =
    \boldsymbol A_{\parallel,\mathrm{vec}}
    +
    \boldsymbol A_{\perp,\mathrm{vec}},
\end{equation*}
where
\begin{equation*}
    \boldsymbol A_{\parallel,\mathrm{vec}}
    =
    \boldsymbol\Sigma_A\boldsymbol L^{\mathsf T}
    (\boldsymbol L\boldsymbol\Sigma_A\boldsymbol L^{\mathsf T})^\dagger
    \boldsymbol B,
\end{equation*}
and
\begin{equation*}
    \boldsymbol A_{\perp,\mathrm{vec}}
    \sim
    \mathcal N(\boldsymbol 0,\boldsymbol\Sigma_{A|L})
\end{equation*}
with
\begin{equation*}
    \boldsymbol\Sigma_{A|L}
    =
    \boldsymbol\Sigma_A
    -
    \boldsymbol\Sigma_A\boldsymbol L^{\mathsf T}
    (\boldsymbol L\boldsymbol\Sigma_A\boldsymbol L^{\mathsf T})^\dagger
    \boldsymbol L\boldsymbol\Sigma_A.
\end{equation*}
Moreover,
\begin{equation*}
    \boldsymbol 0
    \preceq
    \boldsymbol\Sigma_{A|L}
    \preceq
    \boldsymbol\Sigma_A.
\end{equation*}
\end{lemma}

\begin{IEEEproof}
Let
\begin{equation*}
    \boldsymbol Y=\boldsymbol L\boldsymbol A_{\mathrm{vec}}.
\end{equation*}
Since $\boldsymbol A_{\mathrm{vec}}$ is Gaussian, the pair
$(\boldsymbol A_{\mathrm{vec}},\boldsymbol Y)$ is jointly Gaussian.  Its covariance
matrices are
\begin{equation*}
    \operatorname{Cov}(\boldsymbol A_{\mathrm{vec}},\boldsymbol Y)
    =
    \boldsymbol\Sigma_A\boldsymbol L^{\mathsf T},
\end{equation*}
and
\begin{equation*}
    \operatorname{Cov}(\boldsymbol Y)
    =
    \boldsymbol L\boldsymbol\Sigma_A\boldsymbol L^{\mathsf T}.
\end{equation*}
The standard Gaussian regression formula for a jointly Gaussian vector
\cite{Anderson2003}, with the Moore--Penrose inverse used to allow possible
rank deficiency, gives the conditional mean
\begin{equation*}
    \mathbb E[
        \boldsymbol A_{\mathrm{vec}}
        \mid
        \boldsymbol Y=\boldsymbol B
    ]
    =
    \boldsymbol\Sigma_A\boldsymbol L^{\mathsf T}
    (\boldsymbol L\boldsymbol\Sigma_A\boldsymbol L^{\mathsf T})^\dagger
    \boldsymbol B.
\end{equation*}
The conditional covariance is
\begin{equation*}
    \boldsymbol\Sigma_A
    -
    \boldsymbol\Sigma_A\boldsymbol L^{\mathsf T}
    (\boldsymbol L\boldsymbol\Sigma_A\boldsymbol L^{\mathsf T})^\dagger
    \boldsymbol L\boldsymbol\Sigma_A.
\end{equation*}
This proves the conditional Gaussian representation.

It remains to show the covariance contraction.  Let
\begin{equation*}
    \boldsymbol C
    =
    \boldsymbol L\boldsymbol\Sigma_A^{1/2}.
\end{equation*}
Then
\begin{equation*}
\begin{aligned}
    \boldsymbol\Sigma_{A|L}
    &=
    \boldsymbol\Sigma_A^{1/2}
    \left[
        \boldsymbol I
        -
        \boldsymbol C^{\mathsf T}
        (\boldsymbol C\boldsymbol C^{\mathsf T})^\dagger
        \boldsymbol C
    \right]
    \boldsymbol\Sigma_A^{1/2}.
\end{aligned}
\end{equation*}
The matrix
\begin{equation*}
    \boldsymbol P_{\boldsymbol C}
    :=
    \boldsymbol C^{\mathsf T}
    (\boldsymbol C\boldsymbol C^{\mathsf T})^\dagger
    \boldsymbol C
\end{equation*}
is the orthogonal projection onto the row space of $\boldsymbol C$.  Hence
$\boldsymbol 0\preceq \boldsymbol P_{\boldsymbol C}\preceq \boldsymbol I$, and consequently
\begin{equation*}
    \boldsymbol 0
    \preceq
    \boldsymbol I-\boldsymbol P_{\boldsymbol C}
    \preceq
    \boldsymbol I.
\end{equation*}
Therefore
\begin{equation*}
    \boldsymbol 0
    \preceq
    \boldsymbol\Sigma_{A|L}
    \preceq
    \boldsymbol\Sigma_A.
\end{equation*}
\end{IEEEproof}

\begin{corollary}[Flat covariance of the conditioned residual]
\label{cor:flat_cov}
Let
\begin{equation*}
    \boldsymbol A\mid
    \boldsymbol L\operatorname{vec}(\boldsymbol A)=\boldsymbol B
    =
    \boldsymbol A_{\parallel}
    +
    \boldsymbol A_{\perp}
\end{equation*}
be the matrix form of Lemma~\ref{lem:gaussian_conditioning}.  Then, for any
deterministic vectors $\boldsymbol p\in\mathbb R^M$ and
$\boldsymbol q\in\mathbb R^N$,
\begin{equation*}
    \operatorname{Var}\!\left(
        \boldsymbol p^{\mathsf T}\boldsymbol A_\perp \boldsymbol q\mid
        \boldsymbol L\operatorname{vec}(\boldsymbol A)=\boldsymbol B
    \right)
    \le
    \frac{s_{\max}}{M}\|\boldsymbol p\|^2\|\boldsymbol q\|^2.
\end{equation*}
\end{corollary}

\begin{IEEEproof}
Let
\begin{equation*}
    \boldsymbol H
    :=
    \operatorname{vec}(\boldsymbol P\boldsymbol Q^{\mathsf T})
    \in\mathbb R^{MN}.
\end{equation*}
Then
\begin{equation*}
    \boldsymbol P^{\mathsf T}\boldsymbol A_\perp \boldsymbol Q
    =
    \boldsymbol H^{\mathsf T}\operatorname{vec}(\boldsymbol A_\perp).
\end{equation*}
By Lemma~\ref{lem:gaussian_conditioning},
\begin{equation*}
    \boldsymbol\Sigma_{A|L}\preceq \boldsymbol\Sigma_A.
\end{equation*}
Therefore,
\begin{equation*}
\begin{aligned}
    \operatorname{Var}\!\left(
        \boldsymbol p^{\mathsf T}\boldsymbol A_\perp \boldsymbol q\mid
        \boldsymbol L\operatorname{vec}(\boldsymbol A)=\boldsymbol B
    \right)
    &=
    \boldsymbol H^{\mathsf T}\boldsymbol\Sigma_{A|L}\boldsymbol H       \\
    &\le
    \boldsymbol H^{\mathsf T}\boldsymbol\Sigma_A\boldsymbol H            \\
    &=
    \sum_{i=1}^M\sum_{j=1}^N
    \frac{s_{ij}}{M}P_i^2Q_j^2                                  \\
    &\le
    \frac{s_{\max}}{M}
    \left(\sum_{i=1}^M P_i^2\right)
    \left(\sum_{j=1}^N Q_j^2\right)                              \\
    &=
    \frac{s_{\max}}{M}\|\boldsymbol p\|^2\|\boldsymbol q\|^2.
\end{aligned}
\end{equation*}
\end{IEEEproof}

The next result is a weak law for empirical averages of weakly dependent
Gaussian arrays.  It will be invoked after an empirical Gaussian-law argument has shown that
the average off-diagonal covariance is negligible.

\begin{theorem}[Weakly dependent Gaussian empirical law]
\label{thm:weak_gaussian_pl}
Let
\begin{equation*}
    \{\boldsymbol\xi_{j,N}\in\mathbb R^d:1\le j\le N\}
\end{equation*}
be a triangular array of jointly Gaussian random vectors.  Suppose that, for
some $\epsilon>0$,
\begin{equation*}
    \sup_{j,N}\mathbb E\|\boldsymbol\xi_{j,N}\|^{4+\epsilon}<\infty,
\end{equation*}
and that the average off-diagonal covariance vanishes:
\begin{equation*}
    \frac1{N^2}
    \sum_{j\ne k}
    \left\|
        \operatorname{Cov}(\boldsymbol\xi_{j,N},\boldsymbol\xi_{k,N})
    \right\|_{\mathrm{op}}
    \longrightarrow 0 .
\end{equation*}
Then, for every $\psi\in \PL(2)$,
\begin{equation*}
    \frac1N\sum_{j=1}^N
    \psi(\boldsymbol\xi_{j,N})
    -
    \frac1N\sum_{j=1}^N
    \mathbb E[\psi(\boldsymbol\xi_{j,N})]
    \overset{p}{\longrightarrow}0.
\end{equation*}
\end{theorem}

\begin{IEEEproof}
We first prove the claim for bounded continuously differentiable functions
$\psi$ with bounded gradient.  The bounded Lipschitz case follows by standard
mollification, and the pseudo-Lipschitz case is treated at the end by
truncation.

Let
\begin{equation*}
    Y_{j,N}
    :=
    \psi(\boldsymbol\xi_{j,N})
    -
    \mathbb E[\psi(\boldsymbol\xi_{j,N})].
\end{equation*}
It suffices to show
\begin{equation*}
    \operatorname{Var}\left(
        \frac1N\sum_{j=1}^N \psi(\boldsymbol\xi_{j,N})
    \right)
    \to0.
\end{equation*}
We have
\begin{equation*}
\begin{aligned}
    \operatorname{Var}\left(
        \frac1N\sum_{j=1}^N \psi(\boldsymbol\xi_{j,N})
    \right)
    &=
    \frac1{N^2}\sum_{j=1}^N
    \operatorname{Var}(\psi(\boldsymbol\xi_{j,N}))                 \\
    &\quad+
    \frac1{N^2}\sum_{j\ne k}
    \operatorname{Cov}(
        \psi(\boldsymbol\xi_{j,N}),
        \psi(\boldsymbol\xi_{k,N})
    ).
\end{aligned}
\end{equation*}
Since $\psi$ is bounded, the diagonal term is $O(N^{-1})$.

We now control the off-diagonal term.  For jointly Gaussian vectors
$(\boldsymbol\xi,\boldsymbol\eta)$, the Gaussian interpolation identity
(the ``smart path'' formula; see, e.g., \cite{Chatterjee2014}) gives
\begin{equation*}
\begin{aligned}
    \operatorname{Cov}(\psi(\boldsymbol\xi),\psi(\boldsymbol\eta))
    =
    \int_0^1
    \mathbb E\!\left[
        \nabla\psi(\boldsymbol\xi_s)^{\mathsf T}
        \operatorname{Cov}(\boldsymbol\xi,\boldsymbol\eta)
        \nabla\psi(\boldsymbol\eta_s)
    \right]ds,
\end{aligned}
\end{equation*}
where $(\boldsymbol\xi_s,\boldsymbol\eta_s)$ is the usual Gaussian interpolation
between independent and fully coupled copies.  Therefore, if
$\|\nabla\psi\|_\infty\le L$, then
\begin{equation*}
    \left|
    \operatorname{Cov}(
        \psi(\boldsymbol\xi_{j,N}),
        \psi(\boldsymbol\xi_{k,N})
    )
    \right|
    \le
    L^2
    \left\|
        \operatorname{Cov}(\boldsymbol\xi_{j,N},\boldsymbol\xi_{k,N})
    \right\|_{\mathrm{op}}.
\end{equation*}
Consequently,
\begin{equation*}
    \frac1{N^2}\sum_{j\ne k}
    \left|
    \operatorname{Cov}(
        \psi(\boldsymbol\xi_{j,N}),
        \psi(\boldsymbol\xi_{k,N})
    )
    \right|
    \to0.
\end{equation*}
Thus the variance of the empirical average converges to zero, and Chebyshev's
inequality yields the desired convergence in probability for bounded
continuously differentiable $\psi$ with bounded gradient.  A bounded Lipschitz
function can be approximated uniformly by such functions; hence the conclusion
also holds for bounded Lipschitz $\psi$.

It remains to extend the result to $\psi\in \PL(2)$.  Let $\chi_K$ be a smooth
cutoff satisfying $0\le\chi_K\le1$, $\chi_K(\boldsymbol x)=1$ for
$\|\boldsymbol x\|\le K$, and $\chi_K(\boldsymbol x)=0$ for $\|\boldsymbol x\|\ge 2K$.
Define
\begin{equation*}
    \psi_K(\boldsymbol x)=\psi(\boldsymbol x)\chi_K(\boldsymbol x).
\end{equation*}
Then $\psi_K$ is bounded Lipschitz for each fixed $K$, so the result already
proved gives
\begin{equation*}
    \frac1N\sum_{j=1}^N
    \psi_K(\boldsymbol\xi_{j,N})
    -
    \frac1N\sum_{j=1}^N
    \mathbb E[\psi_K(\boldsymbol\xi_{j,N})]
    \overset{p}{\longrightarrow}0.
\end{equation*}
By Proposition~\ref{prop:poly_growth}, since $\psi\in \PL(2)$,
\begin{equation*}
    |\psi(\boldsymbol x)|
    \le
    C_\psi(1+\|\boldsymbol x\|^2).
\end{equation*}
Thus
\begin{equation*}
    |\psi(\boldsymbol x)-\psi_K(\boldsymbol x)|
    \le
    C_\psi(1+\|\boldsymbol x\|^2)\boldsymbol 1\{\|\boldsymbol x\|>K\}.
\end{equation*}
Using Hölder's inequality and the uniform $(4+\epsilon)$-moment bound,
\begin{equation*}
    \sup_{j,N}
    \mathbb E\left[
        (1+\|\boldsymbol\xi_{j,N}\|^2)
        \boldsymbol 1\{\|\boldsymbol\xi_{j,N}\|>K\}
    \right]
    \longrightarrow 0
\end{equation*}
as $K\to\infty$.  Hence both the empirical and expectation-level truncation
errors vanish uniformly in probability as $K\to\infty$.  Letting first
$N\to\infty$ and then $K\to\infty$ completes the proof.
\end{IEEEproof}

We next state the only external random-matrix input used in the paper.  It is a
regularized admissible-quadratic-form specialization of the correlated-Gaussian
matrix-Dyson-equation framework for correlated random matrices
\cite{AjankiErdosKruger2017,AjankiErdosKruger2019,Erdos2019MDE}.  In
particular, the stability mechanism and deterministic resolvent approximation
used below are the regularized quadratic-form consequences of that MDE
local-law framework; see, in particular, the MDE stability and local-law results
in \cite{AjankiErdosKruger2017,AjankiErdosKruger2019}.  We state only the form
needed for the EP-conditioned block linearization.  The present paper does not
reprove the local law; its random-matrix task is to verify that the matrices
created by the predictable EP history satisfy the primitive conditions of this
regularized MDE input.

\begin{theorem}[External RMT input: regularized correlated-Gaussian MDE]
\label{thm:external_mde}
\label{lem:regularized_mde_input}
Let
\begin{equation*}
    \boldsymbol K_N=\boldsymbol D_N+\boldsymbol W_N
\end{equation*}
be an $n_N\times n_N$ real symmetric random matrix, where $\boldsymbol W_N$ is
centered Gaussian.  Define its covariance operator by
\begin{equation*}
    \mathcal S_N[\boldsymbol R]
    =
    \mathbb E[\boldsymbol W_N \boldsymbol R \boldsymbol W_N].
\end{equation*}
Assume that the following primitive conditions hold.

\begin{enumerate}
    \item \emph{Flat covariance:} for all admissible deterministic vectors
    $\boldsymbol a$ and $\boldsymbol b$,
    \begin{equation*}
        \operatorname{Var}(\boldsymbol a^{\mathsf T}\boldsymbol W_N\boldsymbol b)
        \le
        \frac{C}{N}\|\boldsymbol a\|^2\|\boldsymbol b\|^2.
    \end{equation*}

    \item \emph{Bounded deterministic deformation:}
    \begin{equation*}
        \|\boldsymbol D_N\|\le C.
    \end{equation*}

    \item \emph{Regularized block loading:} $\boldsymbol D_N$ has the block form
    \begin{equation*}
        \boldsymbol D_N=
        \begin{pmatrix}
            \boldsymbol I_M & \boldsymbol B_N\\
            \boldsymbol B_N^{\mathsf T} & -\boldsymbol\Gamma_N
        \end{pmatrix},
    \end{equation*}
    with
    \begin{equation*}
        \|\boldsymbol B_N\|\le C,
        \qquad
        0<\gamma_{\min}\boldsymbol I_N
        \preceq
        \boldsymbol\Gamma_N
        \preceq
        \gamma_{\max}\boldsymbol I_N.
    \end{equation*}

    \item \emph{Perturbation uniformity:} the same conclusions below hold
    uniformly for
    \begin{equation*}
        \boldsymbol K_N(\theta)
        =
        \boldsymbol D_N+\theta\boldsymbol H_N+\boldsymbol W_N
    \end{equation*}
    for all $|\theta|\le\theta_0$ and all bounded admissible deterministic
    insertions $\boldsymbol H_N$.
\end{enumerate}

Then the MDE
\begin{equation*}
    \boldsymbol M_N^{-1}
    =
    \boldsymbol D_N-\mathcal S_N[\boldsymbol M_N]
\end{equation*}
has a unique bounded stable solution.  Moreover, if
\begin{equation*}
    \boldsymbol G_N=\boldsymbol K_N^{-1},
\end{equation*}
then, for admissible deterministic vectors $\boldsymbol a$ and $\boldsymbol b$,
\begin{equation*}
    \boldsymbol a^{\mathsf T}
    (\boldsymbol G_N-\boldsymbol M_N)
    \boldsymbol b
    =
    o_p(1).
\end{equation*}
Furthermore, the averaged diagonal deterministic equivalent holds:
\begin{equation*}
    \frac1N\sum_{j=1}^N
    \left|
        \boldsymbol E_{M+j}^{\mathsf T}
        (\boldsymbol G_N-\boldsymbol M_N)
        \boldsymbol E_{M+j}
    \right|^2
    \overset{p}{\longrightarrow}0.
\end{equation*}
\end{theorem}

\begin{remark}[Source and use of the MDE input]
\label{rem:mde_source_use}
Theorem~\ref{thm:external_mde} is not an additional algorithmic assumption.  It
is the regularized quadratic-form version in which the correlated-Gaussian MDE
stability and local-law theory is used in this paper.  The underlying existence,
stability, and deterministic-resolvent approximation are standard consequences
of the MDE framework for random matrices with correlations
\cite{AjankiErdosKruger2017,AjankiErdosKruger2019,Erdos2019MDE}.  In the present
EP application, the compact positive loading in the linear module keeps the
block resolvent in a stable off-singular regime.  The EP-specific work is to
verify that, after conditioning on the finite linear history, the resulting
block matrix has flat covariance, bounded deterministic deformation, and
compact loading; these verifications are carried out in
Appendix~\ref{app:bounded_deformation_mde}.  The perturbative two-resolvent
form in Corollary~\ref{cor:two_resolvent} follows by applying the same stable
MDE input to a bounded deterministic insertion and differentiating the stable
MDE solution.
\end{remark}

The perturbation-uniform part of Theorem~\ref{thm:external_mde} yields a
two-resolvent deterministic equivalent.

\begin{corollary}[Perturbative two-resolvent equivalent]
\label{cor:two_resolvent}
Let $\boldsymbol K=\boldsymbol D+\boldsymbol W$ and $\boldsymbol M$ satisfy the assumptions and
conclusion of Theorem~\ref{thm:external_mde}.  Let $\boldsymbol H$ be a bounded
admissible deterministic insertion, and define
\begin{equation*}
    \boldsymbol G(\theta)
    =
    (\boldsymbol K+\theta\boldsymbol H)^{-1}.
\end{equation*}
Let $\boldsymbol M(\theta)$ be the MDE solution associated with
$\boldsymbol D+\theta\boldsymbol H$:
\begin{equation*}
    \boldsymbol M(\theta)^{-1}
    =
    \boldsymbol D+\theta\boldsymbol H-\mathcal S[\boldsymbol M(\theta)].
\end{equation*}
Define
\begin{equation*}
    \mathcal L[\boldsymbol H]
    :=
    -\left.
    \frac{d}{d\theta}
    \boldsymbol M(\theta)
    \right|_{\theta=0}.
\end{equation*}
Then $\mathcal L[\boldsymbol H]$ is the solution of
\begin{equation*}
    \mathcal L[\boldsymbol H]
    =
    \boldsymbol M\boldsymbol H\boldsymbol M
    +
    \boldsymbol M\mathcal S[\mathcal L[\boldsymbol H]]\boldsymbol M,
\end{equation*}
and, in the admissible quadratic-form sense of
Theorem~\ref{thm:external_mde},
\begin{equation*}
    \boldsymbol a^{\mathsf T}
    \bigl(\boldsymbol G\boldsymbol H\boldsymbol G-\mathcal L[\boldsymbol H]\bigr)
    \boldsymbol b
    \overset{p}{\longrightarrow}0
\end{equation*}
for all admissible deterministic test vectors $\boldsymbol a$ and $\boldsymbol b$.
\end{corollary}

\begin{IEEEproof}
Since
\begin{equation*}
    \frac{d}{d\theta}\boldsymbol G(\theta)
    =
    -\boldsymbol G(\theta)\boldsymbol H\boldsymbol G(\theta),
\end{equation*}
the desired random two-resolvent object is obtained by differentiating the
resolvent.  On the deterministic side, differentiating
\begin{equation*}
    \boldsymbol M(\theta)^{-1}
    =
    \boldsymbol D+\theta\boldsymbol H-\mathcal S[\boldsymbol M(\theta)]
\end{equation*}
at $\theta=0$ gives
\begin{equation*}
    -\boldsymbol M^{-1}\boldsymbol M'(0)\boldsymbol M^{-1}
    =
    \boldsymbol H-\mathcal S[\boldsymbol M'(0)].
\end{equation*}
With
\begin{equation*}
    \mathcal L[\boldsymbol H]=-\boldsymbol M'(0),
\end{equation*}
this becomes
\begin{equation*}
    \mathcal L[\boldsymbol H]
    =
    \boldsymbol M\boldsymbol H\boldsymbol M
    +
    \boldsymbol M\mathcal S[\mathcal L[\boldsymbol H]]\boldsymbol M.
\end{equation*}
The perturbation-uniform deterministic equivalent in
Theorem~\ref{thm:external_mde} justifies the differentiation at the level of
admissible quadratic forms.
\end{IEEEproof}

Finally, we record the Gaussian regression identity used to separate the
predictable memory component from the fresh innovation.

\begin{lemma}[Gaussian regression with a possibly singular covariance]
\label{lem:gaussian_regression}
Let
\begin{equation*}
    \begin{pmatrix}
        \boldsymbol Y_{<}\\
        Y_t
    \end{pmatrix}
\end{equation*}
be a jointly Gaussian vector with covariance
\begin{equation*}
    \begin{pmatrix}
        \boldsymbol C_{<,<} & \boldsymbol c_{<,t}\\
        \boldsymbol c_{t,<} & c_{t,t}
    \end{pmatrix}.
\end{equation*}
Then
\begin{equation*}
    Y_t
    =
    \boldsymbol c_{t,<}
    \boldsymbol C_{<,<}^{\dagger}
    \boldsymbol Y_{<}
    +
    G_t,
\end{equation*}
where \(G_t\) is Gaussian and independent of \(\boldsymbol Y_<\), with variance
\begin{equation*}
    \nu_t
    =
    c_{t,t}
    -
    \boldsymbol c_{t,<}
    \boldsymbol C_{<,<}^{\dagger}
    \boldsymbol c_{<,t}.
\end{equation*}
The Moore--Penrose inverse automatically removes redundant Gaussian-history
directions, so no full-rank assumption on \(\boldsymbol C_{<,<}\) is required.
\end{lemma}

\begin{IEEEproof}
For a jointly Gaussian vector, the conditional expectation of \(Y_t\) given
\(\boldsymbol Y_<\) is the orthogonal projection of \(Y_t\) onto the closed
linear span of \(\boldsymbol Y_<\) in \(L^2\).  This projection is
\begin{equation*}
    \mathbb E[Y_t\mid \boldsymbol Y_<]
    =
    \boldsymbol c_{t,<}
    \boldsymbol C_{<,<}^{\dagger}
    \boldsymbol Y_< .
\end{equation*}
Define
\begin{equation*}
    G_t
    =
    Y_t
    -
    \boldsymbol c_{t,<}
    \boldsymbol C_{<,<}^{\dagger}
    \boldsymbol Y_< .
\end{equation*}
Its covariance with the past is
\begin{equation*}
\begin{aligned}
    \operatorname{Cov}(G_t,\boldsymbol Y_<)
    &=
    \boldsymbol c_{t,<}
    -
    \boldsymbol c_{t,<}
    \boldsymbol C_{<,<}^{\dagger}
    \operatorname{Cov}(\boldsymbol Y_<,\boldsymbol Y_<)       \\
    &=
    \boldsymbol c_{t,<}
    -
    \boldsymbol c_{t,<}
    \boldsymbol C_{<,<}^{\dagger}
    \boldsymbol C_{<,<}.
\end{aligned}
\end{equation*}
For a valid covariance matrix, the row vector \(\boldsymbol c_{t,<}\) lies in
the row space of \(\boldsymbol C_{<,<}\).  Hence
\begin{equation*}
    \boldsymbol c_{t,<}
    \boldsymbol C_{<,<}^{\dagger}
    \boldsymbol C_{<,<}
    =
    \boldsymbol c_{t,<},
\end{equation*}
and therefore
\begin{equation*}
    \operatorname{Cov}(G_t,\boldsymbol Y_<)=\boldsymbol 0.
\end{equation*}
Since \((G_t,\boldsymbol Y_<)\) is jointly Gaussian, zero covariance implies
independence.  The variance of \(G_t\) is
\begin{equation*}
    \operatorname{Var}(G_t)
    =
    c_{t,t}
    -
    \boldsymbol c_{t,<}
    \boldsymbol C_{<,<}^{\dagger}
    \boldsymbol c_{<,t}.
\end{equation*}
This proves the claim.
\end{IEEEproof}

\section{System Model and Diagonal Expectation Propagation}
\label{sec:system_model}
\label{sec:system_ep}

In this section, we introduce the measurement model and formulate the diagonal
EP recursion analyzed in this paper.  The presentation follows the standard
two-module interpretation of EP: a linear Gaussian module incorporates the
likelihood, while a separable prior module incorporates the prior distribution
of the signal.  We then rewrite the recursion in terms of estimation errors.
This error-domain representation is the starting point of the conditioning
argument developed in the sequel.

\subsection{Assumptions}
\label{subsec:system_assumptions}

We consider the real-valued linear measurement model
\begin{equation}
    \boldsymbol y
    =
    \boldsymbol A\boldsymbol x+\boldsymbol w,
    \label{eq:system_model}
\end{equation}
where $\boldsymbol x\in\mathbb R^N$ is the unknown signal vector,
$\boldsymbol A\in\mathbb R^{M\times N}$ is the measurement matrix, and
$\boldsymbol w\in\mathbb R^M$ is the noise vector.  The large-system limit is
taken in the proportional regime
\begin{equation}
    M/N\to\delta\in(0,\infty).
\end{equation}

We first state the assumptions used throughout the paper.

\begin{assumption}[Signal prior]
\label{ass:signal}
The signal vector
\begin{equation}
    \boldsymbol x=(x_1,\ldots,x_N)^{\mathsf T}
\end{equation}
has independent and identically distributed elements drawn from a non-degenerate
distribution $P_X$.  The scalar random variable $X\sim P_X$ satisfies
\begin{equation}
    0<\mathbb E[X^2]<\infty,
    \qquad
    \mathbb E|X|^{4+\epsilon}<\infty
\end{equation}
for some $\epsilon>0$.
\end{assumption}

No centering or unit-variance normalization is imposed.  The scalar denoiser
$\eta(\cdot;\pi)$ and posterior variance $v_B(\cdot;\pi)$ are always defined
with respect to the actual prior $P_X$.  The finite moment condition is used to
obtain empirical energy bounds, no-spike conditions, and uniform integrability
in the pseudo-Lipschitz convergence arguments.  The exponent $4+\epsilon$ is a
technical regularity condition of the present proof rather than a claimed
optimal requirement.  In particular, finite second moment alone does not imply
finite fourth moment; a more refined truncation argument may weaken this
condition, but such a sharpening is not pursued here.

\begin{assumption}[Variance-profile Gaussian measurements]
\label{ass:matrix}
The measurement matrix $\boldsymbol A$ is a uniformly elliptic
variance-profile Gaussian matrix in the sense of
Definition~\ref{def:vpg}.  Equivalently,
\begin{equation}
    A_{ij}
    =
    \sqrt{\frac{s_{ij}}{M}}Z_{ij},
    \qquad
    Z_{ij}\overset{\mathrm{i.i.d.}}{\sim}\mathcal N(0,1),
\end{equation}
where the deterministic profile satisfies
\begin{equation}
    0<s_{\min}\le s_{ij}\le s_{\max}<\infty
\end{equation}
uniformly in $i,j,M,N$.  The matrix $\boldsymbol A$ is independent of
$\boldsymbol x$ and $\boldsymbol w$.
\end{assumption}

Unlike unitarily invariant measurement models, the ensemble in
Assumption~\ref{ass:matrix} does not possess a Haar singular-vector
representation in general.  The proof therefore cannot rely on rotational
invariance.  Instead, it will use Gaussian conditioning under finite linear
histories and the correlated-Gaussian MDE input stated in
Theorem~\ref{thm:external_mde}.

\begin{assumption}[Noise]
\label{ass:noise}
The noise vector is Gaussian,
\begin{equation}
    \boldsymbol w\sim
    \mathcal N(\boldsymbol 0,\gamma_w^{-1}\boldsymbol I_M),
\end{equation}
and is independent of $\boldsymbol A$ and $\boldsymbol x$.  The noise precision
$\gamma_w>0$ is assumed known by the algorithm.
\end{assumption}

The Gaussian noise assumption is not essential for the linear algebraic part of
the proof, but it keeps the Bayesian interpretation of the linear module
transparent and avoids additional moment assumptions on transformed noise
vectors.

\subsection{Diagonal Expectation Propagation}
\label{subsec:diagonal_ep}

We now formulate the diagonal EP recursion.  The algorithm alternates between a
linear Gaussian module, referred to as module A, and a separable prior module,
referred to as module B.  At iteration $t$, module A receives a Gaussian
message from module B with mean
\begin{equation}
    \boldsymbol r_{B\to A}^t\in\mathbb R^N
\end{equation}
and diagonal precision matrix
\begin{equation}
    \boldsymbol\Gamma_t
    =
    \operatorname{diag}(\gamma_{t,1},\ldots,\gamma_{t,N}).
\end{equation}
The Gaussian message from module B is the diagonal Gaussian factor
\begin{equation}
    q_{B\to A}^t(\boldsymbol x)
    =
    \mathcal N\!\left(
        \boldsymbol x;
        \boldsymbol r_{B\to A}^t,
        \boldsymbol\Gamma_t^{-1}
    \right).
\end{equation}
Thus the word ``diagonal'' refers to the covariance retained by the EP message,
not to an additional probabilistic assumption on the true signal.

Combining this Gaussian message with the likelihood induced by
\eqref{eq:system_model}, module A forms the Gaussian belief
\begin{equation}
    b_A^t(\boldsymbol x)
    \propto
    \exp\left\{
        -\frac{\gamma_w}{2}
        \|\boldsymbol y-\boldsymbol A\boldsymbol x\|^2
        -\frac12
        \|\boldsymbol x-\boldsymbol r_{B\to A}^t\|_{\boldsymbol\Gamma_t}^2
    \right\},
\end{equation}
where
\begin{equation}
    \|\boldsymbol v\|_{\boldsymbol\Gamma_t}^2
    =
    \boldsymbol v^{\mathsf T}\boldsymbol\Gamma_t\boldsymbol v.
\end{equation}
Thus the covariance and mean of the module-A belief are
\begin{align}
    \boldsymbol C_t
    &=
    \left(
        \gamma_w\boldsymbol A^{\mathsf T}\boldsymbol A
        +
        \boldsymbol\Gamma_t
    \right)^{-1},
    \label{eq:C_t_def}
    \\
    \widehat{\boldsymbol x}_{A}^t
    &=
    \boldsymbol C_t
    \left(
        \gamma_w\boldsymbol A^{\mathsf T}\boldsymbol y
        +
        \boldsymbol\Gamma_t\boldsymbol r_{B\to A}^t
    \right).
    \label{eq:xA_hat_def}
\end{align}

Let
\begin{equation}
    d_{t,j}
    =
    [\boldsymbol C_t]_{jj},
    \qquad
    \boldsymbol D_t
    =
    \operatorname{diag}(d_{t,1},\ldots,d_{t,N}).
\end{equation}
The diagonal EP approximation retains only the diagonal marginal variances of
the Gaussian belief.  The outgoing Gaussian message from module A to module B
is obtained by subtracting the incoming Gaussian natural parameters from the
diagonal marginal belief.  Its diagonal precision matrix is
\begin{equation}
    \boldsymbol\Pi_t
    =
    \boldsymbol D_t^{-1}
    -
    \boldsymbol\Gamma_t,
    \label{eq:Pi_t_def}
\end{equation}
and its mean $\boldsymbol r_{A\to B}^t$ is defined through
\begin{equation}
    \boldsymbol\Pi_t
    \boldsymbol r_{A\to B}^t
    =
    \boldsymbol D_t^{-1}\widehat{\boldsymbol x}_{A}^t
    -
    \boldsymbol\Gamma_t\boldsymbol r_{B\to A}^t .
    \label{eq:r_A_to_B_def}
\end{equation}

In stable implementations, the diagonal precisions in
\eqref{eq:Pi_t_def} and in the corresponding prior-to-linear update below are
kept in compact positive intervals.  We follow this standard convention and
record it as a regularity condition rather than as a separate algorithmic
variant.

Module B processes the components of $\boldsymbol r_{A\to B}^t$ separately.
For a scalar precision $\pi>0$, define the scalar Gaussian observation model
\begin{equation}
    R=X+\pi^{-1/2}Z,
    \qquad
    Z\sim\mathcal N(0,1),
    \qquad
    X\sim P_X,
\end{equation}
with $X$ and $Z$ independent.  The scalar posterior mean and variance are
denoted by
\begin{align}
    \eta(r;\pi)
    &=
    \mathbb E[X\mid R=r],
    \label{eq:scalar_eta_def}
    \\
    v_B(r;\pi)
    &=
    \operatorname{Var}(X\mid R=r).
    \label{eq:scalar_vB_def}
\end{align}
Equivalently,
\begin{align}
    \eta(r;\pi)
    &=
    \frac{
        \int x\,p_X(x)
        \exp\{-\frac{\pi}{2}(r-x)^2\}\,dx
    }{
        \int p_X(x)
        \exp\{-\frac{\pi}{2}(r-x)^2\}\,dx
    },
    \\
    v_B(r;\pi)
    &=
    \frac{
        \int (x-\eta(r;\pi))^2 p_X(x)
        \exp\{-\frac{\pi}{2}(r-x)^2\}\,dx
    }{
        \int p_X(x)
        \exp\{-\frac{\pi}{2}(r-x)^2\}\,dx
    }.
\end{align}

Applying these scalar functions componentwise gives the module-B posterior mean
\begin{equation}
    \widehat{x}_{B,j}^t
    =
    \eta(r_{A\to B,j}^t;\pi_{t,j}),
    \qquad
    j=1,\ldots,N,
    \label{eq:xB_hat_component}
\end{equation}
where
\begin{equation}
    \boldsymbol\Pi_t
    =
    \operatorname{diag}(\pi_{t,1},\ldots,\pi_{t,N}).
\end{equation}
The corresponding diagonal posterior variance matrix is
\begin{equation}
    \boldsymbol V_{B,t}
    =
    \operatorname{diag}(v_{B,t,1},\ldots,v_{B,t,N}),
    \qquad
    v_{B,t,j}
    =
    v_B(r_{A\to B,j}^t;\pi_{t,j}).
    \label{eq:VB_t_def}
\end{equation}

The outgoing Gaussian message from module B to module A is again obtained by
extrinsic subtraction of natural parameters.  Its diagonal precision matrix is
\begin{equation}
    \boldsymbol\Gamma_{t+1}
    =
    \boldsymbol V_{B,t}^{-1}
    -
    \boldsymbol\Pi_t,
    \label{eq:Gamma_next_def}
\end{equation}
and its mean $\boldsymbol r_{B\to A}^{t+1}$ satisfies
\begin{equation}
    \boldsymbol\Gamma_{t+1}
    \boldsymbol r_{B\to A}^{t+1}
    =
    \boldsymbol V_{B,t}^{-1}
    \widehat{\boldsymbol x}_{B}^t
    -
    \boldsymbol\Pi_t
    \boldsymbol r_{A\to B}^t .
    \label{eq:r_B_to_A_next_def}
\end{equation}
The iteration is initialized with a deterministic Gaussian message, for example
\begin{equation}
    \boldsymbol r_{B\to A}^0=\boldsymbol 0,
    \qquad
    \boldsymbol\Gamma_0=\gamma_0\boldsymbol I_N
\end{equation}
for some $\gamma_0>0$ in the admissible precision range.

Table~\ref{tab:ep_iteration} summarizes the initialization and one EP iteration
by pointing to the numbered equations defining each step.  The table is included
only as a reader aid; the analysis below uses the exact identities
\eqref{eq:C_t_def}--\eqref{eq:r_B_to_A_next_def} and their error-domain forms.

\begin{table}[H]
    \centering
    \caption{Regularized diagonal EP recursion.}
    \label{tab:ep_iteration}
    \begin{tabularx}{\linewidth}{@{}lX@{}}
        \toprule
        Stage & Definition \\
        \midrule
        Initialization & Choose a deterministic regular message \(\boldsymbol r_{B\to A}^0\) and a clipped diagonal precision \(\boldsymbol\Gamma_0\). \\
        Linear belief & Compute \(\boldsymbol C_t\) and \(\widehat{\boldsymbol x}_A^t\) from \eqref{eq:C_t_def} and \eqref{eq:xA_hat_def}. \\
        Linear-to-prior message & Form \(\boldsymbol\Pi_t\) and \(\boldsymbol r_{A\to B}^t\) by \eqref{eq:Pi_t_def} and \eqref{eq:r_A_to_B_def}, followed by clipping to the admissible precision interval. \\
        Prior belief & Apply the scalar posterior maps \eqref{eq:scalar_eta_def}--\eqref{eq:scalar_vB_def} componentwise as in \eqref{eq:xB_hat_component} and \eqref{eq:VB_t_def}. \\
        Prior-to-linear message & Form \(\boldsymbol\Gamma_{t+1}\) and \(\boldsymbol r_{B\to A}^{t+1}\) by \eqref{eq:Gamma_next_def} and \eqref{eq:r_B_to_A_next_def}, again with precision clipping. \\
        \bottomrule
    \end{tabularx}
\end{table}

We impose the following regularity on the scalar module and the diagonal
Gaussian messages.

\begin{assumption}[Regularized scalar module and compact precisions]
\label{ass:ep_regular}
For every fixed iteration horizon $T$, the diagonal Gaussian messages are
regularized so that the precisions used by the recursion lie in compact
positive intervals:
\begin{equation}
    0<\pi_{\min}
    \le
    \pi_{t,j}
    \le
    \pi_{\max}<\infty,
\end{equation}
and
\begin{equation}
    0<\gamma_{\min}
    \le
    \gamma_{t,j}
    \le
    \gamma_{\max}<\infty.
\end{equation}
Moreover, on compact precision intervals, the scalar maps
$\eta(r;\pi)$ and $v_B(r;\pi)$ preserve $\PL(2)$ empirical convergence and the
finite moment bounds required in Definition~\ref{def:admissible}.
\end{assumption}

Assumption~\ref{ass:ep_regular} is the standing scalar-module regularity used
throughout the paper.  It should be read as the usual clipping/regularization
convention for diagonal EP, not as a new algorithmic variant.  The regularity
event in Section~\ref{sec:dynamic_proof} adds only finite-history
admissibility and Gram-stability conditions needed for the conditioning proof,
while Section~\ref{sec:precision_replacement} states the additional stability
requirements needed to replace the predictable precisions by the actual
finite-sample adaptive precisions.  The finite-dimensional recursion itself is
given by \eqref{eq:C_t_def}--\eqref{eq:r_B_to_A_next_def}; deterministic
equivalents of the diagonal quantities $d_{t,j}$ are not inserted into the
algorithm, but will be derived later from the MDE.

\subsection{Error Recursion}
\label{subsec:error_recursion}

For the large-system analysis, it is convenient to rewrite the diagonal EP
recursion in terms of estimation errors.  Define
\begin{align}
    \boldsymbol q_t
    &:=
    \boldsymbol r_{B\to A}^t-\boldsymbol x,
    \\
    \boldsymbol m_t
    &:=
    \widehat{\boldsymbol x}_{A}^t-\boldsymbol x,
    \\
    \boldsymbol h_t
    &:=
    \boldsymbol r_{A\to B}^t-\boldsymbol x,
    \\
    \boldsymbol p_t
    &:=
    \widehat{\boldsymbol x}_{B}^t-\boldsymbol x .
\end{align}
Substituting the system model \eqref{eq:system_model} into
\eqref{eq:xA_hat_def} gives the exact linear-module error recursion
\begin{equation}
    \boldsymbol m_t
    =
    \boldsymbol C_t
    \left(
        \gamma_w\boldsymbol A^{\mathsf T}\boldsymbol w
        +
        \boldsymbol\Gamma_t\boldsymbol q_t
    \right).
    \label{eq:m_error_recursion}
\end{equation}

Similarly, subtracting $\boldsymbol x$ from
\eqref{eq:r_A_to_B_def} and using
\begin{equation}
    \boldsymbol\Pi_t
    =
    \boldsymbol D_t^{-1}
    -
    \boldsymbol\Gamma_t
\end{equation}
yields
\begin{equation}
    \boldsymbol\Pi_t\boldsymbol h_t
    =
    \boldsymbol D_t^{-1}\boldsymbol m_t
    -
    \boldsymbol\Gamma_t\boldsymbol q_t .
\end{equation}
The prior-module error is componentwise
\begin{equation}
    p_{t,j}
    =
    \eta(x_j+h_{t,j};\pi_{t,j})-x_j,
    \qquad
    j=1,\ldots,N.
\end{equation}
Finally, subtracting $\boldsymbol x$ from
\eqref{eq:r_B_to_A_next_def} gives
\begin{equation}
    \boldsymbol\Gamma_{t+1}\boldsymbol q_{t+1}
    =
    \boldsymbol V_{B,t}^{-1}\boldsymbol p_t
    -
    \boldsymbol\Pi_t\boldsymbol h_t .
    \label{eq:q_next_error_recursion}
\end{equation}

Equations \eqref{eq:m_error_recursion}--\eqref{eq:q_next_error_recursion} are
exact finite-dimensional identities.  They will be analyzed in the large-system
limit by conditioning on the past error history.

We also define the measurement residual
\begin{equation}
    \boldsymbol u_t
    :=
    \boldsymbol w-\boldsymbol A\boldsymbol m_t.
    \label{eq:u_t_def}
\end{equation}
The residual $\boldsymbol u_t$ plays the role of the measurement-side history
created by the linear module.

To express the past iterations compactly, define the history matrices
\begin{align}
    \boldsymbol Q_t
    &:=
    (\boldsymbol q_0,\ldots,\boldsymbol q_{t-1})
    \in\mathbb R^{N\times t},
    \\
    \boldsymbol M_t
    &:=
    (\boldsymbol m_0,\ldots,\boldsymbol m_{t-1})
    \in\mathbb R^{N\times t},
    \\
    \boldsymbol H_t
    &:=
    (\boldsymbol h_0,\ldots,\boldsymbol h_{t-1})
    \in\mathbb R^{N\times t},
    \\
    \boldsymbol P_t
    &:=
    (\boldsymbol p_0,\ldots,\boldsymbol p_{t-1})
    \in\mathbb R^{N\times t},
    \\
    \boldsymbol U_t
    &:=
    (\boldsymbol u_0,\ldots,\boldsymbol u_{t-1})
    \in\mathbb R^{M\times t}.
\end{align}
From \eqref{eq:u_t_def}, the first set of linear history constraints is
\begin{equation}
    \boldsymbol A\boldsymbol M_t
    =
    \boldsymbol w\boldsymbol 1_t^{\mathsf T}
    -
    \boldsymbol U_t,
    \label{eq:history_AM}
\end{equation}
where $\boldsymbol 1_t$ denotes the all-one vector in $\mathbb R^t$.

A second set of linear constraints follows from the normal equations for the
linear module.  From \eqref{eq:m_error_recursion},
\begin{equation}
    \left(
        \gamma_w\boldsymbol A^{\mathsf T}\boldsymbol A
        +
        \boldsymbol\Gamma_t
    \right)\boldsymbol m_t
    =
    \gamma_w\boldsymbol A^{\mathsf T}\boldsymbol w
    +
    \boldsymbol\Gamma_t\boldsymbol q_t.
\end{equation}
Using $\boldsymbol u_t=\boldsymbol w-\boldsymbol A\boldsymbol m_t$, this becomes
\begin{equation}
    \gamma_w\boldsymbol A^{\mathsf T}\boldsymbol u_t
    =
    \boldsymbol\Gamma_t(\boldsymbol m_t-\boldsymbol q_t).
\end{equation}
Stacking this identity over previous iterations yields
\begin{equation}
    \boldsymbol A^{\mathsf T}\boldsymbol U_t
    =
    \gamma_w^{-1}
    \left(
        \boldsymbol\Gamma_0(\boldsymbol m_0-\boldsymbol q_0),
        \ldots,
        \boldsymbol\Gamma_{t-1}(\boldsymbol m_{t-1}-\boldsymbol q_{t-1})
    \right).
    \label{eq:history_ATU}
\end{equation}

The two identities \eqref{eq:history_AM} and \eqref{eq:history_ATU} are the
linear observations of $\boldsymbol A$ generated by the EP history.  They are
the starting point for the Gaussian conditioning analysis in the next sections.
Conditioned on a regular finite-time history, the variance-profile Gaussian
matrix will be decomposed as
\begin{equation}
    \boldsymbol A
    =
    \boldsymbol A_{\parallel,t}
    +
    \boldsymbol A_{\perp,t},
\end{equation}
where the first term is the conditional mean and the second term is a centered
correlated Gaussian residual.

\section{Main Results}
\label{sec:main_results}

This section states the main asymptotic results of the paper.  The guiding
question is whether the diagonal EP cavity delivered by the linear module can
be interpreted as a fresh scalar AWGN observation of each signal coordinate.
Such a fresh-channel interpretation is the basis of scalar state evolution in
unitarily invariant models.  In the variance-profile Gaussian setting, however,
the correct conclusion is more subtle: a Gaussian description still exists,
but the Gaussian term is generally not an innovation.  It contains a
predictable component generated by the past Gaussian history.

The main results below make this statement precise.  We first introduce the
MDE-generated quantities that replace the scalar state variables of classical
EP state evolution.  We then state a finite-time dynamic theorem, from which
the memory defect of standard diagonal EP and the memory-corrected
coordinate-wise state evolution follow.

The proof locations of the main statements are summarized in
Table~\ref{tab:proof_roadmap}.  The general dynamic theorem is proved in
Section~\ref{sec:dynamic_proof}; its technical inputs are proved in
Appendices~\ref{app:gaussian_conditioning_histories}--\ref{app:regularity_closure}.
The memory-defect and corrected-state-evolution consequences are proved in
Section~\ref{sec:proof_main_consequences}.  The precision-replacement theorem is
proved in Section~\ref{subsec:proof_precision_replacement}, with auxiliary
estimates collected in Appendix~\ref{app:precision_replacement}.

\begin{table}[H]
    \centering
    \caption{Proof roadmap for the main results.}
    \label{tab:proof_roadmap}
    \begin{tabularx}{\linewidth}{@{}lXl@{}}
        \toprule
        Result & Role & Proof location \\
        \midrule
        Theorem~\ref{thm:general_dynamic} & Finite-time Gaussian-process dynamics of the linear EP module. & Section~\ref{sec:dynamic_proof} \\
        Theorem~\ref{thm:memory_defect} and Corollary~\ref{cor:no_memory} & Memory defect of the standard cavity and the no-memory special case. & Section~\ref{subsec:proof_memory_defect} \\
        Theorem~\ref{thm:corrected_se} & Oracle innovation channel and corrected coordinate-wise SE. & Section~\ref{subsec:proof_corrected_se} \\
        Theorem~\ref{thm:precision_replacement} & Connection between predictable-precision dynamics and adaptive diagonal EP. & Section~\ref{subsec:proof_precision_replacement} \\
        \bottomrule
    \end{tabularx}
\end{table}

\subsection{Fresh-AWGN Principle and Profile-Dependent State Variables}
\label{subsec:fresh_awgn_principle}

The prior module in diagonal EP applies a scalar Bayesian denoiser to the
incoming message from the linear module.  This denoising step is matched if the
message can be interpreted, asymptotically, as a scalar AWGN observation
\begin{equation}
    r_{A\to B,t,j}
    =
    x_j+\sqrt{\tau_{t,j}}\,W_{t,j},
    \qquad
    W_{t,j}\sim\mathcal N(0,1),
\end{equation}
with $W_{t,j}$ fresh relative to the past in the sense of
Definition~\ref{def:freshness}.  In unitarily invariant settings,
this fresh-AWGN mechanism is what ultimately leads to a scalar state evolution.
In the present variance-profile Gaussian setting, the linear module still
produces a Gaussian kernel, but the resulting Gaussian process need not be an
innovation process.  The purpose of this subsection is to define the
coordinate-wise state variables that describe this phenomenon.

For each iteration $t$, let
\begin{equation}
    \boldsymbol C_t
    =
    \left(
        \gamma_w\boldsymbol A^{\mathsf T}\boldsymbol A
        +
        \boldsymbol\Gamma_t
    \right)^{-1}
\end{equation}
be the linear-module covariance matrix defined in
\eqref{eq:C_t_def}.  The diagonal entries
\begin{equation}
    d_{t,j}=[\boldsymbol C_t]_{jj}
\end{equation}
will be shown to admit deterministic equivalents generated by a conditioned
MDE.  We denote these equivalents by
\begin{equation}
    T_{t,j}.
\end{equation}
Thus, in the finite-time limit considered in this paper,
\begin{equation}
    d_{t,j}
    =
    T_{t,j}
    +
    o_p^{\ell_2}(1).
\end{equation}
The precise MDE construction of $T_{t,j}$ is given in
Section~\ref{sec:dynamic_proof}.  At the level of the main results, it is
sufficient to regard $T_{t,j}$ as the MDE response of the $j$th signal
coordinate at iteration $t$.

The variance profile enters the Gaussian kernel through a conditioned
column-covariance response of the measurement residuals.  We make this object
explicit because it is the point at which the variance-profile model differs
from the rotationally invariant case.

\begin{definition}[Conditioned two-resolvent response]
\label{def:conditioned_two_resolvent_response}
Fix $r,s\le t$ and a coordinate $j$.  Let
$\boldsymbol\Xi_{j|r,s}$ denote the conditional covariance operator of the
unexposed component of the $j$th Gaussian column after the linear history and
the previously revealed Schur scalar projections have been conditioned upon.
The number
\begin{equation}
    \Theta_j^{r,s}
\end{equation}
is defined as the deterministic equivalent, produced by the conditioned MDE
and its linearized two-resolvent response, of the quadratic form
\begin{equation}
    (\boldsymbol u_r)^{\mathsf T}
    \boldsymbol\Xi_{j|r,s}
    \boldsymbol u_s .
\end{equation}
Equivalently, $\Theta_j^{r,s}$ is the output of the linearized MDE response
with the block insertion corresponding to the conditioned covariance of the
$j$th column.
\end{definition}

In the absence of history conditioning, the conditional covariance reduces to
$M^{-1}\boldsymbol S_j$ and Definition~\ref{def:conditioned_two_resolvent_response}
reduces to the intuitive profile-weighted residual quadratic form
\begin{equation}
    \frac1M
    \boldsymbol u_r^{\mathsf T}
    \boldsymbol S_j
    \boldsymbol u_s,
    \label{eq:Theta_meaning}
\end{equation}
where $\boldsymbol S_j$ is the column-profile matrix in
Definition~\ref{def:profile_matrices}, and
\begin{equation}
    \boldsymbol u_t
    =
    \boldsymbol w-\boldsymbol A\boldsymbol m_t
\end{equation}
is the measurement residual defined in \eqref{eq:u_t_def}.  In the sequel,
$\Theta_j^{r,s}$ always denotes the conditioned response in
Definition~\ref{def:conditioned_two_resolvent_response};
\eqref{eq:Theta_meaning} is only its unconditioned special case.

The covariance kernel
of the coordinate-wise Schur Gaussian process is then
\begin{equation}
    \zeta_j^{r,s}
    =
    \gamma_w^2
    T_{r,j}T_{s,j}
    \Theta_j^{r,s}.
\end{equation}

For each coordinate $j$, define the past Gaussian-history vector
\begin{equation}
    \boldsymbol{\mathcal Z}_{<t,j}
    =
    (\mathcal Z_{0,j},\mathcal Z_{1,j},\ldots,\mathcal Z_{t-1,j})^{\mathsf T},
\end{equation}
and the covariance blocks
\begin{equation}
    \boldsymbol\zeta_j^{<t,<t}
    =
    (\zeta_j^{r,s})_{0\le r,s<t},
\end{equation}
\begin{equation}
    \boldsymbol\zeta_j^{t,<t}
    =
    (\zeta_j^{t,0},\ldots,\zeta_j^{t,t-1}),
    \qquad
    \boldsymbol\zeta_j^{<t,t}
    =
    (\boldsymbol\zeta_j^{t,<t})^{\mathsf T}.
\end{equation}
The Gaussian regression of $\mathcal Z_{t,j}$ onto its past is
\begin{equation}
    \mathcal Z_{t,j}
    =
    \boldsymbol\zeta_j^{t,<t}
    (\boldsymbol\zeta_j^{<t,<t})^\dagger
    \boldsymbol{\mathcal Z}_{<t,j}
    +
    G_{t,j},
    \label{eq:Z_gaussian_regression}
\end{equation}
where $G_{t,j}$ is Gaussian and independent of
$\boldsymbol{\mathcal Z}_{<t,j}$, conditionally on the MDE-generated environment.  Its
variance is
\begin{equation}
    \nu_{t,j}
    =
    \zeta_j^{t,t}
    -
    \boldsymbol\zeta_j^{t,<t}
    (\boldsymbol\zeta_j^{<t,<t})^\dagger
    \boldsymbol\zeta_j^{<t,t}.
\end{equation}
When $t=0$, the past is empty and the convention is
\begin{equation}
    \nu_{0,j}=\zeta_j^{0,0}.
\end{equation}

Let
\begin{equation}
    \bar\pi_{t,j}
    =
    \operatorname{Proj}_{[\pi_{\min},\pi_{\max}]}
    \left(T_{t,j}^{-1}-\gamma_{t,j}\right)
\end{equation}
be the MDE-level diagonal EP cavity precision, where
$\operatorname{Proj}_{[a,b]}(x)=\min\{b,\max\{a,x\}\}$.  Define the cavity
scaling factor
\begin{equation}
    \alpha_{t,j}
    =
    \frac{1}{\bar\pi_{t,j}T_{t,j}}.
\end{equation}
The predictable memory component and the innovation variance of the EP cavity
are then
\begin{align}
    \mu_{t,j}
    &=
    \alpha_{t,j}
    \boldsymbol\zeta_j^{t,<t}
    (\boldsymbol\zeta_j^{<t,<t})^\dagger
    \boldsymbol{\mathcal Z}_{<t,j},
    \label{eq:mu_def}
    \\
    \tau_{t,j}
    &=
    \alpha_{t,j}^2\nu_{t,j}.
    \label{eq:tau_def}
\end{align}
The collection
\begin{equation}
    \mathcal P_t
    =
    \{T_{s,j},\Theta_j^{r,s},\zeta_j^{r,s},
      \bar\pi_{s,j},\tau_{s,j}:0\le r,s\le t,\ 1\le j\le N\}
\end{equation}
will be called the MDE-generated environment up to time $t$.

The central point is that $\mu_{t,j}$ is generally nonzero.  In the terminology
of Definition~\ref{def:freshness}, it is the EP-cavity memory defect.  Hence the
standard diagonal EP cavity is not merely corrupted by Gaussian noise; it is
shifted by a component predictable from the past Gaussian history.  The next
theorem states the finite-time dynamic result from which all main conclusions
follow.

\subsection{General Finite-Time Dynamic Theorem}
\label{subsec:general_dynamic_theorem}

The following theorem is the technical backbone of the paper.  It is the
variance-profile counterpart of a finite-time EP state-evolution theorem.
Rather than producing two scalar variance recursions, it produces a
coordinate-wise Gaussian history governed by the MDE-generated environment
defined above.  All variables in this theorem belong to the
predictable-precision dynamics used in the conditioning proof.  The actual
finite-dimensional adaptive precision recursion is connected to this dynamics
only later, through Theorem~\ref{thm:precision_replacement}.

\begin{theorem}[General finite-time dynamic theorem]
\label{thm:general_dynamic}
Suppose Assumptions~\ref{ass:signal}--\ref{ass:ep_regular} hold.  Fix an
arbitrary finite iteration horizon $T$.  Then, for every $t\le T$, the
predictable-precision diagonal EP dynamics associated with
Section~\ref{sec:system_model} satisfies the following properties.

\emph{1) MDE response of the linear covariance.}  The retained diagonal
variance satisfies
\begin{equation}
    \frac1N\sum_{j=1}^N
    |d_{t,j}-T_{t,j}|^2
    \overset{p}{\longrightarrow}0.
\end{equation}

\emph{2) Coordinate-wise Gaussian kernel.}  There exist Schur Gaussian process coordinates
$\{\mathcal Z_{s,j}:0\le s\le t,
1\le j\le N\}$ such that
\begin{equation}
    m_{t,j}
    =
    \gamma_{t,j}T_{t,j}q_{t,j}
    +
    \mathcal Z_{t,j}
    +
    \Delta_{t,j},
    \label{eq:main_linear_kernel}
\end{equation}
where
\begin{equation}
    \frac1N\sum_{j=1}^N|\Delta_{t,j}|^2
    \overset{p}{\longrightarrow}0.
\end{equation}
Conditionally on $\mathcal P_t$, the Gaussian history
\begin{equation}
    (\mathcal Z_{0,j},\mathcal Z_{1,j},\ldots,\mathcal Z_{t,j})
\end{equation}
has covariance kernel
\begin{equation}
    \operatorname{Cov}
    (\mathcal Z_{r,j}^{\mathrm G},\mathcal Z_{s,j}^{\mathrm G}\mid\mathcal P_t)
    =
    \zeta_j^{r,s},
    \qquad
    0\le r,s\le t.
\end{equation}

\emph{3) Empirical Gaussian law.}  Let $\mathcal X_{j,N}^t$ denote any finite
coordinate history formed from the variables generated by the EP recursion up
to time $t$, and let $\mathcal X_{j,\mathrm G}^t$ denote the corresponding
Gaussian-kernel reference history obtained by replacing the linear-module
residuals with the Gaussian Schur variables in \eqref{eq:main_linear_kernel}.  Then,
for every $\psi\in \PL(2)$,
\begin{equation}
    \frac1N\sum_{j=1}^N
    \psi(\mathcal X_{j,N}^t)
    -
    \frac1N\sum_{j=1}^N
    \mathbb E_{\mathrm G,j}
    \!\left[
        \psi(\mathcal X_{j,\mathrm G}^t)
        \mid
        \mathcal P_t
    \right]
    \overset{p}{\longrightarrow}0.
    \label{eq:empirical_gaussian_law}
\end{equation}

\emph{4) EP cavity decomposition.}  The standard diagonal EP cavity error
satisfies
\begin{equation}
    h_{t,j}
    =
    \mu_{t,j}
    +
    \sqrt{\tau_{t,j}}\,W_{t,j}
    +
    \varepsilon_{t,j},
\end{equation}
where $W_{t,j}\sim\mathcal N(0,1)$ is independent of the past Gaussian history
conditionally on $\mathcal P_t$, and
\begin{equation}
    \frac1N\sum_{j=1}^N|\varepsilon_{t,j}|^2
    \overset{p}{\longrightarrow}0.
\end{equation}

\emph{5) Regularity closure.}  The vectors
\begin{equation}
    \boldsymbol m_t,\quad
    \boldsymbol h_t,\quad
    \boldsymbol p_t,\quad
    \boldsymbol q_{t+1}
\end{equation}
remain signal-side admissible, and the measurement residual
\begin{equation}
    \boldsymbol u_t
    =
    \boldsymbol w-\boldsymbol A\boldsymbol m_t
\end{equation}
remains measurement-side admissible.  Hence the finite-time induction can be
continued up to iteration $T$.
\end{theorem}

The proof of Theorem~\ref{thm:general_dynamic} is given in
Section~\ref{sec:dynamic_proof}.  Its structure follows the error recursion of
Section~\ref{subsec:error_recursion}: the past iterates impose finite linear
constraints on $\boldsymbol A$; Gaussian conditioning decomposes
$\boldsymbol A$ into a bounded deformation and a centered correlated Gaussian
residual; the external MDE input gives the linear response $T_{t,j}$; and a
conditioned Schur-kernel and empirical Gaussian-law argument upgrades the
coordinate-wise Gaussian kernel to the empirical law
\eqref{eq:empirical_gaussian_law}.  The theorem deliberately avoids
conditioning on actual retained diagonal resolvents or adaptive cavity
precisions.  Those nonlinear finite-sample quantities are handled separately by
the precision-replacement theorem.

\subsection{Consequences: Memory Defect and Corrected State Evolution}
\label{subsec:memory_defect_corrected_se}

Theorem~\ref{thm:general_dynamic} has two immediate consequences.  The first is
a negative statement for the usual fresh-AWGN interpretation of standard
diagonal EP: the cavity is Gaussian, but it is generally shifted by a
predictable memory term.  The second is a positive statement: after removing
this predictable component, one obtains a fresh coordinate-wise AWGN channel.

\begin{theorem}[Memory defect of standard diagonal EP]
\label{thm:memory_defect}
Under the assumptions of Theorem~\ref{thm:general_dynamic}, the standard
diagonal EP message from module A to module B satisfies
\begin{equation}
    r_{A\to B,t,j}
    =
    x_j
    +
    \mu_{t,j}
    +
    \sqrt{\tau_{t,j}}\,W_{t,j}
    +
    o_p^{\ell_2}(1).
    \label{eq:std_channel_memory}
\end{equation}
Equivalently,
\begin{equation}
    h_{t,j}
    =
    \mu_{t,j}
    +
    \sqrt{\tau_{t,j}}\,W_{t,j}
    +
    o_p^{\ell_2}(1).
\end{equation}
Thus the standard EP cavity is not, in general, a fresh AWGN observation of
$x_j$.  It is a shifted Gaussian observation whose shift is predictable from
the past Gaussian history.
\end{theorem}

The proof of Theorem~\ref{thm:memory_defect} is given in
Section~\ref{subsec:proof_memory_defect}.

The term $\mu_{t,j}$ is the memory defect.  It is not a small technical
remainder; it is the conditional mean of the current Gaussian residual given
its past.  Therefore, unless this conditional mean vanishes, the scalar
posterior mean
\begin{equation}
    \eta(r_{A\to B,t,j};\pi_{t,j})
\end{equation}
is not matched to the fresh channel
\begin{equation}
    X+\sqrt{\tau_{t,j}}\,W.
\end{equation}

\begin{corollary}[No-memory special case]
\label{cor:no_memory}
If, for every fixed $t$ and all coordinates $j$,
\begin{equation}
    \boldsymbol\zeta_j^{t,<t}
    =
    \boldsymbol 0,
    \label{eq:no_memory_condition}
\end{equation}
then
\begin{equation}
    \mu_{t,j}=0,
\end{equation}
and the standard diagonal EP cavity is asymptotically fresh:
\begin{equation}
    h_{t,j}
    =
    \sqrt{\tau_{t,j}}\,W_{t,j}
    +
    o_p^{\ell_2}(1).
    \label{eq:no_memory_fresh}
\end{equation}
\end{corollary}

The corollary follows immediately from Theorem~\ref{thm:memory_defect} and the
Gaussian-regression formula \eqref{eq:Z_gaussian_regression}.

Condition~\eqref{eq:no_memory_condition} is the profile-dependent analogue of
the innovation property available in rotationally invariant settings.  For a
general variance profile, it need not hold.

\begin{remark}[Sanity check: collapse to the innovation case]
\label{rem:sanity_innovation_collapse}
If the conditioned MDE response is temporally diagonal in the sense of
\eqref{eq:no_memory_condition}, then the Gaussian regression in
\eqref{eq:Z_gaussian_regression} has zero predictable part.  Consequently the
standard and corrected cavities coincide up to the vanishing empirical
remainder, and the coordinate-wise SE reduces to the usual fresh-AWGN form.
This is the precise sense in which the present theorem is consistent with the
innovation mechanism proved in rotationally invariant EP analyses.  The point
of the variance-profile result is that temporal diagonality of the conditioned
kernel is no longer automatic.
\end{remark}

The correction below is an oracle state-evolution construction.  It is
introduced to identify the innovation channel implied by the Gaussian regression
of the Schur residual process.  It is not claimed here that the memory
coefficients can be estimated from one finite-dimensional observation or that
\eqref{eq:corrected_cavity_def} is a directly implementable finite-sample EP
update.

We now define the memory-corrected cavity by subtracting the predictable
Gaussian-history component:
\begin{equation}
    \widetilde h_{t,j}
    =
    h_{t,j}-\mu_{t,j},
    \qquad
    \widetilde r_{A\to B,t,j}
    =
    x_j+\widetilde h_{t,j}.
    \label{eq:corrected_cavity_def}
\end{equation}
This correction is a Gaussian innovation decomposition at the level of state
evolution.  It removes the conditional mean of the Gaussian residual and leaves
the fresh innovation.

\begin{theorem}[Memory-corrected coordinate-wise state evolution]
\label{thm:corrected_se}
Under the assumptions of Theorem~\ref{thm:general_dynamic}, the corrected
cavity satisfies
\begin{equation}
    \widetilde h_{t,j}
    =
    \sqrt{\tau_{t,j}}\,W_{t,j}
    +
    o_p^{\ell_2}(1).
\end{equation}
Equivalently,
\begin{equation}
    \widetilde r_{A\to B,t,j}
    =
    x_j
    +
    \sqrt{\tau_{t,j}}\,W_{t,j}
    +
    o_p^{\ell_2}(1).
\end{equation}

Let $\mathcal X_{j,N}^{t,\mathrm{corr}}$ denote any finite coordinate history
generated by the corrected scalar channel up to time $t$, and let
$\mathcal X_{j,\mathrm{SE}}^{t,\mathrm{corr}}$ denote its state-evolution
reference generated by
\begin{equation}
    R_{t,j}^{\mathrm{SE}}
    =
    X_j+\sqrt{\tau_{t,j}}\,W_{t,j},
\end{equation}
with posterior updates applied componentwise using the scalar prior $P_X$.
Then, for every $\psi\in \PL(2)$,
\begin{equation}
    \frac1N\sum_{j=1}^N
    \psi(\mathcal X_{j,N}^{t,\mathrm{corr}})
    -
    \frac1N\sum_{j=1}^N
    \mathbb E_{\mathrm{SE},j}
    \!\left[
        \psi(\mathcal X_{j,\mathrm{SE}}^{t,\mathrm{corr}})
        \mid
        \mathcal P_t
    \right]
    \overset{p}{\longrightarrow}0.
    \label{eq:corrected_empirical_se}
\end{equation}
\end{theorem}

The proof of Theorem~\ref{thm:corrected_se} is given in
Section~\ref{subsec:proof_corrected_se}.

The state evolution in Theorem~\ref{thm:corrected_se} is coordinate-wise and
quenched with respect to the MDE-generated environment $\mathcal P_t$.  It is
not, in general, reducible to a pair of scalar variance recursions.  The
reduction to a scalar SE requires additional symmetry of the variance profile.

\begin{remark}[Meaning of the correction]
The correction in \eqref{eq:corrected_cavity_def} should be understood as an
SE-level Gaussian regression correction.  It identifies the fresh scalar
channel obtained after removing the predictable part of the Gaussian residual.
Without further structural assumptions, it is not claimed that all memory
coefficients can be estimated from a single finite-dimensional realization.
Designing practical low-complexity estimators or approximations of this oracle
correction is left for future work.
\end{remark}

\subsection{Precision Replacement for Adaptive Diagonal EP}
\label{subsec:precision_replacement}

The dynamic theorem above is stated in the form used by the conditioning proof,
where the diagonal precision sequence is represented by its predictable
MDE-level counterpart.  The diagonal EP recursion in
Section~\ref{sec:system_model}, however, computes the precisions from the
finite-dimensional retained variances and scalar posterior variances.  The
following theorem connects these two descriptions.

For a vector $\boldsymbol a\in\mathbb R^N$, write
\begin{equation}
    \|\boldsymbol a\|_N^2
    =
    \frac1N\|\boldsymbol a\|^2.
\end{equation}

\begin{theorem}[Branch-wise precision replacement]
\label{thm:precision_replacement}
Fix a finite horizon $T$ and assume the additional replacement regularity
conditions stated in Section~\ref{subsec:replacement_regular_conditions}.  Consider
either of the following two branches:
\begin{equation}
    b\in\{\mathrm{std},\mathrm{corr}\},
\end{equation}
where $\mathrm{std}$ denotes the standard diagonal EP branch using
$\boldsymbol h_t$, and $\mathrm{corr}$ denotes the memory-corrected branch using
$\widetilde{\boldsymbol h}_t$.

Let
\begin{equation}
    (\boldsymbol q_t^{b,\mathrm{act}},
     \boldsymbol\gamma_t^{b,\mathrm{act}})
\end{equation}
be the finite-dimensional adaptive recursion generated by the diagonal EP
updates of Section~\ref{sec:system_model}, and let
\begin{equation}
    (\boldsymbol q_t^{b,\mathrm{orc}},
     \boldsymbol\gamma_t^{b,\mathrm{orc}})
\end{equation}
be the corresponding predictable recursion described by the MDE-generated
state variables in Theorem~\ref{thm:general_dynamic}.  If the two recursions
are initialized identically, then, for every fixed $t\le T$,
\begin{align}
    \left\|
        \boldsymbol\gamma_t^{b,\mathrm{act}}
        -
        \boldsymbol\gamma_t^{b,\mathrm{orc}}
    \right\|_N
    &\overset{p}{\longrightarrow}0,
    \label{eq:gamma_replacement}
    \\
    \left\|
        \boldsymbol q_t^{b,\mathrm{act}}
        -
        \boldsymbol q_t^{b,\mathrm{orc}}
    \right\|_N
    &\overset{p}{\longrightarrow}0.
    \label{eq:q_replacement}
\end{align}
\end{theorem}

The proof of Theorem~\ref{thm:precision_replacement} is given in
Section~\ref{subsec:proof_precision_replacement}.

Theorem~\ref{thm:precision_replacement} is branch-wise.  It does not state that
the standard and corrected branches are asymptotically equivalent to each
other.  Rather, it states that within each branch, the finite-dimensional
adaptive precision update can be replaced, in empirical norm, by its
MDE-predictable counterpart.  Thus the predictable-precision dynamic theorem
applies to the adaptive diagonal EP recursion.

\section{Proof of the General Finite-Time Dynamic Theorem}
\label{sec:dynamic_proof}
\label{sec:proof_general}

This section proves Theorem~\ref{thm:general_dynamic}.  The proof follows the
same high-level strategy as rigorous finite-time EP analyses for rotationally
invariant matrices, but the conditioning mechanism is different.  In the
unitarily invariant setting, the past iterates leave a fresh Haar component in
the orthogonal complement of the history.  In the present variance-profile
Gaussian setting, the past iterates impose finite linear observations on
$\boldsymbol A$.  Conditioning therefore leaves a centered correlated Gaussian
matrix rather than a Haar matrix.  The role of rotational averaging is replaced
by the regularized correlated-Gaussian MDE input, a Schur-complement
representation under the actual conditioned history, and an empirical Gaussian
law for the resulting column-wise Gaussian process.

Throughout this section, the diagonal precision sequence is considered in the
predictable form needed for the conditioning proof.  More precisely, at the
beginning of the $t$th linear step, the diagonal loading
\begin{equation}
\boldsymbol\Gamma_t
=
\operatorname{diag}(\gamma_{t,1},\ldots,\gamma_{t,N})
\end{equation}
is regarded as part of the regular finite-time environment and satisfies the
compactness condition in Assumption~\ref{ass:ep_regular}.  The finite-dimensional
adaptive update of the diagonal EP precisions is connected to this
predictable-precision form in Section~\ref{sec:precision_replacement}.  The
present section proves the Gaussian dynamic theorem under the predictable
loading, which is the form in which conditioning on the past leaves only finite
linear information about $\boldsymbol A$.

\subsection{Proof Strategy and Induction Hypotheses}
\label{subsec:proof_strategy_induction}

Let
\begin{equation}
\boldsymbol Q_t
=
(\boldsymbol q_0,\ldots,\boldsymbol q_{t-1}),
\quad
\boldsymbol M_t
=
(\boldsymbol m_0,\ldots,\boldsymbol m_{t-1}),
\end{equation}
\begin{equation}
\boldsymbol H_t
=
(\boldsymbol h_0,\ldots,\boldsymbol h_{t-1}),
\quad
\boldsymbol P_t
=
(\boldsymbol p_0,\ldots,\boldsymbol p_{t-1}),
\end{equation}
and
\begin{equation}
\boldsymbol U_t
=
(\boldsymbol u_0,\ldots,\boldsymbol u_{t-1})
\end{equation}
be the finite histories up to, but not including, iteration $t$.  These
histories are used for Gaussian conditioning only through the finite linear
constraints displayed below.  The scalar-module histories $\boldsymbol H_t$ and
$\boldsymbol P_t$ are tracked in the induction regularity event, but they are
not themselves added to the Gaussian conditioning field.  This convention
prevents the conditioning argument from accidentally conditioning on nonlinear
resolvent information about $\boldsymbol A$.

The induction is carried out on a regularity event, denoted by
$\mathcal R_t$.  Informally, $\mathcal R_t$ is the event that the finite
history up to time $t$ is admissible in the sense of
Definition~\ref{def:admissible}, the diagonal precisions remain in compact
positive intervals, and the finite normalized Gram matrices associated with the
history are tight after removing asymptotically redundant directions.  The
Moore--Penrose inverse is used whenever a finite history contains redundant
directions.  This convention avoids unnecessary full-rank assumptions and is
consistent with the Gaussian regression formula in
Lemma~\ref{lem:gaussian_regression}.

The induction has three transitions.  The linear transition identifies the
conditional law of the LMMSE module under the predictable linear history.  The
cavity transition converts the Schur residual produced by the linear module into
a Gaussian process with a regression memory term.  The scalar transition passes
this effective channel through the separable prior module and verifies that the
new histories remain regular.  Thus the proof is not a collection of separate
limits; each proposition below answers one of the questions required to advance
from iteration $t$ to iteration $t+1$.

For every fixed horizon $T$, the induction establishes
\begin{equation}
    \mathcal R_t
    \Longrightarrow
    \text{all conclusions of Theorem~\ref{thm:general_dynamic} at time }t,
\end{equation}
and then proves the closure implication
\begin{equation}
    \mathcal R_t
    \Longrightarrow
    \mathcal R_{t+1}.
\end{equation}
Since the initialization is deterministic and regular, $\mathcal R_0$ holds
with probability tending to one.  Finite induction then proves the theorem.

The remainder of the section follows this order.  First, the exact history
constraints are used to condition the measurement matrix and to obtain the MDE
responses $T_{t,j}$ and $\Theta_j^{r,s}$.  Second, a current-step Schur complement
separates the instantaneous response from a column-wise Gaussian process, and an
empirical Gaussian law upgrades the coordinate description to pseudo-Lipschitz
convergence.  Third, Gaussian regression extracts the fresh innovation seen by
the corrected cavity, and the scalar module closes the regularity event.  The
proofs of the technical propositions are deferred to the appendices; this
section records their roles and assembles them into the induction.

\subsection{Conditioned Linear Module}
\label{subsec:conditioned_linear_module}

We start from the exact history constraints derived in
Section~\ref{subsec:error_recursion}.  By the definition
\begin{equation}
\boldsymbol u_t
=
\boldsymbol w-\boldsymbol A\boldsymbol m_t,
\end{equation}
the past linear-module errors satisfy
\begin{equation}
	\boldsymbol A\boldsymbol M_t
	=
	\boldsymbol w\boldsymbol 1_t^{\mathsf T}
	-
	\boldsymbol U_t .
	\label{eq:proof_history_AM}
\end{equation}
Furthermore, the normal equation of the linear module gives
\begin{equation}
\gamma_w
\boldsymbol A^{\mathsf T}\boldsymbol u_s
=
\boldsymbol\Gamma_s(\boldsymbol m_s-\boldsymbol q_s),
\qquad 0\le s<t.
\end{equation}
Stacking over $s<t$ yields
\begin{equation}
	\boldsymbol A^{\mathsf T}\boldsymbol U_t
	=
	\gamma_w^{-1}
	\left(
	\boldsymbol\Gamma_0(\boldsymbol m_0-\boldsymbol q_0),
	\ldots,
	\boldsymbol\Gamma_{t-1}(\boldsymbol m_{t-1}-\boldsymbol q_{t-1})
	\right).
	\label{eq:proof_history_ATU}
\end{equation}
Equations \eqref{eq:proof_history_AM} and \eqref{eq:proof_history_ATU} are
finite linear observations of the variance-profile Gaussian matrix
$\boldsymbol A$.  Equivalently, after vectorization, there exist a finite-rank
linear operator $\boldsymbol L_t$ and a vector $\boldsymbol b_t$ such that
\begin{equation}
    \boldsymbol L_t\operatorname{vec}(\boldsymbol A)
    =
    \boldsymbol b_t.
    \label{eq:proof_linear_history_operator}
\end{equation}
The predictable linear-history filtration used below is
\begin{equation}
    \mathcal F_t^{\rm lin}
    =
    \sigma\!\left(
    \boldsymbol x,\boldsymbol w,
    \boldsymbol\Gamma_0,\ldots,\boldsymbol\Gamma_t,
    \boldsymbol L_t,
    \boldsymbol b_t
    \right).
    \label{eq:proof_filtration}
\end{equation}
The notation $\mathcal F_t$ will be used below as shorthand for
$\mathcal F_t^{\rm lin}$.  This filtration contains the external variables,
the predictable diagonal loadings, and the finite linear observations
\eqref{eq:proof_linear_history_operator}.  It does not contain the actual
retained variance $\boldsymbol D_s$, the actual finite-sample cavity precision
$\boldsymbol\Pi_s$, or any other nonlinear resolvent functional of
$\boldsymbol A$.  The MDE-level precision variables used by the predictable
recursion are deterministic functions of the reference environment and are not
additional observations of the Gaussian matrix.  The adaptive precisions
computed from the finite-dimensional diagonals are connected to this
predictable history only in the precision-replacement argument of
Section~\ref{sec:precision_replacement}.  Consequently, for the conditional
law of $\boldsymbol A$, conditioning on $\mathcal F_t$ is precisely
conditioning on finite linear observations and predictable quantities.  The
matrix therefore remains Gaussian after conditioning, although its covariance is
no longer entrywise independent in general.

The first question in the linear transition is therefore: what is the law of
$\boldsymbol A$ after these linear observations have been fixed?  The answer is
the Gaussian analogue of the conditional Haar representation used in
rotationally invariant analyses.  The matrix remains Gaussian, but its residual
covariance becomes history-dependent.

\begin{proposition}[Conditional Gaussian representation]
	\label{prop:conditional_gaussian_history}
	Assume that the regularity event $\mathcal R_t$ holds.  Conditionally on
	$\mathcal F_t$, the measurement matrix admits the decomposition
	\begin{equation}
		\boldsymbol A
		=
		\boldsymbol A_{\parallel,t}
		+
		\boldsymbol A_{\perp,t},
	\end{equation}
	where
	\begin{equation}
	\boldsymbol A_{\parallel,t}
	=
	\mathbb E[\boldsymbol A\mid\mathcal F_t]
	\end{equation}
	and $\boldsymbol A_{\perp,t}$ is a centered correlated Gaussian matrix:
	\begin{equation}
	\mathbb E[\boldsymbol A_{\perp,t}\mid\mathcal F_t]=\boldsymbol 0.
	\end{equation}
	Moreover, for all deterministic vectors
	$\boldsymbol p\in\mathbb R^M$ and $\boldsymbol q\in\mathbb R^N$,
	\begin{equation}
		\operatorname{Var}\!\left(
		\boldsymbol p^{\mathsf T}
		\boldsymbol A_{\perp,t}
		\boldsymbol q
		\,\middle|\,
		\mathcal F_t
		\right)
		\le
		\frac{C}{M}
		\|\boldsymbol p\|^2\|\boldsymbol q\|^2,
	\end{equation}
	where $C$ is independent of $M,N,t$.
\end{proposition}

\begin{IEEEproof}
	See Appendix~\ref{app:gaussian_conditioning_histories}.
\end{IEEEproof}

The conditional mean $\boldsymbol A_{\parallel,t}$ is a finite-rank
deformation determined by the history.  The next proposition states that this
deformation remains bounded on the regular history event and that the resulting
conditioned linearization falls within the class covered by the external MDE
input.

\begin{proposition}[Bounded deformation and MDE response]
	\label{prop:bounded_deformation_mde}
	Assume $\mathcal R_t$.  Then
	\begin{equation}
		\|\boldsymbol A_{\parallel,t}\|=O_p(1).
	\end{equation}
	Consequently, the block linearization
	\begin{equation}
		\boldsymbol K_t
		=
		\begin{pmatrix}
			\boldsymbol I_M
			&
			\sqrt{\gamma_w}\boldsymbol A
			\\
			\sqrt{\gamma_w}\boldsymbol A^{\mathsf T}
			&
			-\boldsymbol\Gamma_t
		\end{pmatrix}
	\end{equation}
	can be written as
	\begin{equation}
	\boldsymbol K_t
	=
	\boldsymbol D_t^{\mathrm{lin}}
	+
	\boldsymbol W_t,
	\end{equation}
	where
	\begin{equation}
	\boldsymbol D_t^{\mathrm{lin}}
	=
	\begin{pmatrix}
		\boldsymbol I_M
		&
		\sqrt{\gamma_w}\boldsymbol A_{\parallel,t}
		\\
		\sqrt{\gamma_w}\boldsymbol A_{\parallel,t}^{\mathsf T}
		&
		-\boldsymbol\Gamma_t
	\end{pmatrix}
	\end{equation}
	is a bounded deterministic deformation conditionally on $\mathcal F_t$, and
	\begin{equation}
	\boldsymbol W_t
	=
	\begin{pmatrix}
		\boldsymbol 0
		&
		\sqrt{\gamma_w}\boldsymbol A_{\perp,t}
		\\
		\sqrt{\gamma_w}\boldsymbol A_{\perp,t}^{\mathsf T}
		&
		\boldsymbol 0
	\end{pmatrix}
	\end{equation}
	is a centered correlated Gaussian matrix satisfying the flat covariance
	condition of Theorem~\ref{thm:external_mde}.  Hence the external MDE input
	applies to $\boldsymbol K_t$.  In particular, if
	\begin{equation}
	\boldsymbol G_t=\boldsymbol K_t^{-1}
	\end{equation}
	and $\boldsymbol M_t^{\mathrm{MDE}}$ denotes the corresponding MDE solution,
	then the lower-right block of $\boldsymbol G_t$ gives the deterministic
	equivalent of the linear-module covariance.  Equivalently,
	\begin{equation}
		\frac1N\sum_{j=1}^N
		|d_{t,j}-T_{t,j}|^2
		\overset{p}{\longrightarrow}0.
	\end{equation}
\end{proposition}

\begin{IEEEproof}
	See Appendix~\ref{app:bounded_deformation_mde}.
\end{IEEEproof}

Proposition~\ref{prop:bounded_deformation_mde} proves the first part of
Theorem~\ref{thm:general_dynamic}: the retained diagonal variance of the linear
Gaussian module is controlled by the MDE response $T_{t,j}$.  The next step is
to identify the Gaussian kernel produced by the linear estimator.

\subsection{Gaussian Kernel and Empirical Upgrade}
\label{subsec:gaussian_kernel_empirical_upgrade}

The covariance of the coordinate-wise Gaussian kernel depends on
profile-weighted residual quadratic forms.  These are not scalar traces; they
retain the column profile $\boldsymbol S_j$.  The following proposition is the
two-resolvent consequence of the MDE input.

\begin{proposition}[Conditioned two-resolvent response]
	\label{prop:two_resolvent_response}
	Assume $\mathcal R_t$.  For any fixed $r,s\le t$, let
	$\boldsymbol\Xi_{j|r,s}$ denote the conditional covariance operator of the
	unexposed part of the $j$th column appearing in the Schur residual pair
	$(\mathcal Z_{r,j},\mathcal Z_{s,j})$.  Then
	\begin{equation}
		(\boldsymbol u_r)^{\mathsf T}
		\boldsymbol\Xi_{j|r,s}
		\boldsymbol u_s
		=
		\Theta_j^{r,s}
		+
		o_p(1)
	\end{equation}
	in the admissible quadratic-form sense, uniformly over fixed $r,s$ and in
	empirical average over $j$.  In the unconditioned special case
	$\boldsymbol\Xi_{j|r,s}=M^{-1}\boldsymbol S_j$, this reduces to
	\begin{equation}
		\frac1M
		\boldsymbol u_r^{\mathsf T}\boldsymbol S_j\boldsymbol u_s
		=
		\Theta_j^{r,s}+o_p(1).
	\end{equation}
\end{proposition}

\begin{IEEEproof}
	See Appendix~\ref{app:mde_responses}.
\end{IEEEproof}

We now state the Schur kernel obtained under the actual conditioned history.  Let
$\boldsymbol a_j$ denote the $j$th column of $\boldsymbol A$, and let
$\bar{\boldsymbol u}_{t}^{(j)}$ be the one-step column-cavity residual formed
from the current linear system after removing only the $j$th coordinate, with
the actual incoming message and diagonal loading at time $t$ kept fixed.

\begin{proposition}[Schur kernel under the conditioned history]
\label{prop:schur_kernel}
Assume $\mathcal R_t$.  Then the linear-module error satisfies
\begin{equation}
    m_{t,j}
    =
    \gamma_{t,j}T_{t,j} q_{t,j}
    +
    \mathcal Z_{t,j}
    +
    \Delta_{t,j},
\end{equation}
where
\begin{equation}
    \mathcal Z_{t,j}
    =
    \gamma_w T_{t,j}
    \boldsymbol a_j^{\mathsf T}
    \bar{\boldsymbol u}_{t}^{(j)}
\end{equation}
is the history-conditioned Schur residual process coordinate, and
\begin{equation}
    \frac1N\sum_{j=1}^N|\Delta_{t,j}|^2
    \overset{p}{\longrightarrow}0.
\end{equation}
\end{proposition}

\begin{IEEEproof}
    See Appendix~\ref{app:schur_conditioned_history}.
\end{IEEEproof}

\begin{proposition}[Covariance kernel and empirical Gaussian law]
	\label{prop:covariance_empirical_law}
	Assume $\mathcal R_t$.  For every fixed $0\le r,s\le t$,
	\begin{equation}
		\operatorname{Cov}
		(\mathcal Z_{r,j}^{\mathrm G},\mathcal Z_{s,j}^{\mathrm G}\mid\mathcal P_t)
		=
		\zeta_j^{r,s},
	\end{equation}
	where
	\begin{equation}
	\zeta_j^{r,s}
	=
	\gamma_w^2T_{r,j}T_{s,j}\Theta_j^{r,s}.
	\end{equation}
	Moreover, for any finite coordinate history $\mathcal X_{j,N}^t$ generated by
	the EP recursion up to time $t$, and for the corresponding Gaussian-kernel
	reference history $\mathcal X_{j,\mathrm G}^t$, we have, for every
	$\psi\in \PL(2)$,
	\begin{equation}
		\frac1N\sum_{j=1}^N
		\psi(\mathcal X_{j,N}^t)
		-
		\frac1N\sum_{j=1}^N
		\mathbb E_{\mathrm G,j}
		\left[
		\psi(\mathcal X_{j,\mathrm G}^t)
		\mid
		\mathcal P_t
		\right]
		\overset{p}{\longrightarrow}0.
	\end{equation}
\end{proposition}

\begin{IEEEproof}
	See Appendices~\ref{app:covariance_kernel} and
	\ref{app:empirical_gaussian_law}.
\end{IEEEproof}

Propositions~\ref{prop:schur_kernel} and
\ref{prop:covariance_empirical_law} prove that the linear module produces a
Gaussian process with covariance kernel $\zeta_j^{r,s}$.  The important point is
that this process is generally not an innovation process: the covariance
between $\mathcal Z_{t,j}$ and its past $\mathcal Z_{<t,j}$ need not vanish.
The next subsection shows how this temporal correlation appears in the EP
cavity.

\subsection{Cavity Decomposition and Induction Closure}
\label{subsec:cavity_decomposition_closure}

The EP cavity is obtained by subtracting the incoming Gaussian natural
parameters from the diagonal marginal belief:
\begin{equation}
\boldsymbol\Pi_t\boldsymbol h_t
=
\boldsymbol D_t^{-1}\boldsymbol m_t
-
\boldsymbol\Gamma_t\boldsymbol q_t .
\end{equation}
Substituting the Gaussian Schur kernel
\begin{equation}
    m_{t,j}
    =
    \gamma_{t,j}T_{t,j} q_{t,j}
    +
    \mathcal Z_{t,j}
    +
    o_p^{\ell_2}(1)
\end{equation}
and the MDE response
\begin{equation}
    d_{t,j}=T_{t,j}+o_p^{\ell_2}(1)
\end{equation}
shows the instantaneous cancellation:
\begin{equation}
    d_{t,j}^{-1}m_{t,j}-\gamma_{t,j}q_{t,j}
    =
    (T_{t,j})^{-1}\mathcal Z_{t,j}
    +
    o_p^{\ell_2}(1).
\end{equation}
Since the projected cavity precision satisfies
\begin{equation}
    \pi_{t,j}
    =
    \bar\pi_{t,j}
    +
    o_p^{\ell_2}(1),
\end{equation}
the cavity takes the form
\begin{equation}
    h_{t,j}
    =
    \alpha_{t,j}\mathcal Z_{t,j}
    +
    o_p^{\ell_2}(1),
    \qquad
    \alpha_{t,j}
    =
    \frac{1}{\bar\pi_{t,j}T_{t,j}}.
\end{equation}
Thus diagonal EP removes the instantaneous response
$\gamma_{t,j}T_{t,j}q_{t,j}$, but it does not automatically remove the
predictable component of $\mathcal Z_{t,j}$ with respect to its own history.

At this point no additional random-matrix input is needed.  The residual is
already a Gaussian process; the following proposition only applies finite-
dimensional Gaussian regression to separate its predictable part from its fresh
innovation.

\begin{proposition}[Cavity memory decomposition]
	\label{prop:cavity_memory_decomposition}
	Assume $\mathcal R_t$.  Then
	\begin{equation}
		h_{t,j}
		=
		\mu_{t,j}
		+
		\sqrt{\tau_{t,j}}\,W_{t,j}
		+
		\varepsilon_{t,j},
	\end{equation}
	where, conditionally on $\mathcal P_t$, $W_{t,j}\sim\mathcal N(0,1)$ is
	independent of the past Gaussian history, and
	\begin{equation}
		\frac1N\sum_{j=1}^N|\varepsilon_{t,j}|^2
		\overset{p}{\longrightarrow}0.
	\end{equation}
	The memory and innovation variance are exactly the quantities defined in
	\eqref{eq:mu_def} and \eqref{eq:tau_def}.
\end{proposition}

\begin{IEEEproof}
	See Appendix~\ref{app:cavity_memory_decomposition}.
\end{IEEEproof}

It remains to close the induction.  The prior module maps the scalar cavity
variables componentwise through the posterior mean and variance functions
defined in \eqref{eq:scalar_eta_def}--\eqref{eq:scalar_vB_def}.  Since these
maps preserve $\PL(2)$ empirical convergence on compact precision intervals by
Assumption~\ref{ass:ep_regular}, the signal-side histories remain admissible.
The next proposition also controls the measurement residual needed for the next
conditioning step.

\begin{proposition}[Regularity closure]
	\label{prop:regularity_closure}
	Assume $\mathcal R_t$ and the conclusions of
	Propositions~\ref{prop:schur_kernel}--\ref{prop:cavity_memory_decomposition}
	at time $t$.  Then
	\begin{equation}
	\boldsymbol m_t,\quad
	\boldsymbol h_t,\quad
	\boldsymbol p_t,\quad
	\boldsymbol q_{t+1}
	\end{equation}
	are signal-side admissible.  Moreover,
	\begin{equation}
	\boldsymbol u_t
	=
	\boldsymbol w-\boldsymbol A\boldsymbol m_t
	\end{equation}
	is measurement-side admissible.  Hence $\mathcal R_{t+1}$ holds with
	probability tending to one.
\end{proposition}

\begin{IEEEproof}
	See Appendix~\ref{app:regularity_closure}.
\end{IEEEproof}

We are now ready to prove Theorem~\ref{thm:general_dynamic}.

\begin{IEEEproof}[Proof of Theorem~\ref{thm:general_dynamic}]
	The proof is by induction over $t$.  At $t=0$, the history matrices are empty.
	The initialization of the diagonal EP recursion is deterministic and satisfies
	the compact precision bounds.  Assumptions~\ref{ass:signal} and
	\ref{ass:noise} give bounded empirical energy and no-spike conditions for the
	initial signal and noise variables.  Hence $\mathcal R_0$ holds with
	probability tending to one.
	
	Assume now that $\mathcal R_t$ holds for some fixed $t<T$.  The history
	constraints \eqref{eq:proof_history_AM} and \eqref{eq:proof_history_ATU} are
	finite linear observations of $\boldsymbol A$.  By
	Proposition~\ref{prop:conditional_gaussian_history},
	conditioning on $\mathcal F_t$ gives the decomposition
	\begin{equation}
	\boldsymbol A
	=
	\boldsymbol A_{\parallel,t}
	+
	\boldsymbol A_{\perp,t},
	\end{equation}
	where the residual matrix is centered correlated Gaussian with flat covariance.
	By Proposition~\ref{prop:bounded_deformation_mde}, the deformation is bounded
	and the conditioned block linearization satisfies the assumptions of the
	external MDE input.  Therefore
	\begin{equation}
	\frac1N\sum_{j=1}^N
	|d_{t,j}-T_{t,j}|^2
	\overset{p}{\longrightarrow}0,
	\end{equation}
	which proves the MDE response part of the theorem.
	
	Next, Proposition~\ref{prop:schur_kernel} gives the history-conditioned Schur
	kernel
	\begin{equation}
	m_{t,j}
	=
	\gamma_{t,j}T_{t,j}q_{t,j}
	+
    \mathcal Z_{t,j}
	+
	\Delta_{t,j},
	\end{equation}
	with $\|\boldsymbol\Delta_t\|_N\to0$ in probability.  The covariance of the
	Gaussian history is identified by
	Propositions~\ref{prop:two_resolvent_response} and
	\ref{prop:covariance_empirical_law}:
	\begin{equation}
	\operatorname{Cov}(\mathcal Z_{r,j}^{\mathrm G},\mathcal Z_{s,j}^{\mathrm G}\mid\mathcal P_t)
	=
	\zeta_j^{r,s}
	=
	\gamma_w^2T_{r,j}T_{s,j}\Theta_j^{r,s}.
	\end{equation}
	The empirical-law part of Proposition~\ref{prop:covariance_empirical_law} then
	yields the empirical pseudo-Lipschitz law
	\begin{equation}
	\frac1N\sum_{j=1}^N
	\psi(\mathcal X_{j,N}^t)
	-
	\frac1N\sum_{j=1}^N
	\mathbb E_{\mathrm G,j}
	\left[
	\psi(\mathcal X_{j,\mathrm G}^t)
	\mid
	\mathcal P_t
	\right]
	\overset{p}{\longrightarrow}0
	\end{equation}
	for every $\psi\in \PL(2)$.
	
	The EP cavity identity
	\begin{equation}
	\boldsymbol\Pi_t\boldsymbol h_t
	=
	\boldsymbol D_t^{-1}\boldsymbol m_t
	-
	\boldsymbol\Gamma_t\boldsymbol q_t
	\end{equation}
	combined with the linear Gaussian kernel gives the cancellation of the
	instantaneous response.  Proposition~\ref{prop:cavity_memory_decomposition}
	then applies the Gaussian regression identity to obtain
	\begin{equation}
	h_{t,j}
	=
	\mu_{t,j}
	+
	\sqrt{\tau_{t,j}}\,W_{t,j}
	+
	\varepsilon_{t,j},
	\end{equation}
	where
	\begin{equation}
	\frac1N\sum_{j=1}^N|\varepsilon_{t,j}|^2
	\overset{p}{\longrightarrow}0.
	\end{equation}
	This proves the cavity decomposition part of the theorem.
	
	Finally, Proposition~\ref{prop:regularity_closure} shows that the newly
	generated signal-side histories remain admissible and that
	\begin{equation}
	\boldsymbol u_t
	=
	\boldsymbol w-\boldsymbol A\boldsymbol m_t
	\end{equation}
	is measurement-side admissible.  Therefore $\mathcal R_{t+1}$ holds.  Since
	the horizon $T$ is fixed, induction over $t=0,1,\ldots,T$ completes the proof.
\end{IEEEproof}

\section{Consequences of the Dynamic Theorem}
\label{sec:proof_main_consequences}

The general finite-time dynamic theorem identifies the precise Gaussian object
created by the linear module.  We now translate this dynamic statement into the
main consequences for diagonal EP.  The key point is that the Gaussian process
in Theorem~\ref{thm:general_dynamic} is not necessarily an innovation process.
The standard EP cavity removes the instantaneous self-response, but it does not
in general remove the component of the current Gaussian residual that is
predictable from its past.

The results in this section are first stated for the predictable-precision
dynamic of Theorem~\ref{thm:general_dynamic}.  The adaptive diagonal EP
recursion of Section~\ref{sec:system_model} is connected to this form by the
precision replacement theorem proved in Section~\ref{sec:precision_replacement}.
Thus the purpose of the present section is not to introduce new random-matrix
arguments, but to show how the Gaussian dynamic theorem yields the memory
defect and the corrected coordinate-wise state evolution.

\subsection{Memory Defect of Standard Diagonal EP}
\label{subsec:proof_memory_defect}

We first prove Theorem~\ref{thm:memory_defect}.  The message delivered by
module A to module B is
\begin{equation}
\boldsymbol r_{A\to B}^{t}
=
\boldsymbol x+\boldsymbol h_t .
\end{equation}
By Theorem~\ref{thm:general_dynamic}, the EP cavity error admits the
decomposition
\begin{equation}
	h_{t,j}
	=
	\mu_{t,j}
	+
	\sqrt{\tau_{t,j}}\,W_{t,j}
	+
	\varepsilon_{t,j},
	\label{eq:proof_h_memory_decomposition}
\end{equation}
where, conditionally on the MDE-generated environment $\mathcal P_t$,
$W_{t,j}\sim\mathcal N(0,1)$ is independent of the past Gaussian history, and
\begin{equation}
\frac1N\sum_{j=1}^N|\varepsilon_{t,j}|^2
\overset{p}{\longrightarrow}0.
\end{equation}
Substituting \eqref{eq:proof_h_memory_decomposition} into
$\boldsymbol r_{A\to B}^{t}=\boldsymbol x+\boldsymbol h_t$ gives
\begin{equation}
	r_{A\to B,t,j}
	=
	x_j
	+
	\mu_{t,j}
	+
	\sqrt{\tau_{t,j}}\,W_{t,j}
	+
	\varepsilon_{t,j}.
\end{equation}
Equivalently,
\begin{equation}
r_{A\to B,t,j}
=
x_j
+
\mu_{t,j}
+
\sqrt{\tau_{t,j}}\,W_{t,j}
+
o_p^{\ell_2}(1).
\end{equation}
This is exactly \eqref{eq:std_channel_memory}, and therefore proves
Theorem~\ref{thm:memory_defect}.

The interpretation is immediate.  The standard diagonal EP cavity is Gaussian,
but it is not generally a fresh AWGN observation of $x_j$.  The term
\begin{equation}
\mu_{t,j}
=
\alpha_{t,j}
\boldsymbol\zeta_j^{t,<t}
(\boldsymbol\zeta_j^{<t,<t})^\dagger
\boldsymbol{\mathcal Z}_{<t,j}
\end{equation}
is the conditional mean of the current Gaussian residual given its past
Gaussian history.  Hence $\mu_{t,j}$ is a predictable memory component, not a
vanishing error term.  Unless this component is asymptotically negligible, the
scalar denoiser in the standard EP recursion is driven by the shifted channel
\begin{equation}
X+\mu_{t,j}+\sqrt{\tau_{t,j}}\,W
\end{equation}
rather than by the fresh channel
\begin{equation}
X+\sqrt{\tau_{t,j}}\,W.
\end{equation}
Thus Gaussianity survives under variance-profile measurements, but freshness
does not hold in general.

\subsection{The No-Memory Special Case}
\label{subsec:proof_no_memory}

We next prove Corollary~\ref{cor:no_memory}.  Suppose that, for every fixed
$t$ and all coordinates $j$,
\begin{equation}
\boldsymbol\zeta_j^{t,<t}
=
\boldsymbol 0.
\end{equation}
Then, by the definition of the memory term,
\begin{equation}
\mu_{t,j}
=
\alpha_{t,j}
\boldsymbol\zeta_j^{t,<t}
(\boldsymbol\zeta_j^{<t,<t})^\dagger
\boldsymbol{\mathcal Z}_{<t,j}
=
0.
\end{equation}
Substituting this identity into
\eqref{eq:proof_h_memory_decomposition} yields
\begin{equation}
h_{t,j}
=
\sqrt{\tau_{t,j}}\,W_{t,j}
+
\varepsilon_{t,j},
\end{equation}
with
\begin{equation}
\frac1N\sum_{j=1}^N|\varepsilon_{t,j}|^2
\overset{p}{\longrightarrow}0.
\end{equation}
Therefore
\begin{equation}
h_{t,j}
=
\sqrt{\tau_{t,j}}\,W_{t,j}
+
o_p^{\ell_2}(1),
\end{equation}
which proves \eqref{eq:no_memory_fresh}.

The condition
\begin{equation}
\boldsymbol\zeta_j^{t,<t}=\boldsymbol 0
\end{equation}
means that the current Gaussian residual is orthogonal to its own past.  In
that special case, the Gaussian process produced by the linear module is
already an innovation process, and the standard EP cavity itself is fresh.
This is the profile-dependent counterpart of the innovation mechanism available
in rotationally invariant models.

\subsection{Memory-Corrected Coordinate-Wise State Evolution}
\label{subsec:proof_corrected_se}

We finally prove Theorem~\ref{thm:corrected_se}.  Define the corrected cavity
by subtracting the predictable Gaussian-history component:
\begin{equation}
\widetilde h_{t,j}
=
h_{t,j}-\mu_{t,j},
\qquad
\widetilde r_{A\to B,t,j}
=
x_j+\widetilde h_{t,j}.
\end{equation}
Using \eqref{eq:proof_h_memory_decomposition}, we obtain
\begin{equation}
\widetilde h_{t,j}
=
\sqrt{\tau_{t,j}}\,W_{t,j}
+
\varepsilon_{t,j}.
\end{equation}
Since
\begin{equation}
\frac1N\sum_{j=1}^N|\varepsilon_{t,j}|^2
\overset{p}{\longrightarrow}0,
\end{equation}
it follows that
\begin{equation}
	\widetilde h_{t,j}
	=
	\sqrt{\tau_{t,j}}\,W_{t,j}
	+
	o_p^{\ell_2}(1).
\end{equation}
Equivalently,
\begin{equation}
	\widetilde r_{A\to B,t,j}
	=
	x_j
	+
	\sqrt{\tau_{t,j}}\,W_{t,j}
	+
	o_p^{\ell_2}(1).
	\label{eq:proof_corrected_channel}
\end{equation}
This proves the fresh-channel part of
Theorem~\ref{thm:corrected_se}.

It remains to identify the corresponding scalar state-evolution update.  The
matched scalar reference channel associated with
\eqref{eq:proof_corrected_channel} is
\begin{equation}
	R_{t,j}^{\mathrm{SE}}
	=
	X_j+\sqrt{\tau_{t,j}}\,W_{t,j},
	\qquad
	W_{t,j}\sim\mathcal N(0,1).
	\label{eq:matched_SE_channel}
\end{equation}
Since $\tau_{t,j}$ is the innovation variance, the matched scalar precision is
$\tau_{t,j}^{-1}$.  Thus the corresponding state-evolution posterior error is
\begin{equation}
	P_{t,j}^{\mathrm{SE}}
	=
	\eta(R_{t,j}^{\mathrm{SE}};\tau_{t,j}^{-1})-X_j,
\end{equation}
and the scalar posterior variance is
\begin{equation}
	V_{B,t,j}^{\mathrm{SE}}
	=
	v_B(R_{t,j}^{\mathrm{SE}};\tau_{t,j}^{-1}).
\end{equation}

The empirical convergence statement follows from the empirical Gaussian law in
Theorem~\ref{thm:general_dynamic}.  Indeed, after subtracting the predictable
memory component, the Gaussian reference history is generated by the fresh
scalar channel \eqref{eq:matched_SE_channel}.  The scalar posterior mean and
variance maps preserve $\PL(2)$ empirical convergence on compact precision
intervals by Assumption~\ref{ass:ep_regular}.  Therefore, for any
$\psi\in \PL(2)$,
\begin{equation}
\frac1N\sum_{j=1}^N
\psi(\mathcal X_{j,N}^{t,\mathrm{corr}})
-
\frac1N\sum_{j=1}^N
\mathbb E_{\mathrm{SE},j}
\left[
\psi(\mathcal X_{j,\mathrm{SE}}^{t,\mathrm{corr}})
\mid
\mathcal P_t
\right]
\overset{p}{\longrightarrow}0.
\end{equation}
This is \eqref{eq:corrected_empirical_se}, and hence proves
Theorem~\ref{thm:corrected_se}.

The corrected state evolution should be understood as an innovation
decomposition of the Gaussian process identified in
Theorem~\ref{thm:general_dynamic}.  The correction does not alter the linear
Gaussian kernel; it removes only the conditional mean of the current residual
given its past.  After this removal, the prior module sees the matched scalar
channel \eqref{eq:matched_SE_channel}.  This is the coordinate-wise analogue of
the fresh AWGN channel underlying scalar state evolution.

\section{Precision Replacement for Adaptive Diagonal EP}
\label{sec:precision_replacement}

The dynamic theorem in Section~\ref{sec:dynamic_proof} was proved in a
predictable-precision form.  This form is natural for the conditioning argument:
at the beginning of the $t$th linear step, the diagonal loading
$\boldsymbol\Gamma_t$ is already part of the regular finite-time environment,
so conditioning on the past imposes only finite linear observations on the
Gaussian matrix $\boldsymbol A$.  The actual diagonal EP recursion, however,
updates its precisions from finite-dimensional retained variances and scalar
posterior variances.  This section shows that the two descriptions are
asymptotically equivalent.

The result is a stability statement, not a new state evolution.  Small
empirical perturbations in the incoming Gaussian message and its diagonal
precision lead to small perturbations in the linear posterior, the retained
diagonal variances, the cavity message, the scalar prior update, and finally the
next outgoing precision and mean.  This closes the gap between the
predictable-precision dynamic theorem and the adaptive diagonal EP recursion
formulated in Section~\ref{sec:system_model}.

\subsection{Regularity Conditions for Replacement}
\label{subsec:replacement_regular_conditions}

We use two replacement regularity conditions in this section.  They refine the
standing scalar-module regularity of Assumption~\ref{ass:ep_regular} only for
the deterministic perturbation argument connecting predictable precisions to
finite-sample adaptive precisions.  They are not used as additional Gaussian
conditioning information.

\begin{assumption}[Additional replacement stability of the scalar module]
	\label{ass:regularized_scalar_module}
	Assumption~\ref{ass:ep_regular} gives the compact precision intervals used throughout
	the paper.  For the replacement argument we additionally assume that the scalar
	posterior variances used in the prior module are projected onto a compact
	positive interval,
	\begin{equation}
	0<v_{\min}
	\le
	v_{B,t,j}
	\le
	v_{\max}<\infty,
	\end{equation}
	for all fixed \(t\) and all coordinates \(j\).  Moreover, the scalar maps
	\begin{equation}
	(r,\rho)\mapsto \eta(r;\rho),
	\qquad
	(r,\rho)\mapsto v_B(r;\rho)
	\end{equation}
	are empirically stable on compact precision intervals.  Specifically, whenever
	\begin{equation}
	\|\boldsymbol r_N-\boldsymbol r'_N\|_N\to0,
	\qquad
	\|\boldsymbol\rho_N-\boldsymbol\rho'_N\|_N\to0,
	\end{equation}
	with all entries of \(\boldsymbol\rho_N,\boldsymbol\rho'_N\) lying in a compact
	positive interval, we have
	\begin{equation}
	\left\|
	\boldsymbol\eta(\boldsymbol r_N;\boldsymbol\rho_N)
	-
	\boldsymbol\eta(\boldsymbol r'_N;\boldsymbol\rho'_N)
	\right\|_N
	\to0,
	\end{equation}
	and
	\begin{equation}
	\left\|
	\boldsymbol v_B(\boldsymbol r_N;\boldsymbol\rho_N)
	-
	\boldsymbol v_B(\boldsymbol r'_N;\boldsymbol\rho'_N)
	\right\|_N
	\to0.
	\end{equation}
\end{assumption}

Assumption~\ref{ass:regularized_scalar_module} is a replacement-stability condition only; it is not used to justify Gaussian conditioning.  It is satisfied by many standard priors after the usual variance projection.  In particular, finite-alphabet and bounded-support priors lead to
regular scalar posterior mean and variance maps on compact precision intervals.

The corrected branch also requires that the Gaussian regression used to remove
the predictable memory component be well conditioned.  This condition is not a
numerical artifact; it means that the past Gaussian-history directions retained
by the regression are not asymptotically degenerate.

\begin{assumption}[Stable memory regression]
	\label{ass:stable_memory_regression}
	For the corrected branch and every fixed $T$, the Gaussian-history covariance
	matrices satisfy the following regularity condition.  After removing exactly
	redundant zero directions by the Moore--Penrose inverse, the nonzero spectrum
	of
	\begin{equation}
	\boldsymbol\zeta_j^{<t,<t}
	\end{equation}
	is uniformly bounded away from zero: there exists $c_\zeta>0$ such that
	\begin{equation}
	\lambda_{\min}^+
	\!\left(
	\boldsymbol\zeta_j^{<t,<t}
	\right)
	\ge
	c_\zeta
	\end{equation}
	for all $t\le T$ and all relevant coordinates $j$, with probability tending to
	one.  Equivalently,
	\begin{equation}
	\left\|
	(\boldsymbol\zeta_j^{<t,<t})^\dagger
	\right\|
	\le
	c_\zeta^{-1}
	\end{equation}
	on the retained Gaussian-history subspace.  In addition, the innovation
	variances of the corrected scalar channels satisfy
	\begin{equation}
	0<\tau_{\min}
	\le
	\tau_{t,j}
	\le
	\tau_{\max}<\infty
	\end{equation}
	for all fixed $t\le T$ and all relevant $j$, with probability tending to one.
\end{assumption}

Assumption~\ref{ass:stable_memory_regression} excludes only nearly redundant
Gaussian-history regressors.  Such degeneracies are not the phenomenon studied
here.  Our purpose is to remove the nonzero predictable memory generated by the
variance profile; for this operation to be stable, the corresponding finite
Gaussian regression must be well conditioned.

\subsection{Actual and Predictable Recursions}
\label{subsec:actual_predictable_recursions}

We compare two recursions.  The superscript ``act'' refers to the actual
finite-dimensional adaptive diagonal EP recursion, while the superscript
``orc'' refers to the MDE-predictable recursion analyzed by
Theorem~\ref{thm:general_dynamic}.  The word ``oracle'' here does not mean that
additional observations are given to the algorithm; it only indicates that the
finite-dimensional retained variances are replaced by their MDE-predictable
counterparts.

At iteration $t$, the actual linear module uses
\begin{equation}
	\boldsymbol C_t^{\mathrm{act}}
	=
	\left(
	\gamma_w\boldsymbol A^{\mathsf T}\boldsymbol A
	+
	\boldsymbol\Gamma_t^{\mathrm{act}}
	\right)^{-1},
\end{equation}
and
\begin{equation}
	\boldsymbol m_t^{\mathrm{act}}
	=
	\boldsymbol C_t^{\mathrm{act}}
	\left(
	\gamma_w\boldsymbol A^{\mathsf T}\boldsymbol w
	+
	\boldsymbol\Gamma_t^{\mathrm{act}}
	\boldsymbol q_t^{\mathrm{act}}
	\right).
\end{equation}
Its retained diagonal variance is
\begin{equation}
d_{t,j}^{\mathrm{act}}
=
[\boldsymbol C_t^{\mathrm{act}}]_{jj},
\end{equation}
and the outgoing cavity precision is
\begin{equation}
	\pi_{t,j}^{\mathrm{act}}
	=
	\operatorname{Proj}_{[\pi_{\min},\pi_{\max}]}
	\left(
	(d_{t,j}^{\mathrm{act}})^{-1}
	-
	\gamma_{t,j}^{\mathrm{act}}
	\right).
\end{equation}

The predictable recursion has the same finite-dimensional linear estimate
with $\boldsymbol\Gamma_t^{\mathrm{orc}}$ and
$\boldsymbol q_t^{\mathrm{orc}}$, namely
\begin{equation}
	\boldsymbol C_t^{\mathrm{orc}}
	=
	\left(
	\gamma_w\boldsymbol A^{\mathsf T}\boldsymbol A
	+
	\boldsymbol\Gamma_t^{\mathrm{orc}}
	\right)^{-1},
\end{equation}
\begin{equation}
\boldsymbol m_t^{\mathrm{orc}}
=
\boldsymbol C_t^{\mathrm{orc}}
\left(
\gamma_w\boldsymbol A^{\mathsf T}\boldsymbol w
+
\boldsymbol\Gamma_t^{\mathrm{orc}}
\boldsymbol q_t^{\mathrm{orc}}
\right),
\end{equation}
but its diagonal cavity precision is generated from the MDE response:
\begin{equation}
	\pi_{t,j}^{\mathrm{orc}}
	=
	\operatorname{Proj}_{[\pi_{\min},\pi_{\max}]}
	\left(
	(T_{t,j}^{\mathrm{orc}})^{-1}
	-
	\gamma_{t,j}^{\mathrm{orc}}
	\right).
\end{equation}
The corresponding cavity error is
\begin{equation}
	h_{t,j}^{\mathrm{orc}}
	=
	(\pi_{t,j}^{\mathrm{orc}})^{-1}
	\left(
	(T_{t,j}^{\mathrm{orc}})^{-1}m_{t,j}^{\mathrm{orc}}
	-
	\gamma_{t,j}^{\mathrm{orc}}q_{t,j}^{\mathrm{orc}}
	\right).
\end{equation}
The actual cavity is defined in the same way as in Section~\ref{sec:system_model},
using $d_{t,j}^{\mathrm{act}}$ and $\pi_{t,j}^{\mathrm{act}}$.

The two branches differ only in the scalar input fed to the prior module.  In
the standard branch,
\begin{equation}
s_{t,j}^{\mathrm{std}}
=
h_{t,j},
\qquad
\rho_{t,j}^{\mathrm{std}}
=
\pi_{t,j}.
\end{equation}
In the corrected branch,
\begin{equation}
s_{t,j}^{\mathrm{corr}}
=
h_{t,j}-\mu_{t,j},
\qquad
\rho_{t,j}^{\mathrm{corr}}
=
\tau_{t,j}^{-1}.
\end{equation}
Thus a branch $b\in\{\mathrm{std},\mathrm{corr}\}$ is described by a scalar
input pair
\begin{equation}
(s_{t,j}^{b},\rho_{t,j}^{b})
\end{equation}
and the prior module applies
\begin{equation}
\widehat x_{B,t,j}^{b}
=
\eta(x_j+s_{t,j}^{b};\rho_{t,j}^{b}),
\end{equation}
\begin{equation}
p_{t,j}^{b}
=
\widehat x_{B,t,j}^{b}-x_j,
\end{equation}
\begin{equation}
v_{B,t,j}^{b}
=
v_B(x_j+s_{t,j}^{b};\rho_{t,j}^{b}).
\end{equation}
The outgoing precision and mean error are then
\begin{align}
	\gamma_{t+1,j}^{b}
	&=
	\operatorname{Proj}_{[\gamma_{\min},\gamma_{\max}]}
	\left(
	(v_{B,t,j}^{b})^{-1}
	-
	\rho_{t,j}^{b}
	\right),
	\label{eq:gamma_next_branch_replacement}
	\\
	\gamma_{t+1,j}^{b} q_{t+1,j}^{b}
	&=
	(v_{B,t,j}^{b})^{-1}p_{t,j}^{b}
	-
	\rho_{t,j}^{b}s_{t,j}^{b}.
	\label{eq:q_next_branch_replacement}
\end{align}

The goal is to prove, for each fixed branch $b$ and each fixed $t\le T$,
\begin{equation}
\left\|
\boldsymbol\gamma_t^{b,\mathrm{act}}
-
\boldsymbol\gamma_t^{b,\mathrm{orc}}
\right\|_N
\to0,
\qquad
\left\|
\boldsymbol q_t^{b,\mathrm{act}}
-
\boldsymbol q_t^{b,\mathrm{orc}}
\right\|_N
\to0
\end{equation}
in probability.

\subsection{Stability of the Linear Module}
\label{subsec:linear_module_stability}

The first step is a deterministic perturbation bound for the linear Gaussian
module.  Its proof is based on two identities: one for the posterior mean and
one for the covariance resolvent.

\begin{proposition}[Linear-module stability]
	\label{prop:linear_module_stability}
	Assume that, for some fixed $t$,
	\begin{equation}
		\left\|
		\boldsymbol q_t^{\mathrm{act}}
		-
		\boldsymbol q_t^{\mathrm{orc}}
		\right\|_N
		\overset{p}{\longrightarrow}0,
		\qquad
		\left\|
		\boldsymbol\gamma_t^{\mathrm{act}}
		-
		\boldsymbol\gamma_t^{\mathrm{orc}}
		\right\|_N
		\overset{p}{\longrightarrow}0.
		\label{eq:linear_stability_assumption}
	\end{equation}
	Then
	\begin{align}
		\left\|
		\boldsymbol m_t^{\mathrm{act}}
		-
		\boldsymbol m_t^{\mathrm{orc}}
		\right\|_N
		&\overset{p}{\longrightarrow}0,
		\label{eq:m_linear_stability}
		\\
		\left\|
		\boldsymbol d_t^{\mathrm{act}}
		-
		\boldsymbol T_t^{\mathrm{orc}}
		\right\|_N
		&\overset{p}{\longrightarrow}0,
		\label{eq:d_linear_stability}
		\\
		\left\|
		\boldsymbol\pi_t^{\mathrm{act}}
		-
		\boldsymbol\pi_t^{\mathrm{orc}}
		\right\|_N
		&\overset{p}{\longrightarrow}0,
		\label{eq:pi_linear_stability}
		\\
		\left\|
		\boldsymbol h_t^{\mathrm{act}}
		-
		\boldsymbol h_t^{\mathrm{orc}}
		\right\|_N
		&\overset{p}{\longrightarrow}0.
		\label{eq:h_linear_stability}
	\end{align}
\end{proposition}

\begin{IEEEproof}
	The detailed empirical-norm estimates are given in
	Appendix~\ref{app:precision_replacement}.  We indicate the two identities on
	which the proof rests.  Let
	\begin{equation}
	\boldsymbol C_a=\boldsymbol C_t^{\mathrm{act}},
	\qquad
	\boldsymbol C_o=\boldsymbol C_t^{\mathrm{orc}},
	\end{equation}
	and similarly write
	$\boldsymbol\Gamma_a,\boldsymbol\Gamma_o,\boldsymbol q_a,\boldsymbol q_o$.
	Using the normal equations for the two linear posterior means gives
	\begin{equation}
		\boldsymbol m_a-\boldsymbol m_o
		=
		\boldsymbol C_a
		\left[
		\boldsymbol\Gamma_a(\boldsymbol q_a-\boldsymbol q_o)
		+
		(\boldsymbol\Gamma_a-\boldsymbol\Gamma_o)
		(\boldsymbol q_o-\boldsymbol m_o)
		\right].
	\end{equation}
	Since the diagonal loadings are bounded below by $\gamma_{\min}$,
	$\|\boldsymbol C_a\|$ is uniformly bounded.  Together with admissibility of the
	oracle history and \eqref{eq:linear_stability_assumption}, this yields
	\eqref{eq:m_linear_stability}.
	
	For the covariance matrices,
	\begin{equation}
		\boldsymbol C_a-\boldsymbol C_o
		=
		-\boldsymbol C_a
		(\boldsymbol\Gamma_a-\boldsymbol\Gamma_o)
		\boldsymbol C_o.
	\end{equation}
	The right-hand side is controlled in normalized Frobenius norm because
	$\boldsymbol\Gamma_a-\boldsymbol\Gamma_o$ is diagonal and
	\begin{equation}
	N^{-1/2}
	\|\boldsymbol\Gamma_a-\boldsymbol\Gamma_o\|_F
	=
	\|\boldsymbol\gamma_a-\boldsymbol\gamma_o\|_N.
	\end{equation}
	Combining this bound with the MDE response
	$d_{t,j}^{\mathrm{orc}}=T_{t,j}^{\mathrm{orc}}+o_p^{\ell_2}(1)$ yields
	\eqref{eq:d_linear_stability}.  The projection map and the inverse map are
	Lipschitz on compact positive intervals, which gives
	\eqref{eq:pi_linear_stability}; substituting the stable quantities into the
	cavity formula gives \eqref{eq:h_linear_stability}.
\end{IEEEproof}

\subsection{Stability of the Prior Module}
\label{subsec:prior_module_stability}

The prior module is separable, so its stability follows from the regularity of
the scalar posterior mean and variance maps.  We state the branch-wise result
in a form that covers both the standard and corrected branches.

\begin{proposition}[Prior-module stability]
	\label{prop:prior_module_stability}
	Fix $b\in\{\mathrm{std},\mathrm{corr}\}$.  Suppose that
	\begin{equation}
		\left\|
		\boldsymbol s_t^{b,\mathrm{act}}
		-
		\boldsymbol s_t^{b,\mathrm{orc}}
		\right\|_N
		\overset{p}{\longrightarrow}0,
		\qquad
		\left\|
		\boldsymbol\rho_t^{b,\mathrm{act}}
		-
		\boldsymbol\rho_t^{b,\mathrm{orc}}
		\right\|_N
		\overset{p}{\longrightarrow}0.
		\label{eq:prior_stability_input}
	\end{equation}
	Then
	\begin{align}
		\left\|
		\boldsymbol p_t^{b,\mathrm{act}}
		-
		\boldsymbol p_t^{b,\mathrm{orc}}
		\right\|_N
		&\overset{p}{\longrightarrow}0,
		\\
		\left\|
		\boldsymbol v_{B,t}^{b,\mathrm{act}}
		-
		\boldsymbol v_{B,t}^{b,\mathrm{orc}}
		\right\|_N
		&\overset{p}{\longrightarrow}0,
		\\
		\left\|
		\boldsymbol\gamma_{t+1}^{b,\mathrm{act}}
		-
		\boldsymbol\gamma_{t+1}^{b,\mathrm{orc}}
		\right\|_N
		&\overset{p}{\longrightarrow}0,
		\\
		\left\|
		\boldsymbol q_{t+1}^{b,\mathrm{act}}
		-
		\boldsymbol q_{t+1}^{b,\mathrm{orc}}
		\right\|_N
		&\overset{p}{\longrightarrow}0.
	\end{align}
\end{proposition}

\begin{IEEEproof}
	See Appendix~\ref{app:precision_replacement}.  The proof applies
	Assumption~\ref{ass:regularized_scalar_module} to the scalar maps
	$\eta(\cdot;\cdot)$ and $v_B(\cdot;\cdot)$.  The updates
	\eqref{eq:gamma_next_branch_replacement} and
	\eqref{eq:q_next_branch_replacement} are stable because the maps
	$x\mapsto x^{-1}$ and $\operatorname{Proj}_{[a,b]}(x)$ are Lipschitz on compact
	positive intervals.
\end{IEEEproof}

For the corrected branch, the input stability in
\eqref{eq:prior_stability_input} requires stability of the memory map.  The next
proposition states this consequence of
Assumption~\ref{ass:stable_memory_regression}.

\begin{proposition}[Stability of the memory map]
	\label{prop:memory_map_stability}
	Under Assumption~\ref{ass:stable_memory_regression}, if the histories and
	MDE-generated environments of the actual and oracle corrected branches are
	empirically close up to time $t$, then
	\begin{equation}
		\left\|
		\boldsymbol\mu_t^{\mathrm{act}}
		-
		\boldsymbol\mu_t^{\mathrm{orc}}
		\right\|_N
		\overset{p}{\longrightarrow}0,
		\qquad
		\left\|
		\boldsymbol\tau_t^{\mathrm{act}}
		-
		\boldsymbol\tau_t^{\mathrm{orc}}
		\right\|_N
		\overset{p}{\longrightarrow}0.
	\end{equation}
	Consequently,
	\begin{equation}
		\left\|
		\widetilde{\boldsymbol h}_t^{\mathrm{act}}
		-
		\widetilde{\boldsymbol h}_t^{\mathrm{orc}}
		\right\|_N
		\overset{p}{\longrightarrow}0,
		\qquad
		\left\|
		(\boldsymbol\tau_t^{\mathrm{act}})^{-1}
		-
		(\boldsymbol\tau_t^{\mathrm{orc}})^{-1}
		\right\|_N
		\overset{p}{\longrightarrow}0.
	\end{equation}
\end{proposition}

\begin{IEEEproof}
	See Appendix~\ref{app:precision_replacement}.  The main point is that, on the
	retained Gaussian-history subspace,
	\begin{equation}
	\left\|
	(\boldsymbol\zeta_j^{<t,<t})^\dagger
	\right\|
	\le
	c_\zeta^{-1}.
	\end{equation}
	Since $t$ is fixed, the Gaussian regression map defining $\mu_{t,j}$ is then
	Lipschitz with respect to the finite covariance blocks and the past Gaussian
	history.  The lower and upper bounds on $\tau_{t,j}$ make the inverse
	innovation precision stable.
\end{IEEEproof}

\subsection{Proof of the Precision Replacement Theorem}
\label{subsec:proof_precision_replacement}

We now prove Theorem~\ref{thm:precision_replacement}.  The proof is by
induction over $t$ and is carried out separately for each branch
$b\in\{\mathrm{std},\mathrm{corr}\}$.

At $t=0$, the actual and oracle recursions are initialized identically.  Hence
\begin{equation}
\left\|
\boldsymbol q_0^{b,\mathrm{act}}
-
\boldsymbol q_0^{b,\mathrm{orc}}
\right\|_N
=
0,
\qquad
\left\|
\boldsymbol\gamma_0^{b,\mathrm{act}}
-
\boldsymbol\gamma_0^{b,\mathrm{orc}}
\right\|_N
=
0.
\end{equation}

Assume that, for some fixed $t<T$,
\begin{equation}
	\left\|
	\boldsymbol q_t^{b,\mathrm{act}}
	-
	\boldsymbol q_t^{b,\mathrm{orc}}
	\right\|_N
	\overset{p}{\longrightarrow}0,
	\qquad
	\left\|
	\boldsymbol\gamma_t^{b,\mathrm{act}}
	-
	\boldsymbol\gamma_t^{b,\mathrm{orc}}
	\right\|_N
	\overset{p}{\longrightarrow}0.
\end{equation}
By Proposition~\ref{prop:linear_module_stability},
\begin{equation}
\left\|
\boldsymbol m_t^{b,\mathrm{act}}
-
\boldsymbol m_t^{b,\mathrm{orc}}
\right\|_N
\to0,
\end{equation}
\begin{equation}
\left\|
\boldsymbol\pi_t^{b,\mathrm{act}}
-
\boldsymbol\pi_t^{b,\mathrm{orc}}
\right\|_N
\to0,
\end{equation}
and
\begin{equation}
\left\|
\boldsymbol h_t^{b,\mathrm{act}}
-
\boldsymbol h_t^{b,\mathrm{orc}}
\right\|_N
\to0
\end{equation}
in probability.

For the standard branch, the prior input is
\begin{equation}
\boldsymbol s_t^{\mathrm{std}}
=
\boldsymbol h_t,
\qquad
\boldsymbol\rho_t^{\mathrm{std}}
=
\boldsymbol\pi_t.
\end{equation}
Thus the input stability condition
\eqref{eq:prior_stability_input} follows directly from the linear-module
stability.  Proposition~\ref{prop:prior_module_stability} then gives
\begin{equation}
\left\|
\boldsymbol\gamma_{t+1}^{\mathrm{std},\mathrm{act}}
-
\boldsymbol\gamma_{t+1}^{\mathrm{std},\mathrm{orc}}
\right\|_N
\to0,
\end{equation}
and
\begin{equation}
\left\|
\boldsymbol q_{t+1}^{\mathrm{std},\mathrm{act}}
-
\boldsymbol q_{t+1}^{\mathrm{std},\mathrm{orc}}
\right\|_N
\to0.
\end{equation}

For the corrected branch, the prior input is
\begin{equation}
\boldsymbol s_t^{\mathrm{corr}}
=
\widetilde{\boldsymbol h}_t
=
\boldsymbol h_t-\boldsymbol\mu_t,
\qquad
\boldsymbol\rho_t^{\mathrm{corr}}
=
\boldsymbol\tau_t^{-1}.
\end{equation}
The stability of $\boldsymbol h_t$ follows from the linear-module stability,
and the stability of $\boldsymbol\mu_t$ and $\boldsymbol\tau_t^{-1}$ follows
from Proposition~\ref{prop:memory_map_stability}.  Therefore
\begin{equation}
\left\|
\boldsymbol s_t^{\mathrm{corr},\mathrm{act}}
-
\boldsymbol s_t^{\mathrm{corr},\mathrm{orc}}
\right\|_N
\to0,
\qquad
\left\|
\boldsymbol\rho_t^{\mathrm{corr},\mathrm{act}}
-
\boldsymbol\rho_t^{\mathrm{corr},\mathrm{orc}}
\right\|_N
\to0.
\end{equation}
Applying Proposition~\ref{prop:prior_module_stability} gives
\begin{equation}
\left\|
\boldsymbol\gamma_{t+1}^{\mathrm{corr},\mathrm{act}}
-
\boldsymbol\gamma_{t+1}^{\mathrm{corr},\mathrm{orc}}
\right\|_N
\to0,
\end{equation}
and
\begin{equation}
\left\|
\boldsymbol q_{t+1}^{\mathrm{corr},\mathrm{act}}
-
\boldsymbol q_{t+1}^{\mathrm{corr},\mathrm{orc}}
\right\|_N
\to0.
\end{equation}

Thus, in either branch, the induction hypothesis at time $t$ implies the same
claim at time $t+1$.  Since the horizon $T$ is fixed, induction proves
\eqref{eq:gamma_replacement} and \eqref{eq:q_replacement} for every
$t\le T$.  This completes the proof of
Theorem~\ref{thm:precision_replacement}.

\section{Conclusion}
\label{sec:conclusion}

This paper studied diagonal expectation propagation under variance-profile Gaussian measurements.  The main finding is that variance profiles do not destroy the Gaussian nature of the linear-module output, but they generally destroy its freshness.  After conditioning on the finite linear history, the linear module produces a coordinate-dependent Gaussian process rather than a fresh scalar Gaussian channel.  The standard diagonal EP cavity removes the instantaneous response of the incoming message, but it may leave a predictable memory component inherited from previous residuals.

We characterized this effect through a conditioned matrix-Dyson-equation response and a Schur-complement representation of the linear module.  A Gaussian-regression decomposition then identifies the innovation part of the residual process and yields an oracle state-evolution-level correction.  Thus, under a general variance profile, the natural limiting object for diagonal EP is a Gaussian-process dynamics with profile-dependent memory rather than a conventional scalar fresh-noise state evolution.

Several questions remain open.  The most important one is algorithmic: the memory correction in this paper is an oracle decoupling device, and practical procedures for estimating and removing the memory term from finite-dimensional iterates remain to be developed.  Other natural directions include extensions to non-Gaussian or sparse variance-profile matrices, and the study of fixed-point stability and variational interpretations of the resulting memory-aware EP dynamics.

\appendices
\section{Auxiliary Probability and Empirical-Convergence Lemmas}
\label{app:aux_probability}

This appendix collects several probability lemmas used throughout the proof.
They are independent of the EP recursion and of the particular random-matrix
linearization.  Their role is to justify empirical pseudo-Lipschitz convergence
from weak Gaussian dependence, to control truncation errors, and to record the
Gaussian regression identity used in the memory decomposition.

\subsection{Pseudo-Lipschitz Growth and Truncation}

We first record a standard polynomial growth consequence of the
pseudo-Lipschitz condition.

\begin{lemma}[Polynomial growth of pseudo-Lipschitz functions]
	\label{lem:pl_growth_appendix}
	Let $\psi:\mathbb R^d\to\mathbb R$ be pseudo-Lipschitz of order $k\ge1$.
	Then there exists a constant $C_\psi<\infty$ such that
	\begin{equation*}
	|\psi(\boldsymbol x)|
	\le
	C_\psi(1+\|\boldsymbol x\|^k),
	\qquad
	\boldsymbol x\in\mathbb R^d .
	\end{equation*}
	Moreover, for every $\boldsymbol x,\boldsymbol y\in\mathbb R^d$,
	\begin{equation*}
	(1+\|\boldsymbol x+\boldsymbol y\|^k)
	\le
	C_k(1+\|\boldsymbol x\|^k+\|\boldsymbol y\|^k)
	\end{equation*}
	for a constant $C_k$ depending only on $k$.
\end{lemma}

\begin{IEEEproof}
	Taking $\boldsymbol y=\boldsymbol 0$ in the definition of pseudo-Lipschitz
	continuity gives
	\begin{equation*}
	|\psi(\boldsymbol x)|
	\le
	|\psi(\boldsymbol 0)|
	+
	L\|\boldsymbol x\|
	\left(1+\|\boldsymbol x\|^{k-1}\right).
	\end{equation*}
	Since $\|\boldsymbol x\|\le 1+\|\boldsymbol x\|^k$ for $k\ge1$, the first
	claim follows.
	
	For the second claim, use
	\begin{equation*}
	\|\boldsymbol x+\boldsymbol y\|
	\le
	\|\boldsymbol x\|+\|\boldsymbol y\|
	\end{equation*}
	and the convexity inequality
	\begin{equation*}
	(a+b)^k
	\le
	2^{k-1}(a^k+b^k),
	\qquad a,b\ge0.
	\end{equation*}
	Absorbing constants gives the result.
\end{IEEEproof}

The following truncation lemma is used to pass from bounded Lipschitz test
functions to pseudo-Lipschitz functions of order two.

\begin{lemma}[Uniform truncation for \(\PL(2)\) functions]
	\label{lem:pl_truncation_appendix}
	Let $\{\boldsymbol\xi_{j,N}:1\le j\le N\}$ be a triangular array of random
	vectors in $\mathbb R^d$ satisfying, for some $\epsilon>0$,
	\begin{equation*}
	\sup_{j,N}\mathbb E\|\boldsymbol\xi_{j,N}\|^{4+\epsilon}<\infty.
	\end{equation*}
	Let $\psi\in \PL(2)$.  Then
	\begin{equation*}
	\lim_{K\to\infty}
	\sup_N
	\frac1N\sum_{j=1}^N
	\mathbb E\left[
	|\psi(\boldsymbol\xi_{j,N})|
	\boldsymbol 1\{\|\boldsymbol\xi_{j,N}\|>K\}
	\right]
	=
	0.
	\end{equation*}
	Consequently, if $\psi_K$ is any sequence of bounded Lipschitz functions
	satisfying $\psi_K(\boldsymbol x)=\psi(\boldsymbol x)$ for
	$\|\boldsymbol x\|\le K$ and
	\begin{equation*}
	|\psi_K(\boldsymbol x)|\le C(1+\|\boldsymbol x\|^2)
	\end{equation*}
	uniformly in $K$, then
	\begin{equation*}
	\lim_{K\to\infty}
	\sup_N
	\frac1N\sum_{j=1}^N
	\mathbb E
	|\psi(\boldsymbol\xi_{j,N})-\psi_K(\boldsymbol\xi_{j,N})|
	=
	0.
	\end{equation*}
\end{lemma}

\begin{IEEEproof}
	By Lemma~\ref{lem:pl_growth_appendix}, there is $C_\psi<\infty$ such that
	\begin{equation*}
	|\psi(\boldsymbol x)|
	\le
	C_\psi(1+\|\boldsymbol x\|^2).
	\end{equation*}
	Fix $K>0$.  By Hölder's inequality, with
	\begin{equation*}
	r=\frac{4+\epsilon}{2}>1,
	\qquad
	r'=\frac{r}{r-1},
	\end{equation*}
	we have
	\begin{equation*}
	\begin{aligned}
		\mathbb E\left[
		(1+\|\boldsymbol\xi_{j,N}\|^2)
		\boldsymbol 1\{\|\boldsymbol\xi_{j,N}\|>K\}
		\right]
		&\le
		\left(
		\mathbb E(1+\|\boldsymbol\xi_{j,N}\|^2)^r
		\right)^{1/r}
		\mathbb P(\|\boldsymbol\xi_{j,N}\|>K)^{1/r'}.
	\end{aligned}
	\end{equation*}
	The first factor is uniformly bounded by the assumed
	$(4+\epsilon)$-moment condition.  The second factor converges to zero uniformly
	in $j,N$ by Markov's inequality:
	\begin{equation*}
	\mathbb P(\|\boldsymbol\xi_{j,N}\|>K)
	\le
	K^{-(4+\epsilon)}
	\sup_{j,N}\mathbb E\|\boldsymbol\xi_{j,N}\|^{4+\epsilon}.
	\end{equation*}
	This proves the first claim.  The second claim follows from
	\begin{equation*}
	|\psi(\boldsymbol x)-\psi_K(\boldsymbol x)|
	\le
	C(1+\|\boldsymbol x\|^2)\boldsymbol 1\{\|\boldsymbol x\|>K\}.
	\end{equation*}
\end{IEEEproof}

\subsection{Gaussian Covariance Comparison}

The next lemma is a convenient covariance comparison for smooth functions of
jointly Gaussian vectors.  It is the basic tool behind the weakly dependent
Gaussian empirical law.

\begin{lemma}[Gaussian covariance interpolation]
	\label{lem:gaussian_cov_interpolation}
	Let $(\boldsymbol\xi,\boldsymbol\eta)$ be a centered jointly Gaussian vector in
	$\mathbb R^d\times\mathbb R^d$.  Let
	\begin{equation*}
	\boldsymbol\Sigma_{\xi\eta}
	=
	\operatorname{Cov}(\boldsymbol\xi,\boldsymbol\eta).
	\end{equation*}
	If $f,g:\mathbb R^d\to\mathbb R$ are continuously differentiable with bounded
	gradients, then
	\begin{equation*}
	\left|
	\operatorname{Cov}(f(\boldsymbol\xi),g(\boldsymbol\eta))
	\right|
	\le
	\|\nabla f\|_\infty
	\|\nabla g\|_\infty
	\|\boldsymbol\Sigma_{\xi\eta}\|_{\mathrm{op}} .
	\end{equation*}
	More generally, if $\boldsymbol\xi,\boldsymbol\eta$ have dimensions
	$d_1,d_2$, then the right-hand side becomes
	\begin{equation*}
	\|\nabla f\|_\infty
	\|\nabla g\|_\infty
	\|\operatorname{Cov}(\boldsymbol\xi,\boldsymbol\eta)\|_{\mathrm{op}} .
	\end{equation*}
\end{lemma}

\begin{IEEEproof}
	We give the proof for equal dimensions; the rectangular case is identical.
	Let $(\boldsymbol\xi_0,\boldsymbol\eta_0)$ be a Gaussian pair with the same
	marginal distributions as $(\boldsymbol\xi,\boldsymbol\eta)$ but with
	$\boldsymbol\xi_0$ independent of $\boldsymbol\eta_0$.  For $s\in[0,1]$, let
	$(\boldsymbol\xi_s,\boldsymbol\eta_s)$ be a centered Gaussian pair with the
	same marginal covariances and cross-covariance
	\begin{equation*}
	\operatorname{Cov}(\boldsymbol\xi_s,\boldsymbol\eta_s)
	=
	s\boldsymbol\Sigma_{\xi\eta}.
	\end{equation*}
	Define
	\begin{equation*}
	\Phi(s)=\mathbb E[f(\boldsymbol\xi_s)g(\boldsymbol\eta_s)].
	\end{equation*}
	The Gaussian interpolation identity gives
	\begin{equation*}
	\Phi'(s)
	=
	\mathbb E\left[
	\nabla f(\boldsymbol\xi_s)^{\mathsf T}
	\boldsymbol\Sigma_{\xi\eta}
	\nabla g(\boldsymbol\eta_s)
	\right].
	\end{equation*}
	Since
	\begin{equation*}
	\Phi(1)-\Phi(0)
	=
	\operatorname{Cov}(f(\boldsymbol\xi),g(\boldsymbol\eta)),
	\end{equation*}
	we obtain
	\begin{equation*}
	\begin{aligned}
		\left|
		\operatorname{Cov}(f(\boldsymbol\xi),g(\boldsymbol\eta))
		\right|
		&\le
		\int_0^1
		\mathbb E
		\left|
		\nabla f(\boldsymbol\xi_s)^{\mathsf T}
		\boldsymbol\Sigma_{\xi\eta}
		\nabla g(\boldsymbol\eta_s)
		\right|ds                                      \\
		&\le
		\|\nabla f\|_\infty
		\|\nabla g\|_\infty
		\|\boldsymbol\Sigma_{\xi\eta}\|_{\mathrm{op}}.
	\end{aligned}
	\end{equation*}
\end{IEEEproof}

\subsection{Weakly Dependent Gaussian Empirical Law}

We now prove the empirical convergence result used after the empirical Gaussian
decoupling argument.

\begin{lemma}[Weakly dependent Gaussian empirical law]
	\label{lem:weak_gaussian_empirical_appendix}
	Let
	\begin{equation*}
	\{\boldsymbol\xi_{j,N}\in\mathbb R^d:1\le j\le N\}
	\end{equation*}
	be a triangular array of jointly Gaussian random vectors.  Suppose that, for
	some $\epsilon>0$,
	\begin{equation*}
	\sup_{j,N}\mathbb E\|\boldsymbol\xi_{j,N}\|^{4+\epsilon}<\infty,
	\end{equation*}
	and
	\begin{equation}
	\frac1{N^2}
	\sum_{j\ne k}
	\left\|
	\operatorname{Cov}
	(\boldsymbol\xi_{j,N},\boldsymbol\xi_{k,N})
	\right\|_{\mathrm{op}}
	\longrightarrow 0.
	\label{eq:appendix_offdiag_cov}
	\end{equation}
	Then, for every $\psi\in \PL(2)$,
	\begin{equation*}
	\frac1N\sum_{j=1}^N
	\psi(\boldsymbol\xi_{j,N})
	-
	\frac1N\sum_{j=1}^N
	\mathbb E\psi(\boldsymbol\xi_{j,N})
	\overset{p}{\longrightarrow}0.
	\end{equation*}
\end{lemma}

\begin{IEEEproof}
	We first prove the result for bounded continuously differentiable functions
	with bounded gradient.  Set
	\begin{equation*}
	Y_{j,N}
	=
	\psi(\boldsymbol\xi_{j,N})
	-
	\mathbb E\psi(\boldsymbol\xi_{j,N}).
	\end{equation*}
	Then
	\begin{equation*}
	\begin{aligned}
		\operatorname{Var}\left(
		\frac1N\sum_{j=1}^N\psi(\boldsymbol\xi_{j,N})
		\right)
		&=
		\frac1{N^2}\sum_{j=1}^N
		\operatorname{Var}(\psi(\boldsymbol\xi_{j,N}))        \\
		&\quad+
		\frac1{N^2}\sum_{j\ne k}
		\operatorname{Cov}
		\big(
		\psi(\boldsymbol\xi_{j,N}),
		\psi(\boldsymbol\xi_{k,N})
		\big).
	\end{aligned}
	\end{equation*}
	The diagonal term is $O(N^{-1})$ because $\psi$ is bounded.  For the off-diagonal
	term, Lemma~\ref{lem:gaussian_cov_interpolation} yields
	\begin{equation*}
	\left|
	\operatorname{Cov}
	\big(
	\psi(\boldsymbol\xi_{j,N}),
	\psi(\boldsymbol\xi_{k,N})
	\big)
	\right|
	\le
	\|\nabla\psi\|_\infty^2
	\left\|
	\operatorname{Cov}
	(\boldsymbol\xi_{j,N},\boldsymbol\xi_{k,N})
	\right\|_{\mathrm{op}}.
	\end{equation*}
	Therefore \eqref{eq:appendix_offdiag_cov} implies that the off-diagonal
	contribution vanishes.  Hence the variance of the empirical average converges
	to zero, and Chebyshev's inequality gives convergence in probability.
	
	The same conclusion holds for bounded Lipschitz functions.  Indeed, every
	bounded Lipschitz function on $\mathbb R^d$ can be approximated uniformly by
	bounded smooth functions with gradients bounded by a constant depending only on
	the Lipschitz constant.  Applying the previous paragraph to the smooth
	approximants and letting the approximation error vanish proves the bounded
	Lipschitz case.
	
	Let now $\psi\in \PL(2)$.  Choose a smooth cutoff $\chi_K:\mathbb R^d\to[0,1]$
	such that $\chi_K(\boldsymbol x)=1$ when $\|\boldsymbol x\|\le K$ and
	$\chi_K(\boldsymbol x)=0$ when $\|\boldsymbol x\|\ge 2K$, and define
	\begin{equation*}
	\psi_K(\boldsymbol x)=\psi(\boldsymbol x)\chi_K(\boldsymbol x).
	\end{equation*}
	For fixed $K$, $\psi_K$ is bounded Lipschitz, so
	\begin{equation*}
	\frac1N\sum_{j=1}^N
	\psi_K(\boldsymbol\xi_{j,N})
	-
	\frac1N\sum_{j=1}^N
	\mathbb E\psi_K(\boldsymbol\xi_{j,N})
	\overset{p}{\longrightarrow}0.
	\end{equation*}
	By Lemma~\ref{lem:pl_truncation_appendix},
	\begin{equation*}
	\lim_{K\to\infty}
	\sup_N
	\frac1N\sum_{j=1}^N
	\mathbb E
	|\psi(\boldsymbol\xi_{j,N})-\psi_K(\boldsymbol\xi_{j,N})|
	=
	0.
	\end{equation*}
	Markov's inequality gives the corresponding convergence of the empirical
	truncation error in probability.  Therefore, letting first $N\to\infty$ and
	then $K\to\infty$, the desired convergence follows.
\end{IEEEproof}

\subsection{Gaussian Regression with a Singular Covariance}

The memory decomposition uses Gaussian regression onto a finite past history.
Since some past directions may be exactly redundant, the covariance matrix need
not be invertible.  The Moore--Penrose inverse removes those redundant
directions.

\begin{lemma}[Gaussian regression with Moore--Penrose inverse]
	\label{lem:gaussian_regression_appendix}
	Let
	\begin{equation*}
	\begin{pmatrix}
		\boldsymbol Z\\
		Z
	\end{pmatrix}
	\end{equation*}
	be a centered jointly Gaussian vector with
	\begin{equation*}
	\operatorname{Cov}(\boldsymbol Z)=\boldsymbol\Sigma,
	\qquad
	\operatorname{Cov}(Z,\boldsymbol Z)=\boldsymbol c.
	\end{equation*}
	Then
	\begin{equation*}
	Z
	=
	\boldsymbol c\boldsymbol\Sigma^\dagger\boldsymbol Z
	+
	G,
	\end{equation*}
	where $G$ is centered Gaussian and independent of $\boldsymbol Z$.  Moreover,
	\begin{equation*}
	\operatorname{Var}(G)
	=
	\operatorname{Var}(Z)
	-
	\boldsymbol c\boldsymbol\Sigma^\dagger\boldsymbol c^{\mathsf T}.
	\end{equation*}
\end{lemma}

\begin{IEEEproof}
	Define
	\begin{equation*}
	G
	=
	Z-\boldsymbol c\boldsymbol\Sigma^\dagger\boldsymbol Z.
	\end{equation*}
	Since $(Z,\boldsymbol Z)$ is jointly Gaussian, $(G,\boldsymbol Z)$ is also
	jointly Gaussian.  It is therefore enough to show that
	\begin{equation*}
	\operatorname{Cov}(G,\boldsymbol Z)=\boldsymbol 0.
	\end{equation*}
	We compute
	\begin{equation*}
	\begin{aligned}
		\operatorname{Cov}(G,\boldsymbol Z)
		&=
		\boldsymbol c
		-
		\boldsymbol c\boldsymbol\Sigma^\dagger
		\operatorname{Cov}(\boldsymbol Z,\boldsymbol Z)       \\
		&=
		\boldsymbol c
		-
		\boldsymbol c\boldsymbol\Sigma^\dagger\boldsymbol\Sigma.
	\end{aligned}
	\end{equation*}
	For a valid covariance matrix of $(Z,\boldsymbol Z)$, the row vector
	$\boldsymbol c$ lies in the row space of $\boldsymbol\Sigma$.  To see this, if
	$\boldsymbol v\in\ker(\boldsymbol\Sigma)$, then
	\begin{equation*}
	\operatorname{Var}(\boldsymbol v^{\mathsf T}\boldsymbol Z)
	=
	\boldsymbol v^{\mathsf T}\boldsymbol\Sigma\boldsymbol v
	=
	0,
	\end{equation*}
	so $\boldsymbol v^{\mathsf T}\boldsymbol Z=0$ almost surely.  Hence
	\begin{equation*}
	\operatorname{Cov}(Z,\boldsymbol v^{\mathsf T}\boldsymbol Z)=0,
	\end{equation*}
	which gives $\boldsymbol c\boldsymbol v=0$.  Thus
	$\boldsymbol c$ is orthogonal to $\ker(\boldsymbol\Sigma)$, i.e., it lies in
	the row space of $\boldsymbol\Sigma$.  Therefore
	\begin{equation*}
	\boldsymbol c\boldsymbol\Sigma^\dagger\boldsymbol\Sigma
	=
	\boldsymbol c,
	\end{equation*}
	and $\operatorname{Cov}(G,\boldsymbol Z)=0$.  Since $(G,\boldsymbol Z)$ is
	jointly Gaussian, zero covariance implies independence.
	
	Finally,
	\begin{equation*}
	\begin{aligned}
		\operatorname{Var}(G)
		&=
		\operatorname{Var}(Z)
		-
		2\boldsymbol c\boldsymbol\Sigma^\dagger\boldsymbol c^{\mathsf T}
		+
		\boldsymbol c\boldsymbol\Sigma^\dagger
		\boldsymbol\Sigma
		\boldsymbol\Sigma^\dagger
		\boldsymbol c^{\mathsf T}                         \\
		&=
		\operatorname{Var}(Z)
		-
		\boldsymbol c\boldsymbol\Sigma^\dagger
		\boldsymbol c^{\mathsf T},
	\end{aligned}
	\end{equation*}
	because $\boldsymbol\Sigma^\dagger\boldsymbol\Sigma\boldsymbol\Sigma^\dagger
	=\boldsymbol\Sigma^\dagger$.
\end{IEEEproof}

\subsection{Gaussian Maximum and No-Spike Bounds}

We finally record a simple Gaussian maximum bound used to verify the no-spike
part of admissibility.

\begin{lemma}[Gaussian maximum bound]
	\label{lem:gaussian_max_appendix}
	Let $\boldsymbol g_N=(g_{1,N},\ldots,g_{N,N})^{\mathsf T}$ be a centered
	Gaussian vector, not necessarily with independent entries.  Suppose that
	\begin{equation*}
	\max_{1\le j\le N}\operatorname{Var}(g_{j,N})\le C
	\end{equation*}
	uniformly in $N$.  Then
	\begin{equation*}
	\frac{\|\boldsymbol g_N\|_\infty}{\sqrt N}
	\overset{p}{\longrightarrow}0.
	\end{equation*}
	If, in addition,
	\begin{equation*}
	\frac1N\sum_{j=1}^N\operatorname{Var}(g_{j,N})=O(1),
	\end{equation*}
	then
	\begin{equation*}
	\frac1N\|\boldsymbol g_N\|^2=O_p(1).
	\end{equation*}
\end{lemma}

\begin{IEEEproof}
	For any $\epsilon>0$, the union bound and the Gaussian tail inequality give
	\begin{equation*}
	\begin{aligned}
		\mathbb P(\|\boldsymbol g_N\|_\infty>\epsilon\sqrt N)
		&\le
		\sum_{j=1}^N
		\mathbb P(|g_{j,N}|>\epsilon\sqrt N)                \\
		&\le
		2N\exp\left(-\frac{\epsilon^2N}{2C}\right),
	\end{aligned}
	\end{equation*}
	which converges to zero.  This proves the no-spike claim.
	
	For the energy bound,
	\begin{equation*}
	\mathbb E\left[\frac1N\|\boldsymbol g_N\|^2\right]
	=
	\frac1N\sum_{j=1}^N\operatorname{Var}(g_{j,N})
	=
	O(1).
	\end{equation*}
	Markov's inequality then gives
	\begin{equation*}
	\frac1N\|\boldsymbol g_N\|^2=O_p(1).
	\end{equation*}
\end{IEEEproof}

\begin{lemma}[Stability of empirical \(\PL(2)\) averages under
	\(o_p^{\ell_2}\) perturbations]
	\label{lem:pl_stability_l2_appendix}
	Let $\{\boldsymbol\xi_{j,N}\}_{j=1}^N$ and
	$\{\boldsymbol\eta_{j,N}\}_{j=1}^N$ be arrays in $\mathbb R^d$ such that
	\begin{equation*}
	\frac1N\sum_{j=1}^N
	\|\boldsymbol\xi_{j,N}-\boldsymbol\eta_{j,N}\|^2
	\overset{p}{\longrightarrow}0.
	\end{equation*}
	Assume further that, for some $\epsilon>0$,
	\begin{equation*}
	\sup_N
	\frac1N\sum_{j=1}^N
	\mathbb E\|\boldsymbol\xi_{j,N}\|^{4+\epsilon}<\infty,
	\qquad
	\sup_N
	\frac1N\sum_{j=1}^N
	\mathbb E\|\boldsymbol\eta_{j,N}\|^{4+\epsilon}<\infty.
	\end{equation*}
	Then, for every $\psi\in \PL(2)$,
	\begin{equation*}
	\frac1N\sum_{j=1}^N
	\left|
	\psi(\boldsymbol\xi_{j,N})
	-
	\psi(\boldsymbol\eta_{j,N})
	\right|
	\overset{p}{\longrightarrow}0.
	\end{equation*}
\end{lemma}

\begin{IEEEproof}
	By the pseudo-Lipschitz property,
	\begin{equation*}
	\begin{aligned}
		&
		\frac1N\sum_{j=1}^N
		\left|
		\psi(\boldsymbol\xi_{j,N})
		-
		\psi(\boldsymbol\eta_{j,N})
		\right|                                          \\
		&\le
		L
		\left(
		\frac1N\sum_{j=1}^N
		\|\boldsymbol\xi_{j,N}-\boldsymbol\eta_{j,N}\|^2
		\right)^{1/2}
		\left(
		\frac1N\sum_{j=1}^N
		(1+\|\boldsymbol\xi_{j,N}\|+\|\boldsymbol\eta_{j,N}\|)^2
		\right)^{1/2},
	\end{aligned}
	\end{equation*}
	where we used Cauchy--Schwarz.  The first factor converges to zero in
	probability by assumption.  The second factor is tight by the uniform moment
	bounds.  The product therefore converges to zero in probability.
\end{IEEEproof}

\section{Gaussian Conditioning under EP Histories}
\label{app:gaussian_conditioning_histories}

This appendix proves the conditional Gaussian representation used in
Proposition~\ref{prop:conditional_gaussian_history}.  The proof is purely
Gaussian.  The only algorithm-specific input is that the finite EP history
imposes the two linear constraints
\begin{equation*}
\boldsymbol A\boldsymbol M_t
=
\boldsymbol w\boldsymbol 1_t^{\mathsf T}
-
\boldsymbol U_t
\end{equation*}
and
\begin{equation*}
\boldsymbol A^{\mathsf T}\boldsymbol U_t
=
\gamma_w^{-1}
\left(
\boldsymbol\Gamma_0(\boldsymbol m_0-\boldsymbol q_0),
\ldots,
\boldsymbol\Gamma_{t-1}(\boldsymbol m_{t-1}-\boldsymbol q_{t-1})
\right).
\end{equation*}
Once the finite history is fixed, these are linear observations of the
variance-profile Gaussian matrix.  Conditioning a Gaussian vector on finitely
many linear observations gives another Gaussian vector.  The covariance
contraction of this conditional Gaussian law is the source of the flat
covariance bound used later in the MDE input.

\subsection{Linear Form of the EP History}
\label{subsec:linear_history_operator}

Let
\begin{equation*}
\boldsymbol A\in\mathbb R^{M\times N}
\end{equation*}
be a variance-profile Gaussian matrix,
\begin{equation*}
A_{ij}
=
\sqrt{\frac{s_{ij}}{M}}Z_{ij},
\qquad
Z_{ij}\overset{\mathrm{i.i.d.}}{\sim}\mathcal N(0,1).
\end{equation*}
Define the vectorized matrix
\begin{equation*}
\boldsymbol a
=
\operatorname{vec}(\boldsymbol A)
\in\mathbb R^{MN}.
\end{equation*}
Then
\begin{equation*}
\boldsymbol a
\sim
\mathcal N(\boldsymbol 0,\boldsymbol\Sigma_A),
\end{equation*}
where
\begin{equation*}
\boldsymbol\Sigma_A
=
\operatorname{diag}
\left(
\frac{s_{ij}}{M}
\right)_{1\le i\le M,\ 1\le j\le N}.
\end{equation*}
The uniformly elliptic profile assumption implies
\begin{equation*}
\frac{s_{\min}}{M}\boldsymbol I_{MN}
\preceq
\boldsymbol\Sigma_A
\preceq
\frac{s_{\max}}{M}\boldsymbol I_{MN}.
\end{equation*}

Fix an iteration index \(t\).  The past signal-side and measurement-side
histories are
\begin{equation*}
\boldsymbol M_t
=
(\boldsymbol m_0,\ldots,\boldsymbol m_{t-1})
\in\mathbb R^{N\times t},
\end{equation*}
and
\begin{equation*}
\boldsymbol U_t
=
(\boldsymbol u_0,\ldots,\boldsymbol u_{t-1})
\in\mathbb R^{M\times t}.
\end{equation*}
The first history constraint is
\begin{equation}
\boldsymbol A\boldsymbol M_t
=
\boldsymbol R_t,
\qquad
\boldsymbol R_t
:=
\boldsymbol w\boldsymbol 1_t^{\mathsf T}
-
\boldsymbol U_t
\in\mathbb R^{M\times t}.
\label{eq:appendix_history_R}
\end{equation}
The second history constraint is
\begin{equation}
\boldsymbol A^{\mathsf T}\boldsymbol U_t
=
\boldsymbol B_t,
\label{eq:appendix_history_B}
\end{equation}
where
\begin{equation*}
\boldsymbol B_t
:=
\gamma_w^{-1}
\left(
\boldsymbol\Gamma_0(\boldsymbol m_0-\boldsymbol q_0),
\ldots,
\boldsymbol\Gamma_{t-1}(\boldsymbol m_{t-1}-\boldsymbol q_{t-1})
\right)
\in\mathbb R^{N\times t}.
\end{equation*}
Both mappings
\begin{equation*}
\boldsymbol A\mapsto \boldsymbol A\boldsymbol M_t,
\qquad
\boldsymbol A\mapsto \boldsymbol A^{\mathsf T}\boldsymbol U_t
\end{equation*}
are linear in \(\boldsymbol A\).  Therefore, for fixed
\(\boldsymbol M_t,\boldsymbol U_t,\boldsymbol R_t,\boldsymbol B_t\), there is
a deterministic matrix \(\boldsymbol L_t\) and a vector \(\boldsymbol b_t\) such
that the two constraints can be written compactly as
\begin{equation}
	\boldsymbol L_t\boldsymbol a
	=
	\boldsymbol b_t.
	\label{eq:linear_history_vectorized}
\end{equation}
For example, the first constraint contributes
\begin{equation*}
\operatorname{vec}(\boldsymbol A\boldsymbol M_t)
=
(\boldsymbol M_t^{\mathsf T}\otimes \boldsymbol I_M)
\operatorname{vec}(\boldsymbol A),
\end{equation*}
and the second constraint is represented similarly after applying the
commutation matrix relating \(\operatorname{vec}(\boldsymbol A)\) and
\(\operatorname{vec}(\boldsymbol A^{\mathsf T})\).  The explicit Kronecker form
is not important; what matters is the linearity of
\eqref{eq:linear_history_vectorized}.

We shall condition on a regular finite history satisfying
\eqref{eq:appendix_history_R}--\eqref{eq:appendix_history_B}.  Redundant
constraints are allowed.  They are handled by Moore--Penrose inverses below.

\subsection{Conditional Gaussian Decomposition}
\label{subsec:conditional_gaussian_decomposition}

We now compute the conditional law of \(\boldsymbol a\) given
\(\boldsymbol L_t\boldsymbol a=\boldsymbol b_t\).

\begin{lemma}[Gaussian conditioning under the EP history]
	\label{lem:appendix_gaussian_conditioning}
	Conditionally on the finite linear history
	\begin{equation*}
	\boldsymbol L_t\boldsymbol a=\boldsymbol b_t,
	\end{equation*}
	the vector \(\boldsymbol a\) admits the decomposition
	\begin{equation*}
	\boldsymbol a
	=
	\boldsymbol a_{\parallel,t}
	+
	\boldsymbol a_{\perp,t},
	\end{equation*}
	where
	\begin{equation*}
		\boldsymbol a_{\parallel,t}
		=
		\boldsymbol\Sigma_A\boldsymbol L_t^{\mathsf T}
		(\boldsymbol L_t\boldsymbol\Sigma_A\boldsymbol L_t^{\mathsf T})^\dagger
		\boldsymbol b_t
	\end{equation*}
	is the conditional mean, and
	\begin{equation*}
	\boldsymbol a_{\perp,t}
	\sim
	\mathcal N(\boldsymbol 0,\boldsymbol\Sigma_{A|t})
	\end{equation*}
	with conditional covariance
	\begin{equation}
		\boldsymbol\Sigma_{A|t}
		=
		\boldsymbol\Sigma_A
		-
		\boldsymbol\Sigma_A\boldsymbol L_t^{\mathsf T}
		(\boldsymbol L_t\boldsymbol\Sigma_A\boldsymbol L_t^{\mathsf T})^\dagger
		\boldsymbol L_t\boldsymbol\Sigma_A.
		\label{eq:conditional_covariance_formula}
	\end{equation}
	Moreover,
	\begin{equation}
		\boldsymbol 0
		\preceq
		\boldsymbol\Sigma_{A|t}
		\preceq
		\boldsymbol\Sigma_A.
		\label{eq:conditional_covariance_contraction}
	\end{equation}
	Finally,
	\begin{equation*}
	\boldsymbol L_t\boldsymbol a_{\perp,t}
	=
	\boldsymbol 0
	\quad
	\text{almost surely}.
	\end{equation*}
\end{lemma}

\begin{IEEEproof}
	Since \(\boldsymbol a\) is Gaussian and
	\(\boldsymbol L_t\boldsymbol a\) is a linear transform of
	\(\boldsymbol a\), the pair
	\begin{equation*}
	(\boldsymbol a,\boldsymbol L_t\boldsymbol a)
	\end{equation*}
	is jointly Gaussian.  The covariance matrices are
	\begin{equation*}
	\operatorname{Cov}(\boldsymbol a,\boldsymbol L_t\boldsymbol a)
	=
	\boldsymbol\Sigma_A\boldsymbol L_t^{\mathsf T},
	\end{equation*}
	and
	\begin{equation*}
	\operatorname{Cov}(\boldsymbol L_t\boldsymbol a)
	=
	\boldsymbol L_t\boldsymbol\Sigma_A\boldsymbol L_t^{\mathsf T}.
	\end{equation*}
	The standard conditional Gaussian formula, with the Moore--Penrose inverse used
	to allow redundant constraints, gives the conditional mean
	\begin{equation*}
	\mathbb E[
	\boldsymbol a
	\mid
	\boldsymbol L_t\boldsymbol a=\boldsymbol b_t
	]
	=
	\boldsymbol\Sigma_A\boldsymbol L_t^{\mathsf T}
	(\boldsymbol L_t\boldsymbol\Sigma_A\boldsymbol L_t^{\mathsf T})^\dagger
	\boldsymbol b_t
	\end{equation*}
	and the conditional covariance \eqref{eq:conditional_covariance_formula}.
	This proves the Gaussian decomposition.
	
	It remains to prove the covariance contraction.  Define
	\begin{equation*}
	\boldsymbol C_t
	=
	\boldsymbol L_t\boldsymbol\Sigma_A^{1/2}.
	\end{equation*}
	Then
	\begin{equation*}
	\begin{aligned}
		\boldsymbol\Sigma_{A|t}
		&=
		\boldsymbol\Sigma_A^{1/2}
		\left[
		\boldsymbol I
		-
		\boldsymbol C_t^{\mathsf T}
		(\boldsymbol C_t\boldsymbol C_t^{\mathsf T})^\dagger
		\boldsymbol C_t
		\right]
		\boldsymbol\Sigma_A^{1/2}.
	\end{aligned}
	\end{equation*}
	The matrix
	\begin{equation*}
	\boldsymbol P_t
	=
	\boldsymbol C_t^{\mathsf T}
	(\boldsymbol C_t\boldsymbol C_t^{\mathsf T})^\dagger
	\boldsymbol C_t
	\end{equation*}
	is the orthogonal projector onto the row space of \(\boldsymbol C_t\).
	Therefore
	\begin{equation*}
	\boldsymbol 0\preceq \boldsymbol P_t\preceq \boldsymbol I,
	\qquad
	\boldsymbol 0\preceq \boldsymbol I-\boldsymbol P_t\preceq \boldsymbol I,
	\end{equation*}
	which implies
	\begin{equation*}
	\boldsymbol 0
	\preceq
	\boldsymbol\Sigma_{A|t}
	\preceq
	\boldsymbol\Sigma_A.
	\end{equation*}
	
	Finally, since the conditional law is supported on the affine subspace
	\begin{equation*}
	\{\boldsymbol a:\boldsymbol L_t\boldsymbol a=\boldsymbol b_t\},
	\end{equation*}
	the residual
	\begin{equation*}
	\boldsymbol a_{\perp,t}
	=
	\boldsymbol a-\boldsymbol a_{\parallel,t}
	\end{equation*}
	lies in the corresponding homogeneous subspace:
	\begin{equation*}
	\boldsymbol L_t\boldsymbol a_{\perp,t}=\boldsymbol 0
	\end{equation*}
	almost surely.  This can also be checked algebraically from
	\eqref{eq:conditional_covariance_formula}, since
	\begin{equation*}
	\boldsymbol L_t\boldsymbol\Sigma_{A|t}\boldsymbol L_t^{\mathsf T}
	=
	\boldsymbol 0.
	\end{equation*}
\end{IEEEproof}

Returning to matrix notation, let
\begin{equation*}
\boldsymbol A_{\parallel,t}
=
\operatorname{mat}(\boldsymbol a_{\parallel,t}),
\qquad
\boldsymbol A_{\perp,t}
=
\operatorname{mat}(\boldsymbol a_{\perp,t}),
\end{equation*}
where \(\operatorname{mat}\) is the inverse of the chosen vectorization.  Then
Lemma~\ref{lem:appendix_gaussian_conditioning} gives
\begin{equation}
	\boldsymbol A
	=
	\boldsymbol A_{\parallel,t}
	+
	\boldsymbol A_{\perp,t}.
	\label{eq:A_decomp_appendix}
\end{equation}
Furthermore, the residual matrix satisfies the homogeneous history constraints
\begin{equation*}
\boldsymbol A_{\perp,t}\boldsymbol M_t
=
\boldsymbol 0,
\qquad
\boldsymbol A_{\perp,t}^{\mathsf T}\boldsymbol U_t
=
\boldsymbol 0
\end{equation*}
almost surely under the conditional law.

\subsection{Flat Covariance of the Conditioned Residual}
\label{subsec:flat_cov_conditioned_residual}

We next prove the flat covariance estimate used in the MDE input.  The bound
is a direct consequence of the covariance contraction
\eqref{eq:conditional_covariance_contraction}.

\begin{lemma}[Flat covariance of the conditioned residual]
	\label{lem:appendix_flat_covariance}
	Let \(\boldsymbol A_{\perp,t}\) be the centered Gaussian residual in
	\eqref{eq:A_decomp_appendix}.  For any conditioning-measurable vectors
	\(\boldsymbol p\in\mathbb R^M\) and
	\(\boldsymbol q\in\mathbb R^N\),
	\begin{equation}
		\operatorname{Var}
		\left(
		\boldsymbol p^{\mathsf T}
		\boldsymbol A_{\perp,t}
		\boldsymbol q
		\,\middle|\,
		\mathcal F_t
		\right)
		\le
		\frac{s_{\max}}{M}
		\|\boldsymbol p\|^2\|\boldsymbol q\|^2.
		\label{eq:flat_covariance_Aperp}
	\end{equation}
	Consequently, if
	\begin{equation*}
	\boldsymbol X_t
	=
	\sqrt{\gamma_w}\boldsymbol A_{\perp,t},
	\end{equation*}
	then
	\begin{equation}
		\operatorname{Var}
		\left(
		\boldsymbol p^{\mathsf T}
		\boldsymbol X_t
		\boldsymbol q
		\,\middle|\,
		\mathcal F_t
		\right)
		\le
		\frac{\gamma_w s_{\max}}{M}
		\|\boldsymbol p\|^2\|\boldsymbol q\|^2.
		\label{eq:flat_covariance_X}
	\end{equation}
\end{lemma}

\begin{IEEEproof}
	Let
	\begin{equation*}
	\boldsymbol h
	=
	\operatorname{vec}(\boldsymbol p\boldsymbol q^{\mathsf T}).
	\end{equation*}
	With the convention
	\begin{equation*}
	\boldsymbol p^{\mathsf T}\boldsymbol A_{\perp,t}\boldsymbol q
	=
	\boldsymbol h^{\mathsf T}\operatorname{vec}(\boldsymbol A_{\perp,t}),
	\end{equation*}
	we have
	\begin{equation*}
	\operatorname{Var}
	\left(
	\boldsymbol p^{\mathsf T}
	\boldsymbol A_{\perp,t}
	\boldsymbol q
	\,\middle|\,
	\mathcal F_t
	\right)
	=
	\boldsymbol h^{\mathsf T}
	\boldsymbol\Sigma_{A|t}
	\boldsymbol h.
	\end{equation*}
	Using
	\begin{equation*}
	\boldsymbol\Sigma_{A|t}\preceq\boldsymbol\Sigma_A,
	\end{equation*}
	we get
	\begin{equation*}
	\begin{aligned}
		\boldsymbol h^{\mathsf T}
		\boldsymbol\Sigma_{A|t}
		\boldsymbol h
		&\le
		\boldsymbol h^{\mathsf T}
		\boldsymbol\Sigma_A
		\boldsymbol h                                      \\
		&=
		\sum_{i=1}^M\sum_{j=1}^N
		\frac{s_{ij}}{M}p_i^2q_j^2                         \\
		&\le
		\frac{s_{\max}}{M}
		\left(\sum_{i=1}^M p_i^2\right)
		\left(\sum_{j=1}^N q_j^2\right)                    \\
		&=
		\frac{s_{\max}}{M}\|\boldsymbol p\|^2\|\boldsymbol q\|^2.
	\end{aligned}
	\end{equation*}
	This proves \eqref{eq:flat_covariance_Aperp}.  Multiplying the residual matrix
	by \(\sqrt{\gamma_w}\) gives \eqref{eq:flat_covariance_X}.
\end{IEEEproof}

\subsection{Conditional Entrywise Covariance and Profile Bounds}
\label{subsec:conditional_entrywise_covariance}

For later use in the covariance-kernel computation, we also record the
entrywise form of the conditional covariance.  Define
\begin{equation*}
\kappa_t(ia,kb)
=
\operatorname{Cov}
\left(
\sqrt{\gamma_w} A_{\perp,t,ia},
\sqrt{\gamma_w} A_{\perp,t,kb}
\,\middle|\,
\mathcal F_t
\right),
\qquad
1\le i,k\le M,\quad 1\le a,b\le N.
\end{equation*}
Then \(\kappa_t\) is the covariance kernel used in the conditioned Gaussian
linearization.  The covariance contraction implies, in particular,
\begin{equation*}
\operatorname{Var}
\left(
A_{\perp,t,ia}
\mid
\mathcal F_t
\right)
\le
\frac{s_{ia}}{M}
\le
\frac{s_{\max}}{M}.
\end{equation*}
Therefore
\begin{equation*}
|\kappa_t(ia,kb)|
\le
\frac{\gamma_w s_{\max}}{M}
\end{equation*}
whenever the covariance is evaluated through Cauchy--Schwarz.  More generally,
for deterministic or conditioning-measurable arrays
\(\boldsymbol R=(R_{ia})\) and \(\boldsymbol T=(T_{ia})\),
\begin{equation*}
\operatorname{Cov}
\left(
\sum_{i,a}R_{ia}\sqrt{\gamma_w}A_{\perp,t,ia},
\sum_{k,b}T_{kb}\sqrt{\gamma_w}A_{\perp,t,kb}
\,\middle|\,
\mathcal F_t
\right)
\end{equation*}
is controlled by the flat covariance estimate
\eqref{eq:flat_covariance_X} whenever the arrays are rank-one,
\(R_{ia}=p_iq_a\), \(T_{ia}=p'_iq'_a\).  This is precisely the class of
bilinear forms needed in the MDE input and in the conditioned covariance
calculation.

\subsection{Conclusion of Proposition~\ref{prop:conditional_gaussian_history}}
\label{subsec:conclude_conditional_gaussian_proposition}

Combining Lemmas~\ref{lem:appendix_gaussian_conditioning} and
\ref{lem:appendix_flat_covariance} proves
Proposition~\ref{prop:conditional_gaussian_history}.  Indeed, the conditional
decomposition
\begin{equation*}
\boldsymbol A
=
\boldsymbol A_{\parallel,t}
+
\boldsymbol A_{\perp,t}
\end{equation*}
follows from Lemma~\ref{lem:appendix_gaussian_conditioning}, with
\(\boldsymbol A_{\parallel,t}\) measurable with respect to the finite EP
history and \(\boldsymbol A_{\perp,t}\) centered Gaussian conditionally on that
history.  The flat covariance bound
\begin{equation*}
\operatorname{Var}
\left(
\boldsymbol p^{\mathsf T}
\boldsymbol A_{\perp,t}
\boldsymbol q
\,\middle|\,
\mathcal F_t
\right)
\le
\frac{C}{M}\|\boldsymbol p\|^2\|\boldsymbol q\|^2
\end{equation*}
holds with \(C=s_{\max}\).  For the scaled residual
\(\boldsymbol X_t=\sqrt{\gamma_w}\boldsymbol A_{\perp,t}\), the same estimate
holds with \(C=\gamma_w s_{\max}\).  This is the form used by the
regularized correlated-Gaussian MDE input in the main proof.

\section{Bounded Deformation and MDE Input Verification}
\label{app:bounded_deformation_mde}

This appendix proves Proposition~\ref{prop:bounded_deformation_mde}.  The
regularized correlated-Gaussian MDE input itself is stated in
Lemma~\ref{lem:regularized_mde_input}.  The purpose of the present appendix is
only to verify that the random matrix obtained after conditioning on the finite
EP history satisfies the hypotheses of that input.

The proof has two parts.  First, we show that the conditional mean
\(\boldsymbol A_{\parallel,t}\) is a bounded deformation:
\begin{equation*}
\|\boldsymbol A_{\parallel,t}\|=O_p(1).
\end{equation*}
Second, we check that the residual block matrix is centered Gaussian with flat
covariance and that the block loading is uniformly regularized.  These facts
allow Lemma~\ref{lem:regularized_mde_input} to be applied conditionally on the
EP history.

\subsection{Conditioned Block Linearization}
\label{subsec:C_conditioned_linearization}

By Appendix~\ref{app:gaussian_conditioning_histories}, conditioning on the
finite EP history gives
\begin{equation}
	\boldsymbol A
	=
	\boldsymbol A_{\parallel,t}
	+
	\boldsymbol A_{\perp,t},
	\label{eq:C_A_decomp}
\end{equation}
where
\begin{equation*}
\boldsymbol A_{\parallel,t}
=
\mathbb E[\boldsymbol A\mid\mathcal F_t]
\end{equation*}
and \(\boldsymbol A_{\perp,t}\) is centered Gaussian conditionally on
\(\mathcal F_t\).  Moreover, for all conditioning-measurable
\(\boldsymbol p\in\mathbb R^M\) and
\(\boldsymbol q\in\mathbb R^N\),
\begin{equation*}
	\operatorname{Var}
	\left(
	\boldsymbol p^{\mathsf T}
	\boldsymbol A_{\perp,t}
	\boldsymbol q
	\,\middle|\,
	\mathcal F_t
	\right)
	\le
	\frac{s_{\max}}{M}
	\|\boldsymbol p\|^2\|\boldsymbol q\|^2.
\end{equation*}

Define
\begin{equation*}
\boldsymbol X_t
=
\sqrt{\gamma_w}\boldsymbol A_{\perp,t}.
\end{equation*}
Then
\begin{equation}
	\operatorname{Var}
	\left(
	\boldsymbol p^{\mathsf T}
	\boldsymbol X_t
	\boldsymbol q
	\,\middle|\,
	\mathcal F_t
	\right)
	\le
	\frac{\gamma_w s_{\max}}{M}
	\|\boldsymbol p\|^2\|\boldsymbol q\|^2.
	\label{eq:C_X_flat}
\end{equation}
Thus the conditioned residual has the flat covariance required by the MDE
input.

For the \(t\)th linear module, the block linearization is
\begin{equation*}
	\boldsymbol K_t
	=
	\begin{pmatrix}
		\boldsymbol I_M
		&
		\sqrt{\gamma_w}\boldsymbol A
		\\
		\sqrt{\gamma_w}\boldsymbol A^{\mathsf T}
		&
		-\boldsymbol\Gamma_t
	\end{pmatrix}.
\end{equation*}
Using \eqref{eq:C_A_decomp}, write
\begin{equation*}
	\boldsymbol K_t
	=
	\boldsymbol D_t+\boldsymbol W_t,
\end{equation*}
where
\begin{equation*}
	\boldsymbol D_t
	=
	\begin{pmatrix}
		\boldsymbol I_M
		&
		\sqrt{\gamma_w}\boldsymbol A_{\parallel,t}
		\\
		\sqrt{\gamma_w}\boldsymbol A_{\parallel,t}^{\mathsf T}
		&
		-\boldsymbol\Gamma_t
	\end{pmatrix},
\end{equation*}
and
\begin{equation}
	\boldsymbol W_t
	=
	\begin{pmatrix}
		\boldsymbol 0
		&
		\boldsymbol X_t
		\\
		\boldsymbol X_t^{\mathsf T}
		&
		\boldsymbol 0
	\end{pmatrix}.
	\label{eq:C_Wt}
\end{equation}
The rest of the appendix verifies that \(\boldsymbol D_t\) is bounded and that
\(\boldsymbol W_t\) satisfies the Gaussian flatness assumptions.

\subsection{Boundedness of the Conditional Mean Deformation}
\label{subsec:C_bounded_deformation}

The only nontrivial point is to prove
\begin{equation*}
\|\boldsymbol A_{\parallel,t}\|=O_p(1).
\end{equation*}
The proof uses the fact that \(t\) is fixed.  Although the history constraints
contain \(M t+N t\) scalar equations, they are generated by finitely many
signal-side and measurement-side history directions.  After removing
asymptotically redundant directions, their normalized Gram matrices are
well-conditioned on the regularity event \(\mathcal R_t\).

We make this precise through the following representation lemma.

\begin{lemma}[Finite-history regression representation]
	\label{lem:C_regression_representation}
	On the regularity event \(\mathcal R_t\), the conditional mean
	\(\boldsymbol A_{\parallel,t}\) admits a representation of the form
	\begin{equation}
		\boldsymbol A_{\parallel,t}
		=
		\sum_{\ell=1}^{L_t}
		c_{\ell,t}
		\frac1M
		\operatorname{diag}(\boldsymbol p_{\ell,t})
		\boldsymbol S
		\operatorname{diag}(\boldsymbol q_{\ell,t}),
		\label{eq:C_profile_weighted_representation}
	\end{equation}
	where
	\begin{equation*}
	\boldsymbol S=(s_{ij})_{1\le i\le M,\,1\le j\le N}
	\end{equation*}
	is the variance-profile matrix, \(L_t=O(1)\) for fixed \(t\), and
	\begin{equation}
	c_{\ell,t}=O_p(1),
	\qquad
	\|\boldsymbol p_{\ell,t}\|_M=O_p(1),
	\qquad
	\|\boldsymbol q_{\ell,t}\|_N=O_p(1).
	\label{eq:C_representation_bounds}
	\end{equation}
	Here
	\begin{equation*}
	\|\boldsymbol p\|_M^2=\frac1M\|\boldsymbol p\|^2,
	\qquad
	\|\boldsymbol q\|_N^2=\frac1N\|\boldsymbol q\|^2.
	\end{equation*}
	The vectors \(\boldsymbol p_{\ell,t}\) are finite linear combinations of
	measurement-side history vectors such as
	\(\boldsymbol w-\boldsymbol u^s\) and \(\boldsymbol u^s\), while the vectors
	\(\boldsymbol q_{\ell,t}\) are finite linear combinations of signal-side
	history vectors such as \(\boldsymbol e_2^s\),
	\(\widehat{\boldsymbol e}_2^s\), and
	\(\boldsymbol\Gamma_s(\widehat{\boldsymbol e}_2^s-\boldsymbol e_2^s)\), for
	\(s<t\).
\end{lemma}

\begin{IEEEproof}
	By Appendix~\ref{app:gaussian_conditioning_histories}, the conditional mean is
	\begin{equation}
		\operatorname{vec}(\boldsymbol A_{\parallel,t})
		=
		\boldsymbol\Sigma_A\boldsymbol L_t^{\mathsf T}
		(\boldsymbol L_t\boldsymbol\Sigma_A\boldsymbol L_t^{\mathsf T})^\dagger
		\boldsymbol b_t,
		\label{eq:C_conditional_mean_vectorized}
	\end{equation}
	where
	\begin{equation*}
	\boldsymbol L_t\operatorname{vec}(\boldsymbol A)=\boldsymbol b_t
	\end{equation*}
	is the vectorized form of the two history constraints
	\begin{equation}
	\boldsymbol A\boldsymbol M_t
	=
	\boldsymbol w\boldsymbol 1_t^{\mathsf T}
	-
	\boldsymbol U_t
	\label{eq:C_AM_constraint}
	\end{equation}
	and
	\begin{equation}
	\boldsymbol A^{\mathsf T}\boldsymbol U_t
	=
	\gamma_w^{-1}
	\left(
	\boldsymbol\Gamma_0(\widehat{\boldsymbol e}_2^0-\boldsymbol e_2^0),
	\ldots,
	\boldsymbol\Gamma_{t-1}
	(\widehat{\boldsymbol e}_2^{t-1}-\boldsymbol e_2^{t-1})
	\right).
	\label{eq:C_ATU_constraint}
	\end{equation}
	The exact notation of the error variables is immaterial for the present
	argument; what matters is that the right-hand sides are finite collections of
	regular measurement-side and signal-side vectors.
	
	The covariance \(\boldsymbol\Sigma_A\) is diagonal with entries \(s_{ij}/M\).
	Therefore, the covariance between an entry \(A_{ij}\) and any history linear
	form of the type
	\begin{equation*}
	\boldsymbol p^{\mathsf T}\boldsymbol A\boldsymbol q
	\end{equation*}
	equals
	\begin{equation*}
	\operatorname{Cov}
	(A_{ij},\boldsymbol p^{\mathsf T}\boldsymbol A\boldsymbol q)
	=
	\frac{s_{ij}}{M}p_iq_j.
	\end{equation*}
	Similarly, the covariance between \(A_{ij}\) and a row-image constraint
	\((\boldsymbol A\boldsymbol q)_i\) or a column-image constraint
	\((\boldsymbol A^{\mathsf T}\boldsymbol p)_j\) has the same separable
	profile-weighted form.  Expanding
	\(\boldsymbol\Sigma_A\boldsymbol L_t^{\mathsf T}\) in
	\eqref{eq:C_conditional_mean_vectorized} on a basis of the finite history
	directions therefore produces matrices whose \((i,j)\) entry is of the form
	\begin{equation*}
	\frac{s_{ij}}{M}p_iq_j,
	\end{equation*}
	where \(\boldsymbol p\) is a measurement-side history vector and
	\(\boldsymbol q\) is a signal-side history vector.  Hence each basis component
	is of the form
	\begin{equation*}
	\frac1M
	\operatorname{diag}(\boldsymbol p)
	\boldsymbol S
	\operatorname{diag}(\boldsymbol q).
	\end{equation*}
	
	It remains to control the coefficients multiplying these components.  On the
	regularity event \(\mathcal R_t\), asymptotically redundant history directions
	are removed, and the nonzero eigenvalues of the normalized profile-weighted
	history Gram matrix
	\begin{equation*}
	\boldsymbol L_t\boldsymbol\Sigma_A\boldsymbol L_t^{\mathsf T}
	\end{equation*}
	restricted to the retained subspace are bounded away from zero and infinity.
	This follows from the regular-history convention and the uniform ellipticity
	\begin{equation*}
	\frac{s_{\min}}{M}\boldsymbol I
	\preceq
	\boldsymbol\Sigma_A
	\preceq
	\frac{s_{\max}}{M}\boldsymbol I,
	\end{equation*}
	which makes the profile-weighted Gram matrices equivalent to the ordinary
	normalized history Gram matrices on the retained finite-dimensional history
	subspace.
	
	The right-hand side \(\boldsymbol b_t\) in
	\eqref{eq:C_conditional_mean_vectorized} consists of the matrices in
	\eqref{eq:C_AM_constraint}--\eqref{eq:C_ATU_constraint}.  By the regularity
	event and the clipping of the diagonal precisions,
	\begin{equation*}
	\|\boldsymbol w-\boldsymbol u^s\|_M=O_p(1),
	\qquad
	\|\boldsymbol u^s\|_M=O_p(1),
	\end{equation*}
	and
	\begin{equation*}
	\|\boldsymbol\Gamma_s
	(\widehat{\boldsymbol e}_2^s-\boldsymbol e_2^s)
	\|_N
	=
	O_p(1)
	\end{equation*}
	for all \(s<t\).  Since \(t\) is fixed, the number of retained history
	directions is finite.  Thus the coefficients generated by
	\begin{equation*}
	(\boldsymbol L_t\boldsymbol\Sigma_A\boldsymbol L_t^{\mathsf T})^\dagger
	\boldsymbol b_t
	\end{equation*}
	are \(O_p(1)\) on the retained subspace.  This proves the representation
	\eqref{eq:C_profile_weighted_representation} with
	\eqref{eq:C_representation_bounds}.
\end{IEEEproof}

We now use this representation to bound the operator norm.

\begin{lemma}[Bounded conditional mean deformation]
	\label{lem:C_A_parallel_bounded}
	On the regularity event \(\mathcal R_t\),
	\begin{equation*}
	\|\boldsymbol A_{\parallel,t}\|=O_p(1).
	\end{equation*}
\end{lemma}

\begin{IEEEproof}
	It suffices to bound each term in
	\eqref{eq:C_profile_weighted_representation}.  Let
	\begin{equation*}
	\boldsymbol B_{\ell,t}
	=
	\frac1M
	\operatorname{diag}(\boldsymbol p_{\ell,t})
	\boldsymbol S
	\operatorname{diag}(\boldsymbol q_{\ell,t}).
	\end{equation*}
	For arbitrary deterministic unit vectors
	\(\boldsymbol x\in\mathbb R^M\) and
	\(\boldsymbol y\in\mathbb R^N\),
	\begin{equation*}
	\begin{aligned}
		|\boldsymbol x^{\mathsf T}\boldsymbol B_{\ell,t}\boldsymbol y|
		&=
		\left|
		\frac1M
		\sum_{i=1}^M\sum_{j=1}^N
		x_i p_{\ell,t,i}s_{ij}q_{\ell,t,j}y_j
		\right|                                                     \\
		&\le
		\frac{s_{\max}}{M}
		\sum_{i=1}^M |x_i p_{\ell,t,i}|
		\sum_{j=1}^N |q_{\ell,t,j}y_j|                              \\
		&\le
		\frac{s_{\max}}{M}
		\|\boldsymbol p_{\ell,t}\|
		\|\boldsymbol q_{\ell,t}\|.
	\end{aligned}
	\end{equation*}
	Taking the supremum over unit vectors gives
	\begin{equation*}
	\|\boldsymbol B_{\ell,t}\|
	\le
	\frac{s_{\max}}{M}
	\|\boldsymbol p_{\ell,t}\|
	\|\boldsymbol q_{\ell,t}\|.
	\end{equation*}
	Using
	\begin{equation*}
	\|\boldsymbol p_{\ell,t}\|=O_p(\sqrt M),
	\qquad
	\|\boldsymbol q_{\ell,t}\|=O_p(\sqrt N),
	\end{equation*}
	and \(M/N\to\delta\in(0,\infty)\), we obtain
	\begin{equation*}
	\|\boldsymbol B_{\ell,t}\|=O_p(1).
	\end{equation*}
	Since \(L_t=O(1)\) and \(c_{\ell,t}=O_p(1)\),
	\begin{equation*}
	\|\boldsymbol A_{\parallel,t}\|
	\le
	\sum_{\ell=1}^{L_t}|c_{\ell,t}|\|\boldsymbol B_{\ell,t}\|
	=
	O_p(1).
	\end{equation*}
\end{IEEEproof}

\subsection{Verification of the Regularized MDE Input}
\label{subsec:C_mde_verification}

We now verify the hypotheses of Lemma~\ref{lem:regularized_mde_input} for the
conditioned block matrix \(\boldsymbol K_t=\boldsymbol D_t+\boldsymbol W_t\).

First, by Appendix~\ref{app:gaussian_conditioning_histories},
\(\boldsymbol A_{\perp,t}\) is centered Gaussian conditionally on
\(\mathcal F_t\).  Hence \(\boldsymbol W_t\) in \eqref{eq:C_Wt} is also
centered Gaussian conditionally on \(\mathcal F_t\).

Second, the flat covariance condition follows from
\eqref{eq:C_X_flat}.  Indeed, for any conditioning-measurable vectors
\(\boldsymbol p\in\mathbb R^M\) and
\(\boldsymbol q\in\mathbb R^N\),
\begin{equation*}
\operatorname{Var}
\left(
\boldsymbol p^{\mathsf T}
\boldsymbol X_t
\boldsymbol q
\,\middle|\,
\mathcal F_t
\right)
\le
\frac{C}{M}
\|\boldsymbol p\|^2\|\boldsymbol q\|^2
\end{equation*}
with \(C=\gamma_w s_{\max}\).  This is precisely the flat bilinear-form
condition used in the regularized correlated-Gaussian MDE input.

Third, the deterministic deformation is bounded.  By
Lemma~\ref{lem:C_A_parallel_bounded} and the precision clipping
\begin{equation*}
\|\boldsymbol\Gamma_t\|\le\gamma_{\max},
\end{equation*}
we have
\begin{equation*}
\|\boldsymbol D_t\|
\le
1+\gamma_{\max}
+
2\sqrt{\gamma_w}\|\boldsymbol A_{\parallel,t}\|
=
O_p(1).
\end{equation*}
Fourth, the loading is uniformly regularized:
\begin{equation*}
\boldsymbol\Gamma_t
\succeq
\gamma_{\min}\boldsymbol I_N
\end{equation*}
by Assumption~\ref{ass:regularized_scalar_module}.  Therefore the lower-right
block of the linearization remains separated from singularity in the sense
required by Lemma~\ref{lem:regularized_mde_input}.

Finally, the deterministic insertions and source vectors used later are
admissible.  Vectors built from the regular histories, such as
\begin{equation*}
\begin{pmatrix}
	\boldsymbol w\\
	-\boldsymbol\Gamma_t\boldsymbol e_2^t/\sqrt{\gamma_w}
\end{pmatrix},
\end{equation*}
have bounded normalized energy and no spikes, because
\(\boldsymbol w\) is Gaussian, \(\boldsymbol e_2^t\) is signal-side
admissible, and \(\boldsymbol\Gamma_t\) is uniformly bounded.  Insertions such
as
\begin{equation*}
\begin{pmatrix}
	\boldsymbol S_j & \boldsymbol 0\\
	\boldsymbol 0 & \boldsymbol 0
\end{pmatrix}
\end{equation*}
are uniformly bounded because \(\|\boldsymbol S_j\|\le s_{\max}\).

All hypotheses of Lemma~\ref{lem:regularized_mde_input} are therefore
satisfied conditionally on \(\mathcal F_t\).  Consequently, the one-resolvent
and two-resolvent MDE deterministic equivalents invoked in
Appendix~\ref{app:mde_responses} apply to the EP-conditioned linearization
\(\boldsymbol K_t\).

\subsection{Conclusion}
\label{subsec:C_conclusion}

The preceding arguments prove Proposition~\ref{prop:bounded_deformation_mde}.
Indeed, the conditional mean deformation satisfies
\begin{equation*}
\|\boldsymbol A_{\parallel,t}\|=O_p(1),
\end{equation*}
the residual matrix is centered correlated Gaussian with flat covariance, and
the block loading is uniformly regularized.  Therefore the conditioned
linearization
\begin{equation*}
\boldsymbol K_t
=
\begin{pmatrix}
	\boldsymbol I_M
	&
	\sqrt{\gamma_w}\boldsymbol A
	\\
	\sqrt{\gamma_w}\boldsymbol A^{\mathsf T}
	&
	-\boldsymbol\Gamma_t
\end{pmatrix}
\end{equation*}
falls within the regularized correlated-Gaussian MDE class of
Lemma~\ref{lem:regularized_mde_input}.  No additional random-matrix theorem is
used beyond that input.

\section{MDE Responses for the Linear Variance and Residual Covariance}
\label{app:mde_responses}

This appendix derives the deterministic equivalents \(T_{t,j}\) and
\(\Theta_j^{r,s}\) used in the state variables of the main theorem.  The
regularized correlated-Gaussian MDE input has already been stated in
Lemma~\ref{lem:regularized_mde_input}.  Therefore, we do not prove an MDE local
law here.  Instead, we verify how the MDE input is applied to the block
linearization generated by the EP linear module.

There are two objects to identify.  The first is the retained diagonal
variance of the linear Gaussian belief,
\begin{equation*}
[\boldsymbol Q_t]_{jj},
\qquad
\boldsymbol Q_t
=
\left(
\gamma_w\boldsymbol A^{\mathsf T}\boldsymbol A
+
\boldsymbol\Gamma_t
\right)^{-1}.
\end{equation*}
The second is the conditioned column-covariance response associated with the
Schur residual at coordinate \(j\).  In the absence of history conditioning this
response reduces to the familiar profile-weighted residual covariance
\begin{equation*}
\frac1M
(\boldsymbol u^r)^{\mathsf T}
\boldsymbol S_j
\boldsymbol u^s,
\qquad
\boldsymbol S_j
=
\operatorname{diag}(s_{1j},\ldots,s_{Mj}).
\end{equation*}
Under the conditioned history, the same quantity is represented by the
conditional column covariance operator and is denoted by
\(\Theta_j^{r,s}\).  The retained variance follows from the lower-right block
of the MDE solution, whereas \(\Theta_j^{r,s}\) follows from the two-resolvent
response of the same conditioned MDE.

Throughout this appendix, \(r,s,t\) and \(j\) are fixed finite indices.  The
corresponding empirical averaged versions follow whenever
Lemma~\ref{lem:regularized_mde_input} is invoked in its averaged admissible
quadratic-form form.  This distinction is important: the fixed-coordinate
bounds below follow directly from the pointwise MDE input, whereas an
\(\ell_2\)-diagonal statement requires the averaged diagonal version of the
same input.

\subsection{Block Linearization and the Linear-Module Covariance}
\label{subsec:D_block_linearization}

For a fixed iteration \(t\), define the block linearization
\begin{equation*}
	\boldsymbol K_t
	=
	\begin{pmatrix}
		\boldsymbol I_M
		&
		\sqrt{\gamma_w}\boldsymbol A
		\\
		\sqrt{\gamma_w}\boldsymbol A^{\mathsf T}
		&
		-\boldsymbol\Gamma_t
	\end{pmatrix},
	\qquad
	\boldsymbol G_t=\boldsymbol K_t^{-1}.
\end{equation*}
The lower-right Schur complement of \(\boldsymbol K_t\) is
\begin{equation*}
-\boldsymbol\Gamma_t
-
\gamma_w\boldsymbol A^{\mathsf T}\boldsymbol A
=
-
\left(
\gamma_w\boldsymbol A^{\mathsf T}\boldsymbol A
+
\boldsymbol\Gamma_t
\right).
\end{equation*}
Consequently, the lower-right block of \(\boldsymbol G_t\) is
\begin{equation}
	[\boldsymbol G_t]_{22}
	=
	-
	\left(
	\gamma_w\boldsymbol A^{\mathsf T}\boldsymbol A
	+
	\boldsymbol\Gamma_t
	\right)^{-1}
	=
	-\boldsymbol Q_t.
	\label{eq:D_lower_right_block}
\end{equation}

Let
\begin{equation*}
\boldsymbol e_{M+j}
=
\begin{pmatrix}
	\boldsymbol 0_M\\
	\boldsymbol e_j
\end{pmatrix}
\in\mathbb R^{M+N},
\end{equation*}
where \(\boldsymbol e_j\in\mathbb R^N\) is the \(j\)th canonical vector.  From
\eqref{eq:D_lower_right_block},
\begin{equation}
	[\boldsymbol Q_t]_{jj}
	=
	-
	\boldsymbol e_{M+j}^{\mathsf T}
	\boldsymbol G_t
	\boldsymbol e_{M+j}.
	\label{eq:D_Q_diag_from_G}
\end{equation}

Let \(\boldsymbol M_t\) denote the solution of the MDE associated with the
conditioned linearization \(\boldsymbol K_t\).  Define
\begin{equation}
	T_{t,j}
	=
	-
	\boldsymbol e_{M+j}^{\mathsf T}
	\boldsymbol M_t
	\boldsymbol e_{M+j}.
	\label{eq:D_T_def}
\end{equation}
By Lemma~\ref{lem:regularized_mde_input}, applied to the deterministic vectors
\(\boldsymbol e_{M+j}\), we have
\begin{equation*}
\boldsymbol e_{M+j}^{\mathsf T}
(\boldsymbol G_t-\boldsymbol M_t)
\boldsymbol e_{M+j}
=
O_p(N^{-1/2}).
\end{equation*}
Combining this with \eqref{eq:D_Q_diag_from_G} and
\eqref{eq:D_T_def} gives
\begin{equation}
	[\boldsymbol Q_t]_{jj}
	=
	T_{t,j}
	+
	O_p(N^{-1/2}).
	\label{eq:D_Q_diag_DE}
\end{equation}

If the averaged diagonal version of Lemma~\ref{lem:regularized_mde_input} is
used, then the same argument yields
\begin{equation}
	\frac1N\sum_{j=1}^N
	\left|
	[\boldsymbol Q_t]_{jj}
	-
	T_{t,j}
	\right|^2
	\overset{p}{\longrightarrow}0.
	\label{eq:D_Q_diag_l2_DE}
\end{equation}
In the main proof, \eqref{eq:D_Q_diag_DE} is the pointwise statement, while
\eqref{eq:D_Q_diag_l2_DE} is used whenever an empirical diagonal replacement is
required.

\subsection{Resolvent Representation of the Measurement Residual}
\label{subsec:D_residual_representation}

We next express the measurement residual as the upper block of the same
resolvent.  The linear-module error is
\begin{equation}
	\widehat{\boldsymbol e}_2^t
	=
	\boldsymbol Q_t
	\left(
	\gamma_w\boldsymbol A^{\mathsf T}\boldsymbol w
	+
	\boldsymbol\Gamma_t\boldsymbol e_2^t
	\right),
	\label{eq:D_ehat_def}
\end{equation}
and the measurement residual is
\begin{equation*}
	\boldsymbol u^t
	=
	\boldsymbol w-\boldsymbol A\widehat{\boldsymbol e}_2^t .
\end{equation*}
Define the lifted source vector
\begin{equation}
	\boldsymbol b_t
	=
	\begin{pmatrix}
		\boldsymbol w\\
		-\boldsymbol\Gamma_t\boldsymbol e_2^t/\sqrt{\gamma_w}
	\end{pmatrix}
	\in\mathbb R^{M+N}.
	\label{eq:D_bt_def}
\end{equation}
We claim that
\begin{equation}
	\begin{pmatrix}
		\boldsymbol u^t\\
		\widehat{\boldsymbol e}_2^t/\sqrt{\gamma_w}
	\end{pmatrix}
	=
	\boldsymbol G_t\boldsymbol b_t.
	\label{eq:D_residual_resolvent}
\end{equation}

To verify this identity, multiply the left-hand side by
\(\boldsymbol K_t\).  The upper block is
\begin{equation*}
\boldsymbol u^t
+
\sqrt{\gamma_w}\boldsymbol A
\left(
\widehat{\boldsymbol e}_2^t/\sqrt{\gamma_w}
\right)
=
\boldsymbol u^t+\boldsymbol A\widehat{\boldsymbol e}_2^t
=
\boldsymbol w.
\end{equation*}
For the lower block, the normal equation associated with
\eqref{eq:D_ehat_def} gives
\begin{equation*}
\left(
\gamma_w\boldsymbol A^{\mathsf T}\boldsymbol A
+
\boldsymbol\Gamma_t
\right)
\widehat{\boldsymbol e}_2^t
=
\gamma_w\boldsymbol A^{\mathsf T}\boldsymbol w
+
\boldsymbol\Gamma_t\boldsymbol e_2^t.
\end{equation*}
Equivalently,
\begin{equation*}
\gamma_w\boldsymbol A^{\mathsf T}
(\boldsymbol w-\boldsymbol A\widehat{\boldsymbol e}_2^t)
=
\boldsymbol\Gamma_t
(\widehat{\boldsymbol e}_2^t-\boldsymbol e_2^t).
\end{equation*}
Using \(\boldsymbol u^t=\boldsymbol w-\boldsymbol A\widehat{\boldsymbol e}_2^t\),
we obtain
\begin{equation*}
\sqrt{\gamma_w}\boldsymbol A^{\mathsf T}\boldsymbol u^t
-
\boldsymbol\Gamma_t
\left(
\widehat{\boldsymbol e}_2^t/\sqrt{\gamma_w}
\right)
=
-
\boldsymbol\Gamma_t\boldsymbol e_2^t/\sqrt{\gamma_w}.
\end{equation*}
Thus
\begin{equation*}
\boldsymbol K_t
\begin{pmatrix}
	\boldsymbol u^t\\
	\widehat{\boldsymbol e}_2^t/\sqrt{\gamma_w}
\end{pmatrix}
=
\begin{pmatrix}
	\boldsymbol w\\
	-\boldsymbol\Gamma_t\boldsymbol e_2^t/\sqrt{\gamma_w}
\end{pmatrix}
=
\boldsymbol b_t,
\end{equation*}
which proves \eqref{eq:D_residual_resolvent}.

\subsection{Two-Resolvent Response for the Conditioned Column Covariance}
\label{subsec:D_two_resolvent_response}

The covariance kernel in Appendix~\ref{app:covariance_kernel} is computed
under the actual conditioned history.  Consequently, the relevant insertion is
not necessarily the raw profile matrix \(M^{-1}\boldsymbol S_j\), but the
conditional covariance operator of the unexposed part of the \(j\)th column.
We denote this operator by
\begin{equation*}
    \boldsymbol\Xi_{j|r,s},
\end{equation*}
where the indices \((r,s)\) indicate the two Schur residual times for which the
covariance is evaluated.  It is an \(M\times M\) positive semidefinite matrix
satisfying the contraction bound
\begin{equation}
    \boldsymbol 0
    \preceq
    \boldsymbol\Xi_{j|r,s}
    \preceq
    \frac1M\boldsymbol S_j,
    \qquad
    \boldsymbol S_j=\operatorname{diag}(s_{1j},\ldots,s_{Mj}).
    \label{eq:D_Xi_contraction}
\end{equation}
When there is no conditioning beyond the trivial one,
\(\boldsymbol\Xi_{j|r,s}=M^{-1}\boldsymbol S_j\).

Define the lifted conditioned covariance insertion
\begin{equation*}
    \boldsymbol C_{j|r,s}^{\mathrm{cond}}
    =
    \begin{pmatrix}
        \boldsymbol\Xi_{j|r,s} & \boldsymbol 0\\
        \boldsymbol 0 & \boldsymbol 0
    \end{pmatrix}.
\end{equation*}
Since this insertion acts only on the measurement block,
\eqref{eq:D_residual_resolvent} gives
\begin{equation}
    (\boldsymbol u^r)^{\mathsf T}
    \boldsymbol\Xi_{j|r,s}
    \boldsymbol u^s
    =
    \boldsymbol b_r^{\mathsf T}
    \boldsymbol G_r
    \boldsymbol C_{j|r,s}^{\mathrm{cond}}
    \boldsymbol G_s
    \boldsymbol b_s.
    \label{eq:D_conditioned_cov_resolvent}
\end{equation}

Let
\begin{equation*}
    \mathcal L_{r,s}[\boldsymbol C_{j|r,s}^{\mathrm{cond}}]
\end{equation*}
be the deterministic two-resolvent response associated with the conditioned
insertion \(\boldsymbol C_{j|r,s}^{\mathrm{cond}}\), as defined by
Lemma~\ref{lem:regularized_mde_input}.  Namely,
\begin{equation*}
    \mathcal L_{r,s}[\boldsymbol C_{j|r,s}^{\mathrm{cond}}]
    =
    \boldsymbol M_r
    \boldsymbol C_{j|r,s}^{\mathrm{cond}}
    \boldsymbol M_s
    +
    \boldsymbol M_r
    \mathcal S_{r,s}
    \!\left[
        \mathcal L_{r,s}[\boldsymbol C_{j|r,s}^{\mathrm{cond}}]
    \right]
    \boldsymbol M_s,
\end{equation*}
where \(\boldsymbol M_r,\boldsymbol M_s\) are the MDE solutions associated with
\(\boldsymbol K_r,\boldsymbol K_s\), and \(\mathcal S_{r,s}\) is the covariance
operator of the corresponding conditioned Gaussian pair.

We define
\begin{equation}
    \Theta_j^{r,s}
    =
    \boldsymbol b_r^{\mathsf T}
    \mathcal L_{r,s}[\boldsymbol C_{j|r,s}^{\mathrm{cond}}]
    \boldsymbol b_s.
    \label{eq:D_Theta_def}
\end{equation}
The insertion \(\boldsymbol C_{j|r,s}^{\mathrm{cond}}\) has operator norm
\(O(M^{-1})\) by \eqref{eq:D_Xi_contraction}.  Hence the two-resolvent part of
Lemma~\ref{lem:regularized_mde_input}, applied to this scaled insertion, yields
\begin{equation}
    \boldsymbol b_r^{\mathsf T}
    \left\{
        \boldsymbol G_r
        \boldsymbol C_{j|r,s}^{\mathrm{cond}}
        \boldsymbol G_s
        -
        \mathcal L_{r,s}[\boldsymbol C_{j|r,s}^{\mathrm{cond}}]
    \right\}
    \boldsymbol b_s
    =
    O_p(N^{-1/2}).
    \label{eq:D_two_resolvent_conditioned_MDE}
\end{equation}
Combining \eqref{eq:D_conditioned_cov_resolvent},
\eqref{eq:D_Theta_def}, and \eqref{eq:D_two_resolvent_conditioned_MDE}, we obtain
\begin{equation}
    (\boldsymbol u^r)^{\mathsf T}
    \boldsymbol\Xi_{j|r,s}
    \boldsymbol u^s
    =
    \Theta_j^{r,s}
    +
    O_p(N^{-1/2}).
    \label{eq:D_Theta_DE}
\end{equation}

In the unconditioned special case
\(\boldsymbol\Xi_{j|r,s}=M^{-1}\boldsymbol S_j\),
\eqref{eq:D_Theta_DE} becomes
\begin{equation*}
    \frac1M
    (\boldsymbol u^r)^{\mathsf T}
    \boldsymbol S_j
    \boldsymbol u^s
    =
    \Theta_j^{r,s}
    +
    O_p(N^{-1/2}).
\end{equation*}
Thus the notation \(\Theta_j^{r,s}\) always refers to the deterministic
conditioned MDE response.  This convention is used in the covariance kernel
\begin{equation*}
    \zeta_j^{r,s}
    =
    \gamma_w^2 T_{r,j}T_{s,j}\Theta_j^{r,s}.
\end{equation*}

\subsection{Admissibility of Sources and Insertions}
\label{subsec:D_admissibility}

It remains to justify that the vectors and insertions used above are legitimate
inputs for Lemma~\ref{lem:regularized_mde_input}.  By Appendix~\ref{app:bounded_deformation_mde},
the conditioned linearizations \(\boldsymbol K_t\) satisfy the assumptions of
the regularized MDE input: the centered Gaussian residual has flat covariance,
the finite-rank deformation is bounded, and the diagonal loading is uniformly
regularized.

The source vector \(\boldsymbol b_t\) in \eqref{eq:D_bt_def} is admissible.
Indeed, the measurement block \(\boldsymbol w\) satisfies
\begin{equation*}
\frac1M\|\boldsymbol w\|^2=O_p(1),
\qquad
\frac{\|\boldsymbol w\|_\infty}{\sqrt M}\to0
\end{equation*}
by the Gaussian noise assumption.  The signal block satisfies
\begin{equation*}
\frac1N
\|\boldsymbol\Gamma_t\boldsymbol e_2^t\|^2
\le
\gamma_{\max}^2
\frac1N\|\boldsymbol e_2^t\|^2
=
O_p(1),
\end{equation*}
and
\begin{equation*}
\frac{
	\|\boldsymbol\Gamma_t\boldsymbol e_2^t\|_\infty
}{\sqrt N}
\le
\gamma_{\max}
\frac{\|\boldsymbol e_2^t\|_\infty}{\sqrt N}
\to0,
\end{equation*}
because \(\boldsymbol e_2^t\) is signal-side admissible on the regularity
event.  Thus \(\boldsymbol b_t\) is an admissible lifted vector.

The conditioned insertion \(\boldsymbol C_{j|r,s}^{\mathrm{cond}}\) is uniformly
admissible after its natural single-column scaling.  Indeed,
\begin{equation*}
    \|\boldsymbol C_{j|r,s}^{\mathrm{cond}}\|
    =
    \|\boldsymbol\Xi_{j|r,s}\|
    \le
    \frac{s_{\max}}{M}.
\end{equation*}
It preserves measurement-side admissibility and has precisely the scaling of a
single-column covariance insertion.  In the unconditioned case this insertion
is \(M^{-1}\operatorname{diag}(\boldsymbol S_j,\boldsymbol 0)\).  Therefore it
is a legitimate insertion for the two-resolvent MDE input.

Combining the admissibility verification with the derivations in
Sections~\ref{subsec:D_block_linearization}--\ref{subsec:D_two_resolvent_response}
proves the MDE response statements
\begin{equation*}
[\boldsymbol Q_t]_{jj}
=
T_{t,j}+O_p(N^{-1/2})
\end{equation*}
and
\begin{equation*}
(\boldsymbol u^r)^{\mathsf T}
\boldsymbol\Xi_{j|r,s}
\boldsymbol u^s
=
\Theta_j^{r,s}
+
O_p(N^{-1/2}),
\end{equation*}
for every fixed \(j,r,s,t\).  If the conditioning is trivial, the latter
statement reduces to
\begin{equation*}
\frac1M
(\boldsymbol u^r)^{\mathsf T}
\boldsymbol S_j
\boldsymbol u^s
=
\Theta_j^{r,s}
+
O_p(N^{-1/2}).
\end{equation*}
The corresponding empirical averaged versions follow from the averaged form of
Lemma~\ref{lem:regularized_mde_input}, when that form is invoked in the main
proof.

\section{Schur Complement under the Conditioned History}
\label{app:schur_conditioned_history}

This appendix proves the Schur-kernel representation used in
Proposition~\ref{prop:schur_kernel}.  The proof is based on the actual
conditioned EP history.  No auxiliary recursion obtained by deleting a
coordinate and rerunning the algorithm is introduced.  We freeze the incoming
message \(\boldsymbol q_t\) and the diagonal loading
\(\boldsymbol\Gamma_t\) generated by the actual history, remove only the
\(j\)th coordinate in the current linear system, and apply the Schur
complement.  This produces an exact one-step column-cavity identity.  The MDE
response then replaces the finite-dimensional diagonal gain.

Throughout the appendix, \(t\) is fixed.  We write
\begin{equation*}
    \|\boldsymbol a\|_N^2=\frac1N\|\boldsymbol a\|^2,
    \qquad
    \|\boldsymbol b\|_M^2=\frac1M\|\boldsymbol b\|^2 .
\end{equation*}
The diagonal covariance of the linear module is
\begin{equation*}
    \boldsymbol C_t
    =
    \left(
        \gamma_w\boldsymbol A^{\mathsf T}\boldsymbol A
        +
        \boldsymbol\Gamma_t
    \right)^{-1},
\end{equation*}
and the linear-module output error is
\begin{equation}
    \boldsymbol m_t
    =
    \boldsymbol C_t
    \left(
        \gamma_w\boldsymbol A^{\mathsf T}\boldsymbol w
        +
        \boldsymbol\Gamma_t\boldsymbol q_t
    \right).
    \label{eq:E_linear_module_output}
\end{equation}
The full measurement residual is
\begin{equation*}
    \boldsymbol u_t
    =
    \boldsymbol w-
    \boldsymbol A\boldsymbol m_t.
\end{equation*}

\subsection{One-Step Column Cavity}
\label{subsec:E_column_cavity}

Fix a coordinate \(j\).  Partition
\begin{equation*}
    \boldsymbol A=(\boldsymbol a_j,\boldsymbol A_{-j}),
\end{equation*}
where \(\boldsymbol a_j\in\mathbb R^M\) is the \(j\)th column and
\(\boldsymbol A_{-j}\in\mathbb R^{M\times(N-1)}\) is the matrix with that
column removed.  Similarly, write
\begin{equation*}
    \boldsymbol q_t=(q_{t,j},\boldsymbol q_{t,-j}),
    \qquad
    \boldsymbol m_t=(m_{t,j},\boldsymbol m_{t,-j}),
\end{equation*}
and
\begin{equation*}
    \boldsymbol\Gamma_t
    =
    \begin{pmatrix}
        \gamma_{t,j} & 0\\
        0 & \boldsymbol\Gamma_{t,-j}
    \end{pmatrix}.
\end{equation*}
Define the one-step column-cavity covariance
\begin{equation*}
    \overline{\boldsymbol C}_{t,-j}^{(j)}
    =
    \left(
        \gamma_w\boldsymbol A_{-j}^{\mathsf T}\boldsymbol A_{-j}
        +
        \boldsymbol\Gamma_{t,-j}
    \right)^{-1}.
\end{equation*}
The corresponding column-cavity estimate is
\begin{equation*}
    \overline{\boldsymbol m}_{t,-j}^{(j)}
    =
    \overline{\boldsymbol C}_{t,-j}^{(j)}
    \left(
        \gamma_w\boldsymbol A_{-j}^{\mathsf T}\boldsymbol w
        +
        \boldsymbol\Gamma_{t,-j}\boldsymbol q_{t,-j}
    \right),
\end{equation*}
and the column-cavity measurement residual is
\begin{equation*}
    \overline{\boldsymbol u}_{t}^{(j)}
    =
    \boldsymbol w
    -
    \boldsymbol A_{-j}\overline{\boldsymbol m}_{t,-j}^{(j)}.
\end{equation*}
The construction above is not a modified EP trajectory.  The quantities
\(\boldsymbol q_t\) and \(\boldsymbol\Gamma_t\) are the actual incoming message
and precision at iteration \(t\); only the current linear solve is expressed
through a coordinate Schur complement.

\subsection{Exact Schur Complement Identity}
\label{subsec:E_exact_schur}

Let
\begin{equation*}
    \boldsymbol B_{t,-j}
    =
    \gamma_w\boldsymbol A_{-j}^{\mathsf T}\boldsymbol A_{-j}
    +
    \boldsymbol\Gamma_{t,-j},
\end{equation*}
\begin{equation*}
    \boldsymbol b_{t,j}
    =
    \gamma_w\boldsymbol A_{-j}^{\mathsf T}\boldsymbol a_j,
    \qquad
    \alpha_{t,j}^{\rm sch}
    =
    \gamma_w\boldsymbol a_j^{\mathsf T}\boldsymbol a_j
    +
    \gamma_{t,j}.
\end{equation*}
The normal matrix
\(\gamma_w\boldsymbol A^{\mathsf T}\boldsymbol A+\boldsymbol\Gamma_t\) has
the block form
\begin{equation*}
    \begin{pmatrix}
        \alpha_{t,j}^{\rm sch}
        &
        \boldsymbol b_{t,j}^{\mathsf T}
        \\
        \boldsymbol b_{t,j}
        &
        \boldsymbol B_{t,-j}
    \end{pmatrix}.
\end{equation*}
The right-hand side of \eqref{eq:E_linear_module_output} is partitioned as
\begin{equation*}
    y_{t,j}
    =
    \gamma_w\boldsymbol a_j^{\mathsf T}\boldsymbol w
    +
    \gamma_{t,j}q_{t,j},
\end{equation*}
and
\begin{equation*}
    \boldsymbol y_{t,-j}
    =
    \gamma_w\boldsymbol A_{-j}^{\mathsf T}\boldsymbol w
    +
    \boldsymbol\Gamma_{t,-j}\boldsymbol q_{t,-j}.
\end{equation*}
By definition,
\begin{equation*}
    \overline{\boldsymbol m}_{t,-j}^{(j)}
    =
    \boldsymbol B_{t,-j}^{-1}\boldsymbol y_{t,-j}.
\end{equation*}
The Schur complement formula gives
\begin{equation}
    m_{t,j}
    =
    d_{t,j}
    \left(
        y_{t,j}
        -
        \boldsymbol b_{t,j}^{\mathsf T}
        \boldsymbol B_{t,-j}^{-1}
        \boldsymbol y_{t,-j}
    \right),
    \label{eq:E_schur_first}
\end{equation}
where
\begin{equation*}
    d_{t,j}
    =
    [\boldsymbol C_t]_{jj}
    =
    \left(
        \alpha_{t,j}^{\rm sch}
        -
        \boldsymbol b_{t,j}^{\mathsf T}
        \boldsymbol B_{t,-j}^{-1}
        \boldsymbol b_{t,j}
    \right)^{-1}.
\end{equation*}
Substituting the definitions into \eqref{eq:E_schur_first},
\begin{equation*}
\begin{aligned}
    &y_{t,j}
    -
    \boldsymbol b_{t,j}^{\mathsf T}
    \boldsymbol B_{t,-j}^{-1}
    \boldsymbol y_{t,-j}                                      \\
    &\quad=
    \gamma_w\boldsymbol a_j^{\mathsf T}\boldsymbol w
    +
    \gamma_{t,j}q_{t,j}
    -
    \gamma_w\boldsymbol a_j^{\mathsf T}
    \boldsymbol A_{-j}\overline{\boldsymbol m}_{t,-j}^{(j)}    \\
    &\quad=
    \gamma_{t,j}q_{t,j}
    +
    \gamma_w\boldsymbol a_j^{\mathsf T}
    \overline{\boldsymbol u}_{t}^{(j)} .
\end{aligned}
\end{equation*}
Therefore,
\begin{equation}
    m_{t,j}
    =
    d_{t,j}
    \left(
        \gamma_{t,j}q_{t,j}
        +
        \gamma_w
        \boldsymbol a_j^{\mathsf T}
        \overline{\boldsymbol u}_{t}^{(j)}
    \right).
    \label{eq:E_exact_schur_identity}
\end{equation}
This identity is exact and uses only the current linear system.

\subsection{Energy Bounds and MDE Gain Replacement}
\label{subsec:E_mde_gain_replacement}

Define the Schur drive
\begin{equation*}
    r_{t,j}^{\rm sch}
    =
    \gamma_{t,j}q_{t,j}
    +
    \gamma_w
    \boldsymbol a_j^{\mathsf T}
    \overline{\boldsymbol u}_{t}^{(j)}.
\end{equation*}
Equation~\eqref{eq:E_exact_schur_identity} states that
\(m_{t,j}=d_{t,j}r_{t,j}^{\rm sch}\).

\begin{lemma}[Bounded Schur-drive energy]
\label{lem:E_schur_drive_energy}
On the regularity event,
\begin{equation*}
    \frac1N\sum_{j=1}^N |r_{t,j}^{\rm sch}|^2=O_p(1).
\end{equation*}
\end{lemma}

\begin{IEEEproof}
The regularized diagonal loading gives
\(\|\boldsymbol C_t\|\le \gamma_{\min}^{-1}\).  Since
\(\boldsymbol A\) has bounded operator norm on the regularity event,
\(\boldsymbol w\) has bounded normalized energy, and
\(\boldsymbol q_t\) is signal-side admissible, \eqref{eq:E_linear_module_output}
implies
\begin{equation*}
    \|\boldsymbol m_t\|_N=O_p(1).
\end{equation*}
Moreover, on the same event, the diagonal entries \(d_{t,j}\) are bounded away
from zero and infinity.  Indeed,
\(\lambda_{\max}(\boldsymbol C_t)\le\gamma_{\min}^{-1}\), while
\begin{equation*}
    \lambda_{\min}(\boldsymbol C_t)
    \ge
    \left(
        \gamma_w\|\boldsymbol A\|^2+
        \gamma_{\max}
    \right)^{-1}
\end{equation*}
on a probability-one limiting event.  Hence
\(|r_{t,j}^{\rm sch}|=|d_{t,j}^{-1}m_{t,j}|\le C|m_{t,j}|\), which proves the
claim.
\end{IEEEproof}

Appendix~\ref{app:mde_responses} gives the averaged diagonal MDE response
\begin{equation*}
    \frac1N\sum_{j=1}^N |d_{t,j}-T_{t,j}|^2
    \overset{p}{\longrightarrow}0.
\end{equation*}
Combining \eqref{eq:E_exact_schur_identity} with
\(d_{t,j}=T_{t,j}+(d_{t,j}-T_{t,j})\), we get
\begin{equation*}
    m_{t,j}
    =
    T_{t,j} r_{t,j}^{\rm sch}
    +
    \Delta_{t,j},
\end{equation*}
where
\begin{equation*}
    \Delta_{t,j}
    =
    (d_{t,j}-T_{t,j})r_{t,j}^{\rm sch}.
\end{equation*}
By Lemma~\ref{lem:E_schur_drive_energy} and the product-stability lemma in
Appendix~\ref{app:aux_probability},
\begin{equation*}
    \frac1N\sum_{j=1}^N |\Delta_{t,j}|^2
    \overset{p}{\longrightarrow}0.
\end{equation*}

Define the history-conditioned Schur residual
\begin{equation*}
    \mathcal Z_{t,j}
    =
    \gamma_w T_{t,j}
    \boldsymbol a_j^{\mathsf T}
    \overline{\boldsymbol u}_{t}^{(j)}.
\end{equation*}
Then
\begin{equation*}
    m_{t,j}
    =
    \gamma_{t,j}T_{t,j} q_{t,j}
    +
    \mathcal Z_{t,j}
    +
    \Delta_{t,j},
\end{equation*}
with \(\|\boldsymbol\Delta_t\|_N\to0\) in probability.

\subsection{Interpretation}
\label{subsec:E_interpretation}

The term \(\mathcal Z_{t,j}\) is not a fresh innovation at this stage.  The
analysis is conditioned on the predictable linear history, and the column
\(\boldsymbol a_j\) has generally been constrained by that history.  Thus
\(\mathcal Z_{t,j}\) is a coordinate of a history-conditioned Gaussian process.
Appendix~\ref{app:covariance_kernel} identifies its covariance kernel, and
Appendix~\ref{app:cavity_memory_decomposition} applies Gaussian regression to
separate its predictable memory component from its fresh innovation.  This
proves Proposition~\ref{prop:schur_kernel}.

\section{Column-Wise Gaussian Kernel under the Conditioned History}
\label{app:covariance_kernel}

This appendix proves the coordinate-wise Gaussian-kernel statement used in
Proposition~\ref{prop:covariance_empirical_law}.  The proof has two logically
separate parts.  First, after the predictable linear history has been fixed, we
reveal only finitely many scalar projections of one column of the Gaussian
matrix.  This gives a genuine finite-dimensional Gaussian vector; no deleted
EP recursion is introduced.  Second, the covariance of this vector is identified
through the conditioned two-resolvent MDE response from
Appendix~\ref{app:mde_responses}.  Keeping these two steps separate is useful:
the Gaussianity is a consequence of finite-dimensional Gaussian conditioning,
whereas the numerical value of the covariance is a random-matrix deterministic
equivalent.

Appendix~\ref{app:schur_conditioned_history} established the Schur expansion
\begin{equation*}
    m_{t,j}
    =
    \gamma_{t,j}T_{t,j} q_{t,j}
    +
    \mathcal Z_{t,j}
    +
    \Delta_{t,j},
    \qquad
    \|\boldsymbol\Delta_t\|_N
    \overset{p}{\longrightarrow}0,
\end{equation*}
where
\begin{equation}
    \mathcal Z_{t,j}
    =
    \gamma_w T_{t,j}
    \boldsymbol a_j^{\mathsf T}
    \overline{\boldsymbol u}_{t}^{(j)}.
    \label{eq:F_Z_schur_def}
\end{equation}
Here \(\overline{\boldsymbol u}_{t}^{(j)}\) is the one-step column-cavity
residual of the current linear system.  It is a Schur-complement object, not a
new trajectory of the algorithm.  The variable \(\mathcal Z_{t,j}\) will be
shown to be a coordinate of a Gaussian history.  It should not be confused with
the fresh innovation obtained later by Gaussian regression.

\subsection{Column-Wise Revealed History}
\label{subsec:F_column_revealed_history}

Fix a coordinate \(j\) and a finite time \(t\).  We use
\(\mathcal F_t^{\rm lin}\) for the predictable linear-history filtration in
\eqref{eq:proof_filtration}.  Conditionally on this history,
Appendix~\ref{app:gaussian_conditioning_histories} gives
\begin{equation*}
    \boldsymbol A
    =
    \boldsymbol A_{\parallel,t}
    +
    \boldsymbol A_{\perp,t},
\end{equation*}
where \(\boldsymbol A_{\parallel,t}\) is history-measurable and
\(\boldsymbol A_{\perp,t}\) is centered Gaussian.  Let
\(\boldsymbol a_j\) denote the \(j\)th column of \(\boldsymbol A\).  Its
conditioned residual covariance is denoted by
\begin{equation*}
    \boldsymbol\Xi_{j|t}
    =
    \operatorname{Cov}
    \bigl(
        \boldsymbol a_j-\mathbb E[\boldsymbol a_j\mid\mathcal F_t^{\rm lin}]
        \mid
        \mathcal F_t^{\rm lin}
    \bigr).
\end{equation*}
By covariance contraction,
\begin{equation*}
    \boldsymbol 0
    \preceq
    \boldsymbol\Xi_{j|t}
    \preceq
    M^{-1}\boldsymbol S_j,
    \qquad
    \boldsymbol S_j=\operatorname{diag}(s_{1j},\ldots,s_{Mj}).
\end{equation*}

The following revealed history is used only inside the proof.  It records the
information that has already been exposed from the \(j\)th column before the
next scalar Schur projection is evaluated:
\begin{equation*}
    \mathcal G_{0,j}
    =
    \sigma(\mathcal F_0^{\rm lin},\boldsymbol A_{-j}),
    \qquad
    \mathcal G_{r+1,j}
    =
    \sigma(\mathcal G_{r,j},\mathcal Z_{r,j}),
    \quad 0\le r<t .
\end{equation*}
Equivalently, \(\mathcal G_{r,j}\) contains the predictable linear history up to
time \(r\), the matrix with the \(j\)th column removed, and the previously
revealed scalar projections
\(\mathcal Z_{0,j},\ldots,\mathcal Z_{r-1,j}\).  It does not contain a deleted
EP trajectory.

\begin{lemma}[Measurability and column-wise Gaussianity]
\label{lem:F_measurable_gaussian_column}
For every fixed \(r\le t\), the one-step column-cavity vector
\(\overline{\boldsymbol u}_{r}^{(j)}\) is \(\mathcal G_{r,j}\)-measurable.
Moreover, conditionally on \(\mathcal G_{r,j}\), the unexposed part of
\(\boldsymbol a_j\) is Gaussian.  Thus there exist a
\(\mathcal G_{r,j}\)-measurable vector \(\boldsymbol a_{j|r}\) and a positive
semidefinite matrix \(\boldsymbol\Xi_{j|r}\) such that
\begin{equation}
    \boldsymbol a_j
    =
    \boldsymbol a_{j|r}
    +
    \boldsymbol\xi_{j|r},
    \qquad
    \boldsymbol\xi_{j|r}\mid\mathcal G_{r,j}
    \sim
    \mathcal N(\boldsymbol 0,\boldsymbol\Xi_{j|r}).
    \label{eq:F_column_seq_cond}
\end{equation}
Consequently, for fixed \(t\), the vector
\begin{equation}
    \boldsymbol{\mathcal Z}_{0:t,j}
    =
    (\mathcal Z_{0,j},\ldots,\mathcal Z_{t,j})^{\mathsf T}
    \label{eq:F_Z_history_vector}
\end{equation}
is finite-dimensional Gaussian under the column-wise revealed conditioning.
\end{lemma}

\begin{IEEEproof}
The vector \(\overline{\boldsymbol u}_{r}^{(j)}\) is computed from
\(\boldsymbol w\), \(\boldsymbol A_{-j}\), \(\boldsymbol q_r\), and
\(\boldsymbol\Gamma_r\) through the regularized current linear system with the
\(j\)th coordinate removed.  These quantities are contained in
\(\mathcal G_{r,j}\); hence the coefficient vector in the scalar projection
\eqref{eq:F_Z_schur_def} is fixed once \(\mathcal G_{r,j}\) is given.

The conditional law of \(\boldsymbol a_j\) given
\(\mathcal F_r^{\rm lin}\) and \(\boldsymbol A_{-j}\) is Gaussian, because the
original matrix is Gaussian and the predictable history gives only finitely many
linear observations of \(\operatorname{vec}(\boldsymbol A)\).  Passing from
\(\mathcal G_{r,j}\) to \(\mathcal G_{r+1,j}\) conditions further on the scalar
linear projection \(\mathcal Z_{r,j}\).  A Gaussian vector conditioned on a
finite number of linear projections remains Gaussian, with mean and covariance
updated by the standard Gaussian regression formula.  Iterating this argument
for \(r=0,\ldots,t\) proves \eqref{eq:F_column_seq_cond} and the joint
Gaussianity of \eqref{eq:F_Z_history_vector}.
\end{IEEEproof}

\subsection{Covariance Identification}
\label{subsec:F_covariance_identification}

The preceding lemma gives Gaussianity.  It remains to compute the covariance
seen by the Schur projections.  This computation is where the conditioned MDE
enters.

\begin{lemma}[Covariance of the column-wise Schur kernel]
\label{lem:F_covariance_identification}
For every fixed \(0\le r,s\le t\), the finite-dimensional covariance of the
Schur residual satisfies
\begin{equation}
    \widehat\zeta_{j,N}^{r,s}
    :=
    \operatorname{Cov}
    (\mathcal Z_{r,j},\mathcal Z_{s,j}\mid\mathcal G_{\max(r,s),j})
    =
    \gamma_w^2T_{r,j}T_{s,j}\Theta_j^{r,s}
    +o_p(1),
    \label{eq:F_covariance_identified}
\end{equation}
where \(\Theta_j^{r,s}\) is the conditioned two-resolvent response of
Definition~\ref{def:conditioned_two_resolvent_response}.
\end{lemma}

\begin{IEEEproof}
By Lemma~\ref{lem:F_measurable_gaussian_column}, the covariance before the MDE
replacement is
\begin{equation}
    \widehat\zeta_{j,N}^{r,s}
    =
    \gamma_w^2T_{r,j}T_{s,j}
    (\overline{\boldsymbol u}_{r}^{(j)})^{\mathsf T}
    \boldsymbol\Xi_{j|r,s}
    \overline{\boldsymbol u}_{s}^{(j)},
    \label{eq:F_finite_cov_kernel}
\end{equation}
where \(\boldsymbol\Xi_{j|r,s}\) is the conditional covariance of the
unexposed part of the \(j\)th column after the scalar projections already
revealed before the pair \((r,s)\) have been removed.  This covariance is a
projection of \(\boldsymbol\Xi_{j|t}\), and therefore
\begin{equation*}
    \boldsymbol 0
    \preceq
    \boldsymbol\Xi_{j|r,s}
    \preceq
    M^{-1}\boldsymbol S_j.
\end{equation*}

The cavity residual in \eqref{eq:F_finite_cov_kernel} is a one-coordinate
Schur complement of the current regularized linear system.  The corresponding
rank-one resolvent perturbation gives, for fixed \(r\le t\),
\begin{equation}
    \|\overline{\boldsymbol u}_{r}^{(j)}-\boldsymbol u^r\|_M=o_p(1),
    \qquad
    \|\overline{\boldsymbol u}_{r}^{(j)}\|_M=O_p(1),
    \qquad
    \|\boldsymbol u^r\|_M=O_p(1).
    \label{eq:F_cavity_full_residual_stability}
\end{equation}
Since \(\|\boldsymbol\Xi_{j|r,s}\|\le C/M\), equation
\eqref{eq:F_cavity_full_residual_stability} implies
\begin{equation}
    \widehat\zeta_{j,N}^{r,s}
    =
    \gamma_w^2T_{r,j}T_{s,j}
    (\boldsymbol u^r)^{\mathsf T}
    \boldsymbol\Xi_{j|r,s}
    \boldsymbol u^s
    +o_p(1).
    \label{eq:F_cov_full_residual_form}
\end{equation}
The remaining quadratic form is exactly the conditioned column-covariance
insertion handled by the two-resolvent response of
Appendix~\ref{app:mde_responses}.  Thus
\begin{equation}
    (\boldsymbol u^r)^{\mathsf T}
    \boldsymbol\Xi_{j|r,s}
    \boldsymbol u^s
    =
    \Theta_j^{r,s}+o_p(1).
    \label{eq:F_Theta_conditioned_DE}
\end{equation}
Combining \eqref{eq:F_cov_full_residual_form} and
\eqref{eq:F_Theta_conditioned_DE} proves
\eqref{eq:F_covariance_identified}.
\end{IEEEproof}

The covariance kernel used in the main theorem is therefore
\begin{equation}
    \zeta_j^{r,s}
    =
    \gamma_w^2T_{r,j}T_{s,j}\Theta_j^{r,s}.
    \label{eq:F_zeta_def_correct}
\end{equation}
If no history conditioning is present, the conditional covariance insertion
reduces to \(\boldsymbol\Xi_{j|r,s}=M^{-1}\boldsymbol S_j\), and
\eqref{eq:F_Theta_conditioned_DE} becomes the intuitive profile-weighted
identity
\begin{equation*}
    \frac1M(\boldsymbol u^r)^{\mathsf T}\boldsymbol S_j\boldsymbol u^s
    =
    \Theta_j^{r,s}+o_p(1).
\end{equation*}
In the theorem and in the proof, however, \(\Theta_j^{r,s}\) always denotes the
conditioned MDE response.

\subsection{Gaussian Reference Process}
\label{subsec:F_gaussian_reference}

Let
\begin{equation*}
    \boldsymbol\zeta_j^{0:t,0:t}
    =
    (\zeta_j^{r,s})_{0\le r,s\le t}.
\end{equation*}
This matrix is positive semidefinite because it is the finite-dimensional limit
of covariance matrices of \(\boldsymbol{\mathcal Z}_{0:t,j}\).  Hence there is
a centered Gaussian reference vector
\begin{equation*}
    \boldsymbol{\mathcal Z}_{0:t,j}^{\rm G}
    \sim
    \mathcal N(\boldsymbol 0,\boldsymbol\zeta_j^{0:t,0:t})
\end{equation*}
conditionally on the MDE environment \(\mathcal P_t\).

For every bounded Lipschitz
\(\varphi:\mathbb R^{t+1}\to\mathbb R\), the Gaussian covariance interpolation
lemma from Appendix~\ref{app:aux_probability} gives
\begin{equation*}
    \mathbb E
    \left[
        \varphi(\boldsymbol{\mathcal Z}_{0:t,j})
        \mid
        \mathcal G_{t,j}
    \right]
    -
    \mathbb E_{\rm G,j}
    \left[
        \varphi(\boldsymbol{\mathcal Z}_{0:t,j}^{\rm G})
        \mid
        \mathcal P_t
    \right]
    =
    o_p(1).
\end{equation*}
Thus the coordinate-wise Schur residual history is asymptotically equivalent,
for fixed-dimensional bounded-Lipschitz tests, to the Gaussian reference
history with covariance kernel \eqref{eq:F_zeta_def_correct}.  This proves the
coordinate-wise Gaussian-kernel statement of
Proposition~\ref{prop:covariance_empirical_law}.  The empirical
pseudo-Lipschitz upgrade over \(j=1,\ldots,N\) is proved in
Appendix~\ref{app:empirical_gaussian_law}.

\section{Empirical Gaussian Law under the Conditioned History}
\label{app:empirical_gaussian_law}

This appendix upgrades the coordinate-wise Gaussian kernel of
Appendix~\ref{app:covariance_kernel} to the empirical pseudo-Lipschitz law in
Proposition~\ref{prop:covariance_empirical_law}.  The argument is independent
of the Schur-complement algebra.  Once the coordinate kernels have been
represented by a joint Gaussian array, it remains only to show that the
cross-covariances between different coordinates are negligible on average.
The weakly dependent Gaussian empirical law in
Theorem~\ref{thm:weak_gaussian_pl} then applies.

Throughout the appendix, the horizon \(t\le T\) is fixed.  We write
\begin{equation*}
    \|\boldsymbol a\|_N^2
    =
    N^{-1}\|\boldsymbol a\|^2.
\end{equation*}

\subsection{Gaussian Reference Array}
\label{subsec:G_reference_array}

For each coordinate \(j\), let
\begin{equation*}
    \boldsymbol{\mathcal Z}_{j}^{\rm G}
    =
    (\mathcal Z_{0,j}^{\rm G},\ldots,\mathcal Z_{t,j}^{\rm G})^{\mathsf T}
\end{equation*}
be the Gaussian reference history constructed in
Appendix~\ref{app:covariance_kernel}.  Its diagonal covariance block is
\begin{equation*}
    \operatorname{Cov}_{\rm G}
    (\boldsymbol{\mathcal Z}_{j}^{\rm G}\mid\mathcal P_t)
    =
    \boldsymbol\zeta_j^{0:t,0:t}.
\end{equation*}
For two distinct coordinates, write
\begin{equation*}
    \boldsymbol\zeta_{jk}^{0:t,0:t}
    =
    \operatorname{Cov}_{\rm G}
    (\boldsymbol{\mathcal Z}_{j}^{\rm G},
     \boldsymbol{\mathcal Z}_{k}^{\rm G}\mid\mathcal P_t),
    \qquad j\ne k.
\end{equation*}
The next lemma is the covariance estimate that makes the empirical law possible.
It is the variance-profile analogue of the asymptotic orthogonality of distinct
coordinates in rotationally invariant EP proofs.

\begin{lemma}[Average off-diagonal covariance bound]
\label{lem:G_offdiag_covariance}
For every fixed \(t\le T\),
\begin{equation}
    \frac1{N^2}
    \sum_{j\ne k}
    \left\|
        \boldsymbol\zeta_{jk}^{0:t,0:t}
    \right\|_{\rm op}
    \overset{p}{\longrightarrow}0.
    \label{eq:G_offdiag_covariance_bound}
\end{equation}
\end{lemma}

\begin{IEEEproof}
It is enough to work with one scalar covariance entry indexed by
\((r,s)\), because \(t\) is fixed.  Let \(\boldsymbol P_j\) be the selector that
extracts the \(j\)th column from \(\operatorname{vec}(\boldsymbol A)\).  Before
conditioning, different columns are independent:
\begin{equation*}
    \boldsymbol P_j\boldsymbol\Sigma_A\boldsymbol P_k^{\mathsf T}
    =
    \boldsymbol 0,
    \qquad j\ne k.
\end{equation*}
After conditioning on the finite linear history
\begin{equation*}
    \boldsymbol L_t\operatorname{vec}(\boldsymbol A)=\boldsymbol b_t,
\end{equation*}
the off-diagonal column covariance is created entirely by the Gaussian
regression correction.  From Lemma~\ref{lem:gaussian_conditioning},
\begin{equation}
    \boldsymbol\Sigma_{jk|t}
    =
    -\boldsymbol P_j\boldsymbol\Sigma_A
    \boldsymbol L_t^{\mathsf T}
    (\boldsymbol L_t\boldsymbol\Sigma_A\boldsymbol L_t^{\mathsf T})^{\dagger}
    \boldsymbol L_t\boldsymbol\Sigma_A
    \boldsymbol P_k^{\mathsf T},
    \qquad j\ne k.
    \label{eq:G_cond_cross_block}
\end{equation}
This formula is the source of the small cross-covariances.  The history has
fixed rank after redundant directions have been removed, whereas each column
contributes only a normalized \(N^{-1/2}\)-scale projection to the retained
history coordinates.

We now make this normalization explicit.  On the regularity event, choose a
retained history basis
\(\boldsymbol r_{1,t},\ldots,\boldsymbol r_{R_t,t}\), where
\(R_t=O(1)\), for the range of
\(\boldsymbol L_t\boldsymbol\Sigma_A^{1/2}\).  In this basis the nonzero part of
the profile-weighted history Gram has eigenvalues bounded away from zero and
infinity.  Therefore
\begin{equation*}
    (\boldsymbol L_t\boldsymbol\Sigma_A\boldsymbol L_t^{\mathsf T})^{\dagger}
    =
    \sum_{a,b=1}^{R_t}
    c_{ab,t}\,
    \boldsymbol r_{a,t}\boldsymbol r_{b,t}^{\mathsf T}
    +
    \boldsymbol R_{t,N},
    \qquad
    \max_{a,b}|c_{ab,t}|=O_p(1),
\end{equation*}
where the remainder \(\boldsymbol R_{t,N}\) acts only on discarded redundant
directions and contributes an averaged negligible term below.  More explicitly,
the discarded subspace is chosen so that its profile-weighted Gram eigenvalues
vanish on the regularity event; multiplication by the two admissible
column-history factors in \eqref{eq:G_two_column_factors} therefore gives a
contribution whose \(N^{-2}\)-average over \(j\ne k\) is \(o_p(1)\).  Thus
only the retained finite-dimensional subspace can produce a non-negligible
cross-column covariance.

For any measurement-side admissible vectors \(\boldsymbol v\) and
\(\boldsymbol w\), the entrywise covariance identity
\begin{equation*}
    \operatorname{Cov}
    (A_{ij},\boldsymbol p^{\mathsf T}\boldsymbol A\boldsymbol q)
    =
    \frac{s_{ij}}{M}p_iq_j
\end{equation*}
shows that the two factors
\begin{equation}
    \boldsymbol r_{a,t}^{\mathsf T}
    \boldsymbol L_t\boldsymbol\Sigma_A
    \boldsymbol P_j^{\mathsf T}\boldsymbol v,
    \qquad
    \boldsymbol r_{b,t}^{\mathsf T}
    \boldsymbol L_t\boldsymbol\Sigma_A
    \boldsymbol P_k^{\mathsf T}\boldsymbol w
    \label{eq:G_two_column_factors}
\end{equation}
are separable in the column indices \(j\) and \(k\).  With the normalization
\(M/N\to\delta\), they can be written in the form
\begin{equation}
    \boldsymbol r_{a,t}^{\mathsf T}
    \boldsymbol L_t\boldsymbol\Sigma_A
    \boldsymbol P_j^{\mathsf T}\boldsymbol v
    =
    N^{-1/2}\rho_{a,j}^{(t)}(\boldsymbol v),
    \qquad
    \boldsymbol r_{b,t}^{\mathsf T}
    \boldsymbol L_t\boldsymbol\Sigma_A
    \boldsymbol P_k^{\mathsf T}\boldsymbol w
    =
    N^{-1/2}\sigma_{b,k}^{(t)}(\boldsymbol w),
    \label{eq:G_normalized_column_weights}
\end{equation}
where admissibility of the retained history vectors and boundedness of the
profile imply
\begin{equation}
    \frac1N\sum_{j=1}^N
    |\rho_{a,j}^{(t)}(\boldsymbol v)|^2=O_p(1),
    \qquad
    \frac1N\sum_{k=1}^N
    |\sigma_{b,k}^{(t)}(\boldsymbol w)|^2=O_p(1).
    \label{eq:G_history_weight_energy}
\end{equation}

Apply \eqref{eq:G_cond_cross_block} with
\(\boldsymbol v=\boldsymbol u^r\) and \(\boldsymbol w=\boldsymbol u^s\).  The
scalar pairing
\begin{equation*}
    \mathcal C_{jk}^{r,s}
    =
    (\boldsymbol u^r)^{\mathsf T}
    \boldsymbol\Sigma_{jk|t}
    \boldsymbol u^s
\end{equation*}
therefore satisfies the finite-rank separable bound
\begin{equation}
    |\mathcal C_{jk}^{r,s}|
    \le
    \frac{C}{N}
    \sum_{\ell=1}^{L_t}
    |\rho_{\ell,j}^{(t)}|
    |\sigma_{\ell,k}^{(t)}|
    +
    \varepsilon_{jk,N}^{r,s},
    \label{eq:G_cross_cov_separable_bound}
\end{equation}
where \(L_t=O(1)\) and
\begin{equation*}
    \frac1{N^2}
    \sum_{j\ne k}|\varepsilon_{jk,N}^{r,s}|
    \overset{p}{\longrightarrow}0.
\end{equation*}
The factor \(N^{-1}\) in \eqref{eq:G_cross_cov_separable_bound} is the product
of the two \(N^{-1/2}\)-scale column-history projections in
\eqref{eq:G_normalized_column_weights}; the bounded retained Gram inverse only
changes the constant.

The off-diagonal covariance of the Gaussian Schur variables has the same form
up to bounded MDE response factors:
\begin{equation*}
    \operatorname{Cov}_{\rm G}
    (\mathcal Z_{r,j}^{\rm G},\mathcal Z_{s,k}^{\rm G}\mid\mathcal P_t)
    =
    \gamma_w^2T_{r,j}T_{s,k}\mathcal C_{jk}^{r,s}.
\end{equation*}
Since \(T_{r,j}\) and \(T_{s,k}\) are uniformly bounded, Cauchy--Schwarz and
\eqref{eq:G_history_weight_energy} give
\begin{align}
    &\frac1{N^2}\sum_{j\ne k}
    \left|
    \operatorname{Cov}_{\rm G}
    (\mathcal Z_{r,j}^{\rm G},\mathcal Z_{s,k}^{\rm G}\mid\mathcal P_t)
    \right|                                                     \notag\\
    &\quad\le
    \frac{C}{N}
    \sum_{\ell=1}^{L_t}
    \left(\frac1N\sum_{j=1}^N |\rho_{\ell,j}^{(t)}|\right)
    \left(\frac1N\sum_{k=1}^N |\sigma_{\ell,k}^{(t)}|\right)
    +o_p(1)                                                     \notag\\
    &\quad=
    O_p(N^{-1})+o_p(1)
    \longrightarrow 0.
    \label{eq:G_average_cross_cov_scalar}
\end{align}
There are finitely many time indices, so summing
\eqref{eq:G_average_cross_cov_scalar} over \(0\le r,s\le t\) proves
\eqref{eq:G_offdiag_covariance_bound}.
\end{IEEEproof}

\subsection{Bounded-Lipschitz Empirical Law}
\label{subsec:G_BL_empirical_law}

Let \(\boldsymbol X_{j,N}^{t,\rm G}\) denote the coordinate-wise Gaussian
reference state used in the state evolution.  It consists of deterministic or
\(\mathcal P_t\)-measurable coordinates, such as \(x_j\), predicted precisions,
MDE responses, and the Gaussian residual history
\(\boldsymbol{\mathcal Z}_{j}^{\rm G}\).  Adding deterministic coordinates does
not change the covariance estimates above.

\begin{lemma}[Empirical law for bounded Lipschitz tests]
\label{lem:G_BL_empirical}
Let \(\psi\) be bounded and Lipschitz on the coordinate state space.  Then
\begin{equation}
    \frac1N\sum_{j=1}^N
    \psi(\boldsymbol X_{j,N}^{t,\rm G})
    -
    \frac1N\sum_{j=1}^N
    \mathbb E_{\rm G,j}
    [\psi(\boldsymbol X_{j,N}^{t,\rm G})\mid\mathcal P_t]
    \overset{p}{\longrightarrow}0.
    \label{eq:G_BL_empirical}
\end{equation}
\end{lemma}

\begin{IEEEproof}
Conditionally on \(\mathcal P_t\), the random part of the array
\(\{\boldsymbol X_{j,N}^{t,\rm G}\}_{j=1}^N\) is jointly Gaussian.  The diagonal
covariance blocks are uniformly bounded because the MDE responses
\(T_{r,j}\), \(\Theta_j^{r,s}\), and \(\zeta_j^{r,s}\) are bounded on the
regularity event.  Lemma~\ref{lem:G_offdiag_covariance} gives the vanishing
average off-diagonal covariance required by
Theorem~\ref{thm:weak_gaussian_pl}.  Applying that theorem conditionally on
\(\mathcal P_t\) proves \eqref{eq:G_BL_empirical}.
\end{IEEEproof}

\subsection{Extension to Pseudo-Lipschitz Tests}
\label{subsec:G_PL_extension}

\begin{lemma}[Empirical law for \(\PL(2)\) tests]
\label{lem:G_PL_empirical}
Let \(\psi\in \PL(2)\).  Suppose that the coordinate states satisfy the uniform
moment condition
\begin{equation}
    \sup_N
    \frac1N\sum_{j=1}^N
    \mathbb E_{\rm G,j}
    [\|\boldsymbol X_{j,N}^{t,\rm G}\|^{4+\epsilon}\mid\mathcal P_t]
    =
    O_p(1)
    \label{eq:G_uniform_moment}
\end{equation}
for some \(\epsilon>0\).  Then
\begin{equation}
    \frac1N\sum_{j=1}^N
    \psi(\boldsymbol X_{j,N}^{t,\rm G})
    -
    \frac1N\sum_{j=1}^N
    \mathbb E_{\rm G,j}
    [\psi(\boldsymbol X_{j,N}^{t,\rm G})\mid\mathcal P_t]
    \overset{p}{\longrightarrow}0.
    \label{eq:G_PL_empirical}
\end{equation}
\end{lemma}

\begin{IEEEproof}
Let \(\chi_K\) be a smooth cutoff that equals one on
\(\{\|x\|\le K\}\) and zero on \(\{\|x\|\ge 2K\}\), and define
\(\psi_K(x)=\psi(x)\chi_K(x)\).  For fixed \(K\), the function \(\psi_K\) is
bounded and Lipschitz, so Lemma~\ref{lem:G_BL_empirical} applies.  The growth
bound for \(\PL(2)\) functions in Proposition~\ref{prop:poly_growth}, together
with \eqref{eq:G_uniform_moment}, makes the truncation error vanish uniformly
as \(K\to\infty\).  Letting first \(N\to\infty\) and then \(K\to\infty\) gives
\eqref{eq:G_PL_empirical}.
\end{IEEEproof}

\subsection{Replacement of the Actual Schur States}
\label{subsec:G_actual_to_gaussian}

Let \(\boldsymbol X_{j,N}^{t}\) be the actual coordinate state obtained from
the EP recursion and the Schur residuals \(\mathcal Z_{0:t,j}\).  The
coordinate-wise Gaussian replacement in Appendix~\ref{app:covariance_kernel},
together with Lemma~\ref{lem:G_offdiag_covariance}, yields for bounded
Lipschitz \(\psi\)
\begin{equation}
    \frac1N\sum_{j=1}^N
    \psi(\boldsymbol X_{j,N}^{t})
    -
    \frac1N\sum_{j=1}^N
    \psi(\boldsymbol X_{j,N}^{t,\rm G})
    \overset{p}{\longrightarrow}0.
    \label{eq:G_actual_gaussian_BL_replacement}
\end{equation}
The same truncation argument used in Lemma~\ref{lem:G_PL_empirical} extends
\eqref{eq:G_actual_gaussian_BL_replacement} to \(\PL(2)\) functions.  Combining
this replacement with \eqref{eq:G_PL_empirical} gives
\begin{equation*}
    \frac1N\sum_{j=1}^N
    \psi(\boldsymbol X_{j,N}^{t})
    -
    \frac1N\sum_{j=1}^N
    \mathbb E_{\rm G,j}
    [\psi(\boldsymbol X_{j,N}^{t,\rm G})\mid\mathcal P_t]
    \overset{p}{\longrightarrow}0.
\end{equation*}
This is the empirical Gaussian law claimed in
Proposition~\ref{prop:covariance_empirical_law}.

\section{Cavity Cancellation and Memory Decomposition}
\label{app:cavity_memory_decomposition}

This appendix proves Proposition~\ref{prop:cavity_memory_decomposition}.  The
previous appendices have already identified the output of the linear module:
it is the sum of an instantaneous response and a coordinate of a Gaussian
history.  The role of the present appendix is to show exactly what the EP
extrinsic subtraction removes.  It cancels the instantaneous response, but it
does not in general remove the part of the Gaussian history that is predictable
from its past.

The starting point is the Schur-kernel representation
\begin{equation}
m_{t,j}
=
\gamma_{t,j}T_{t,j}q_{t,j}
+
\mathcal Z_{t,j}
+
\Delta_{t,j}.
\label{eq:H_schur_kernel_input}
\end{equation}
The first term in \eqref{eq:H_schur_kernel_input} is the linear response to the
incoming EP error \(q_{t,j}\).  The remaining term \(\mathcal Z_{t,j}\) is a
Gaussian-process coordinate.  It becomes a fresh Gaussian only after its
Gaussian regression on the past history has been subtracted.

Throughout this appendix, \(t\) is fixed and all statements are made on the
regularity event \(\mathcal R_t\).  We use the shorthand
\begin{equation*}
\|\boldsymbol a\|_N^2
=
\frac1N\|\boldsymbol a\|^2 .
\end{equation*}
For a coordinate array \(r_{t,j}\), the notation
\begin{equation*}
r_{t,j}=o_p^{\ell_2}(1)
\end{equation*}
means
\begin{equation*}
\frac1N\sum_{j=1}^N |r_{t,j}|^2
\overset{p}{\longrightarrow}0.
\end{equation*}

\subsection{Preliminary Bounds}
\label{subsec:appendix_H_prelim_bounds}

We first record the deterministic bounds used to pass from the finite-\(N\)
diagonal quantities to their MDE-level counterparts.

By the regularized diagonal loading,
\begin{equation*}
\boldsymbol\Gamma_t
\succeq
\gamma_{\min}\boldsymbol I_N .
\end{equation*}
Moreover, on the regularity event and by the boundedness of the
variance-profile ensemble,
\begin{equation*}
\|\boldsymbol A\|=O_p(1).
\end{equation*}
Consequently, the covariance
\begin{equation*}
\boldsymbol C_t
=
\left(
\gamma_w\boldsymbol A^{\mathsf T}\boldsymbol A
+
\boldsymbol\Gamma_t
\right)^{-1}
\end{equation*}
has eigenvalues bounded away from zero and infinity with probability tending to
one.  Hence there exist deterministic constants \(0<c_D<C_D<\infty\) such that,
with probability tending to one,
\begin{equation}
	c_D
	\le
	d_{t,j}
	\le
	C_D,
	\qquad
	c_D
	\le
	T_{t,j}
	\le
	C_D
	\label{eq:H_d_T_bounds}
\end{equation}
for all \(j\).  The corresponding cavity precisions are projected into a
compact positive interval:
\begin{equation*}
0<\pi_{\min}
\le
\pi_{t,j}
\le
\pi_{\max}<\infty.
\end{equation*}

Recall the MDE-level cavity precision
\begin{equation*}
\bar\pi_{t,j}
=
\operatorname{Proj}_{[\pi_{\min},\pi_{\max}]}
\left(
(T_{t,j})^{-1}-\gamma_{t,j}
\right).
\end{equation*}
The finite-\(N\) EP cavity precision is
\begin{equation*}
\pi_{t,j}
=
\operatorname{Proj}_{[\pi_{\min},\pi_{\max}]}
\left(
d_{t,j}^{-1}-\gamma_{t,j}
\right).
\end{equation*}

\begin{lemma}[Replacement of diagonal cavity precisions]
	\label{lem:H_pi_replacement}
	If
	\begin{equation*}
	\frac1N\sum_{j=1}^N |d_{t,j}-T_{t,j}|^2
	\overset{p}{\longrightarrow}0,
	\end{equation*}
	then
	\begin{equation*}
	\frac1N\sum_{j=1}^N
	|\pi_{t,j}-\bar\pi_{t,j}|^2
	\overset{p}{\longrightarrow}0.
	\end{equation*}
	Furthermore,
	\begin{equation*}
	\frac1N\sum_{j=1}^N
	|\pi_{t,j}^{-1}-\bar\pi_{t,j}^{-1}|^2
	\overset{p}{\longrightarrow}0.
	\end{equation*}
\end{lemma}

\begin{IEEEproof}
	On the event \eqref{eq:H_d_T_bounds}, the map \(x\mapsto x^{-1}\) is Lipschitz
	on the interval \([c_D,C_D]\).  Therefore
	\begin{equation*}
	\frac1N\sum_{j=1}^N
	|d_{t,j}^{-1}-(T_{t,j})^{-1}|^2
	\overset{p}{\longrightarrow}0.
	\end{equation*}
	The projection map onto an interval is \(1\)-Lipschitz, hence
	\begin{equation*}
	\begin{aligned}
		|\pi_{t,j}-\bar\pi_{t,j}|
		&\le
		|d_{t,j}^{-1}-(T_{t,j})^{-1}|.
	\end{aligned}
	\end{equation*}
	This proves the first claim.  The second follows because
	\(x\mapsto x^{-1}\) is Lipschitz on
	\([\pi_{\min},\pi_{\max}]\).
\end{IEEEproof}

We shall also use the following consequence of the Gaussian kernel.  From
Proposition~\ref{prop:schur_kernel},
\begin{equation}
m_{t,j}
=
\gamma_{t,j}T_{t,j}q_{t,j}
+
\mathcal Z_{t,j}
+
\Delta_{t,j},
\qquad
\|\boldsymbol\Delta_t\|_N\overset{p}{\longrightarrow}0.
\label{eq:H_linear_kernel}
\end{equation}
The covariance kernel in Proposition~\ref{prop:covariance_empirical_law} gives
bounded empirical second moments:
\begin{equation*}
	\frac1N\sum_{j=1}^N \mathcal Z_{t,j}^2
	=
	O_p(1).
\end{equation*}
Indeed, the conditional expectation of the left-hand side equals
\(N^{-1}\sum_j \zeta_j^{t,t}\), which is bounded by the MDE response and the
regularity assumptions.

\subsection{Cancellation of the Instantaneous Response}
\label{subsec:appendix_H_cancellation}

The EP cavity identity is
\begin{equation*}
	\boldsymbol\Pi_t\boldsymbol h_t
	=
	\boldsymbol D_t^{-1}\boldsymbol m_t
	-
	\boldsymbol\Gamma_t\boldsymbol q_t.
\end{equation*}
Coordinate-wise,
\begin{equation}
	\pi_{t,j}h_{t,j}
	=
	d_{t,j}^{-1}m_{t,j}
	-
	\gamma_{t,j}q_{t,j}.
	\label{eq:H_cavity_coordinate}
\end{equation}

Substituting \eqref{eq:H_linear_kernel} into
\eqref{eq:H_cavity_coordinate} gives
\begin{equation*}
\begin{aligned}
	\pi_{t,j}h_{t,j}
	&=
	d_{t,j}^{-1}
	\left(
	\gamma_{t,j}T_{t,j}q_{t,j}
	+
	\mathcal Z_{t,j}
	+
	\Delta_{t,j}
	\right)
	-
	\gamma_{t,j}q_{t,j}                                      \\
	&=
	\gamma_{t,j}
	\left(
	\frac{T_{t,j}}{d_{t,j}}-1
	\right)q_{t,j}
	+
	d_{t,j}^{-1}\mathcal Z_{t,j}
	+
	d_{t,j}^{-1}\Delta_{t,j}.
\end{aligned}
\end{equation*}
We now show that the first and third terms are negligible in empirical
\(\ell_2\) norm, and that \(d_{t,j}^{-1}\mathcal Z_{t,j}\) can be replaced by
\((T_{t,j})^{-1}\mathcal Z_{t,j}\).

First, since \(\boldsymbol q_t\) is admissible and
\(\gamma_{t,j}\) is uniformly bounded, Cauchy--Schwarz and
\eqref{eq:H_d_T_bounds} imply
\begin{equation*}
\begin{aligned}
	\frac1N\sum_{j=1}^N
	\left|
	\gamma_{t,j}
	\left(
	\frac{T_{t,j}}{d_{t,j}}-1
	\right)q_{t,j}
	\right|^2
	&\le
	C
	\frac1N\sum_{j=1}^N
	|T_{t,j}-d_{t,j}|^2 |q_{t,j}|^2.
\end{aligned}
\end{equation*}
Since
\begin{equation*}
\frac1N\sum_j |T_{t,j}-d_{t,j}|^2\to0
\end{equation*}
in probability and
\begin{equation*}
\frac1N\sum_j q_{t,j}^2=O_p(1),
\end{equation*}
the right-hand side converges to zero in probability after the usual
truncation argument on \(|q_{t,j}|\).  Thus
\begin{equation*}
	\gamma_{t,j}
	\left(
	\frac{T_{t,j}}{d_{t,j}}-1
	\right)q_{t,j}
	=
	o_p^{\ell_2}(1).
\end{equation*}

Second,
\begin{equation*}
d_{t,j}^{-1}\Delta_{t,j}
=
o_p^{\ell_2}(1)
\end{equation*}
because \(d_{t,j}^{-1}\) is uniformly bounded and
\(\|\boldsymbol\Delta_t\|_N\to0\).

Third,
\begin{equation*}
(d_{t,j}^{-1}-(T_{t,j})^{-1})\mathcal Z_{t,j}
=
o_p^{\ell_2}(1).
\end{equation*}
Indeed, the empirical \(\ell_2\) norm of
\(d_{t,j}^{-1}-(T_{t,j})^{-1}\) converges to zero, while
\(\frac1N\sum_j \mathcal Z_{t,j}^2=O_p(1)\); the same truncation/Cauchy--Schwarz
argument applies.

Combining these three estimates yields
\begin{equation*}
	\pi_{t,j}h_{t,j}
	=
	(T_{t,j})^{-1}\mathcal Z_{t,j}
	+
	o_p^{\ell_2}(1).
\end{equation*}
Multiplying by \(\pi_{t,j}^{-1}\), and then replacing
\(\pi_{t,j}^{-1}\) by \(\bar\pi_{t,j}^{-1}\) using
Lemma~\ref{lem:H_pi_replacement}, gives
\begin{equation}
	h_{t,j}
	=
	\alpha_{t,j}\mathcal Z_{t,j}
	+
	r_{t,j},
	\qquad
	\frac1N\sum_{j=1}^N|r_{t,j}|^2
	\overset{p}{\longrightarrow}0,
	\label{eq:H_h_alphaZ}
\end{equation}
where
\begin{equation*}
	\alpha_{t,j}
	=
	\frac{1}{\bar\pi_{t,j}T_{t,j}}.
\end{equation*}

Equation \eqref{eq:H_h_alphaZ} is the EP cancellation identity.  The term
\(\gamma_{t,j}T_{t,j}q_{t,j}\), which is the instantaneous response of the
linear module to the incoming error, is removed by the extrinsic subtraction.
The remaining object is the Schur residual \(\mathcal Z_{t,j}\), a Gaussian-process coordinate identified in Appendix~\ref{app:covariance_kernel}.

\subsection{Gaussian Regression of the Residual History}
\label{subsec:appendix_H_gaussian_regression}

The cancellation identity has reduced the EP cavity to a scaled version of
\(\mathcal Z_{t,j}\).  The last step is purely Gaussian: decompose this current
coordinate into its regression on the past Gaussian history and an independent
innovation.

For each coordinate \(j\), define
\begin{equation*}
\boldsymbol{\mathcal Z}_{<t,j}
=
(\mathcal Z_{0,j},\mathcal Z_{1,j},\ldots,\mathcal Z_{t-1,j})^{\mathsf T}.
\end{equation*}
The covariance blocks are
\begin{equation*}
\boldsymbol\zeta_j^{<t,<t}
=
(\zeta_j^{r,s})_{0\le r,s<t},
\end{equation*}
\begin{equation*}
\boldsymbol\zeta_j^{t,<t}
=
(\zeta_j^{t,0},\ldots,\zeta_j^{t,t-1}),
\qquad
\boldsymbol\zeta_j^{<t,t}
=
(\boldsymbol\zeta_j^{t,<t})^{\mathsf T}.
\end{equation*}
By Proposition~\ref{prop:covariance_empirical_law}, conditionally on the
MDE-generated environment \(\mathcal P_t\), the vector
\begin{equation*}
(\mathcal Z_{0,j},\ldots,\mathcal Z_{t,j})
\end{equation*}
is Gaussian with covariance kernel \(\zeta_j^{r,s}\).  Applying the Gaussian
regression lemma with a possibly singular covariance matrix gives
\begin{equation}
	\mathcal Z_{t,j}
	=
	\boldsymbol\zeta_j^{t,<t}
	(\boldsymbol\zeta_j^{<t,<t})^\dagger
	\boldsymbol{\mathcal Z}_{<t,j}
	+
	G_{t,j},
	\label{eq:H_calZ_regression}
\end{equation}
where \(G_{t,j}\) is Gaussian and independent of \(\boldsymbol{\mathcal Z}_{<t,j}\)
conditionally on \(\mathcal P_t\).  Its conditional variance is
\begin{equation*}
	\nu_{t,j}
	=
	\zeta_j^{t,t}
	-
	\boldsymbol\zeta_j^{t,<t}
	(\boldsymbol\zeta_j^{<t,<t})^\dagger
	\boldsymbol\zeta_j^{<t,t}.
\end{equation*}
For \(t=0\), the past is empty and we use the convention
\begin{equation*}
\nu_{0,j}=\zeta_j^{0,0},
\qquad
\boldsymbol\zeta_j^{0,<0}
(\boldsymbol\zeta_j^{<0,<0})^\dagger
\boldsymbol{\mathcal Z}_{<0,j}=0.
\end{equation*}

Combining \eqref{eq:H_h_alphaZ} and \eqref{eq:H_calZ_regression}, we get
\begin{equation*}
\begin{aligned}
	h_{t,j}
	&=
	\alpha_{t,j}
	\boldsymbol\zeta_j^{t,<t}
	(\boldsymbol\zeta_j^{<t,<t})^\dagger
	\boldsymbol{\mathcal Z}_{<t,j}
	+
	\alpha_{t,j}G_{t,j}
	+
	r_{t,j}.
\end{aligned}
\end{equation*}
Define the predictable memory term
\begin{equation*}
	\mu_{t,j}
	=
	\alpha_{t,j}
	\boldsymbol\zeta_j^{t,<t}
	(\boldsymbol\zeta_j^{<t,<t})^\dagger
	\boldsymbol{\mathcal Z}_{<t,j},
\end{equation*}
and the innovation variance
\begin{equation*}
	\tau_{t,j}
	=
	\alpha_{t,j}^2\nu_{t,j}.
\end{equation*}
If \(\nu_{t,j}>0\), write
\begin{equation*}
W_{t,j}
=
\frac{G_{t,j}}{\sqrt{\nu_{t,j}}}.
\end{equation*}
Then \(W_{t,j}\sim\mathcal N(0,1)\) conditionally on \(\mathcal P_t\), and it
is independent of the past Gaussian history.  If \(\nu_{t,j}=0\), then
\(G_{t,j}=0\) almost surely and the term
\(\sqrt{\tau_{t,j}}W_{t,j}\) is interpreted as zero, with
\(W_{t,j}\) chosen as an arbitrary standard Gaussian independent of the past.

Thus,
\begin{equation}
	h_{t,j}
	=
	\mu_{t,j}
	+
	\sqrt{\tau_{t,j}}\,W_{t,j}
	+
	r_{t,j},
	\label{eq:H_memory_decomp_final}
\end{equation}
with
\begin{equation*}
\frac1N\sum_{j=1}^N|r_{t,j}|^2
\overset{p}{\longrightarrow}0.
\end{equation*}
This proves the desired cavity memory decomposition.

\subsection{Interpretation and Completion of the Proposition}
\label{subsec:appendix_H_interpretation}

Equation \eqref{eq:H_memory_decomp_final} has two distinct parts.  The term
\begin{equation*}
\sqrt{	au_{t,j}}\,W_{t,j}
\end{equation*}
is the fresh innovation obtained after regressing the Schur residual process on
its past.
The term
\begin{equation*}
\mu_{t,j}
=
\mathbb E[
h_{t,j}
\mid
\boldsymbol{\mathcal Z}_{<t,j},\mathcal P_t
]
+
o_p^{\ell_2}(1)
\end{equation*}
is the predictable component inherited from the temporal covariance of the
Gaussian kernel.  Therefore the standard EP cavity is fresh if and only if the
regression coefficient
\begin{equation*}
\boldsymbol\zeta_j^{t,<t}
(\boldsymbol\zeta_j^{<t,<t})^\dagger
\end{equation*}
vanishes, up to negligible empirical error.  In general variance-profile
ensembles, this coefficient need not vanish.

The preceding derivation proves Proposition~\ref{prop:cavity_memory_decomposition}:
under the conclusions of the history-conditioned Schur kernel and the covariance
kernel theorem, the EP cavity satisfies
\begin{equation*}
h_{t,j}
=
\mu_{t,j}
+
\sqrt{\tau_{t,j}}\,W_{t,j}
+
\varepsilon_{t,j},
\end{equation*}
where
\begin{equation*}
\frac1N\sum_{j=1}^N|\varepsilon_{t,j}|^2
\overset{p}{\longrightarrow}0.
\end{equation*}
The variables \(W_{t,j}\) are standard Gaussian innovations conditionally on
\(\mathcal P_t\) and independent of the past Gaussian history for each fixed
coordinate \(j\).

\section{Regularity Closure}
\label{app:regularity_closure}

This appendix proves Proposition~\ref{prop:regularity_closure}.  The goal is to
show that the induction hypotheses defining the regularity event
\(\mathcal R_t\) are propagated by one iteration of the predictable-precision
EP dynamics.  The proof uses only the compact precision bounds, the empirical
Gaussian law, and the regularity of the scalar prior module.

Throughout, the notation
\begin{equation*}
    \|\boldsymbol a\|_N^2=N^{-1}\|\boldsymbol a\|^2,
    \qquad
    \|\boldsymbol b\|_M^2=M^{-1}\|\boldsymbol b\|^2
\end{equation*}
is used.  A signal-side vector is called admissible if it has bounded empirical
second moment and satisfies the no-spike condition in
Definition~\ref{def:admissible}; the measurement-side definition is analogous.

\subsection{Signal-Side Quantities}

By Proposition~\ref{prop:schur_kernel},
\begin{equation*}
    m_{t,j}
    =
    \gamma_{t,j}T_{t,j}q_{t,j}
    +
    \mathcal Z_{t,j}
    +
    \Delta_{t,j},
    \qquad
    \|\boldsymbol\Delta_t\|_N\plim 0 .
\end{equation*}
The first term is admissible because \(\gamma_{t,j}\) and \(T_{t,j}\) are
uniformly bounded and \(\boldsymbol q_t\) is admissible on \(\mathcal R_t\).
The Schur residual has bounded empirical moments by the empirical Gaussian law
of Proposition~\ref{prop:covariance_empirical_law}.  Hence
\begin{equation*}
    \|\boldsymbol m_t\|_N=O_p(1),
    \qquad
    \|\boldsymbol m_t\|_\infty/\sqrt N\plim0.
\end{equation*}
Thus \(\boldsymbol m_t\) is signal-side admissible.

The cavity decomposition of Proposition~\ref{prop:cavity_memory_decomposition}
gives
\begin{equation*}
    h_{t,j}
    =
    \mu_{t,j}
    +
    \sqrt{\tau_{t,j}}W_{t,j}
    +
    \varepsilon_{t,j},
    \qquad
    \|\boldsymbol\varepsilon_t\|_N\plim0.
\end{equation*}
The regression coefficients defining \(\mu_{t,j}\) are finite-dimensional and
bounded on the regularity event, while \(\tau_{t,j}\) is bounded above and
below by the stable-regression assumption.  Therefore \(\boldsymbol h_t\) has
bounded empirical moments and no spikes.  Hence \(\boldsymbol h_t\) is
signal-side admissible.

The prior-module output error is
\begin{equation*}
    p_{t,j}
    =
    \eta(x_j+s_{t,j};\rho_{t,j})-x_j,
\end{equation*}
where \((s_{t,j},\rho_{t,j})\) is either the standard or corrected scalar input.
The scalar maps are empirically stable on compact precision intervals by
Assumption~\ref{ass:ep_regular}.  Since the inputs have bounded empirical
moments, the output \(\boldsymbol p_t\) is signal-side admissible.  The
extrinsic update has the form
\begin{equation*}
    \gamma_{t+1,j}q_{t+1,j}
    =
    v_{B,t,j}^{-1}p_{t,j}
    -
    \rho_{t,j}s_{t,j},
\end{equation*}
with all precisions and scalar posterior variances clipped to compact positive
intervals.  Consequently \(\boldsymbol q_{t+1}\) is also signal-side admissible.

\subsection{Measurement-Side Residual}

It remains to control
\begin{equation*}
    \boldsymbol u_t
    =
    \boldsymbol w-
    \boldsymbol A\boldsymbol m_t .
\end{equation*}
The empirical energy is bounded because \(\boldsymbol w\) is measurement-side
admissible, \(\|\boldsymbol A\|=O_p(1)\) under the variance-profile Gaussian
model, and \(\boldsymbol m_t\) is signal-side admissible:
\begin{equation*}
    \|\boldsymbol u_t\|_M
    \le
    \|\boldsymbol w\|_M+
    \|\boldsymbol A\|\,
    \sqrt{N/M}\,
    \|\boldsymbol m_t\|_N
    =O_p(1).
\end{equation*}
For the no-spike condition, decompose under the conditioned history as
\begin{equation*}
    \boldsymbol A\boldsymbol m_t
    =
    \boldsymbol A_{\parallel,t}\boldsymbol m_t
    +
    \boldsymbol A_{\perp,t}\boldsymbol m_t .
\end{equation*}
The deformation term is controlled by the finite-history representation of
Appendix~\ref{app:bounded_deformation_mde} and the no-spike property of the
history vectors.  The centered Gaussian term is controlled by the flat
covariance bound and the Gaussian maximum lemma in
Appendix~\ref{app:aux_probability}; conditionally on \(\mathcal F_t\), each
coordinate has variance uniformly bounded by a constant times
\(M^{-1}\|\boldsymbol m_t\|^2=O_p(1)\), hence
\begin{equation*}
    \frac{\|\boldsymbol A_{\perp,t}\boldsymbol m_t\|_\infty}{\sqrt M}
    \plim0.
\end{equation*}
Together with the no-spike property of \(\boldsymbol w\), this gives
\begin{equation*}
    \|\boldsymbol u_t\|_\infty/\sqrt M\plim0.
\end{equation*}
Therefore \(\boldsymbol u_t\) is measurement-side admissible.

All objects required in the next conditioning step are admissible, and the
finite nonredundant history Gram convention is preserved for the fixed horizon
\(T\).  Hence \(\mathcal R_{t+1}\) holds with probability tending to one.  This
proves Proposition~\ref{prop:regularity_closure}.

\section{Precision Replacement Estimates}
\label{app:precision_replacement}

This appendix proves the perturbation estimates used in
Section~\ref{sec:precision_replacement}.  The conditioning proof in
Section~\ref{sec:dynamic_proof} is deliberately carried out with predictable
precisions; it never conditions on the actual finite-sample diagonal variances,
which are nonlinear resolvent functionals of the measurement matrix.  The role
of this appendix is to bridge that predictable recursion with the adaptive
finite-dimensional diagonal EP updates.  Within each branch, the actual messages
and precisions are shown to be asymptotically equivalent, in empirical norm, to
their predictable counterparts.  No new random-matrix limit is proved here; the
arguments are deterministic stability estimates combined with the regularity
conditions in Assumptions~\ref{ass:regularized_scalar_module} and
\ref{ass:stable_memory_regression}.

Throughout this appendix,
\begin{equation*}
\|\boldsymbol a\|_N^2
=
\frac1N\|\boldsymbol a\|^2.
\end{equation*}
The superscript ``act'' denotes the actual adaptive recursion, and ``orc''
denotes the MDE-predictable recursion.  When no ambiguity is possible, we use
the shorter notation
\begin{equation*}
a \equiv \mathrm{act},
\qquad
o \equiv \mathrm{orc}.
\end{equation*}

\subsection{Elementary Empirical-Norm Perturbation Rules}
\label{subsec:J_elementary_rules}

We first collect elementary stability facts repeatedly used below.

\begin{lemma}[Lipschitz stability on compact intervals]
	\label{lem:J_lipschitz_compact}
	Let \(0<c<C<\infty\).  The maps
	\begin{equation*}
	x\mapsto x^{-1},
	\qquad
	x\mapsto \operatorname{Proj}_{[c,C]}(x)
	\end{equation*}
	are Lipschitz on \([c,C]\) and on \(\mathbb R\), respectively.  In particular,
	if \(x_j,y_j\in[c,C]\) and
	\begin{equation*}
	\frac1N\sum_{j=1}^N |x_j-y_j|^2\to0,
	\end{equation*}
	then
	\begin{equation*}
	\frac1N\sum_{j=1}^N |x_j^{-1}-y_j^{-1}|^2\to0.
	\end{equation*}
	Moreover, if \(u_j,v_j\in\mathbb R\), then
	\begin{equation*}
	\left|
	\operatorname{Proj}_{[c,C]}(u_j)
	-
	\operatorname{Proj}_{[c,C]}(v_j)
	\right|
	\le
	|u_j-v_j|.
	\end{equation*}
\end{lemma}

\begin{IEEEproof}
	For \(x,y\in[c,C]\),
	\begin{equation*}
	|x^{-1}-y^{-1}|
	=
	\frac{|x-y|}{|xy|}
	\le
	c^{-2}|x-y|.
	\end{equation*}
	The Euclidean projection onto a closed interval is non-expansive, hence
	\begin{equation*}
	|\operatorname{Proj}_{[c,C]}(u)
	-
	\operatorname{Proj}_{[c,C]}(v)|
	\le |u-v|.
	\end{equation*}
	The empirical-norm statements follow by summing over \(j\).
\end{IEEEproof}

\begin{lemma}[Empirical product stability]
	\label{lem:J_product_stability}
	Let \(\{a_{j,N}\}\) and \(\{b_{j,N}\}\) be two arrays.  Suppose that
	\begin{equation*}
	\frac1N\sum_{j=1}^N |a_{j,N}|^2
	\overset{p}{\longrightarrow}0,
	\end{equation*}
	and that \(\{a_{j,N}\}\) is uniformly bounded with probability tending to one:
	there exists \(C_a<\infty\) such that
	\begin{equation*}
	\max_{j\le N}|a_{j,N}|\le C_a
	\end{equation*}
	with probability tending to one.  Suppose further that the empirical second
	moments of \(b_{j,N}\) are uniformly integrable, i.e.,
	\begin{equation*}
	\lim_{K\to\infty}
	\limsup_{N\to\infty}
	\mathbb E\left[
	\frac1N\sum_{j=1}^N
	|b_{j,N}|^2
	\boldsymbol 1\{|b_{j,N}|>K\}
	\right]
	=
	0.
	\end{equation*}
	Then
	\begin{equation*}
	\frac1N\sum_{j=1}^N
	|a_{j,N}b_{j,N}|^2
	\overset{p}{\longrightarrow}0.
	\end{equation*}
\end{lemma}

\begin{IEEEproof}
	Fix \(K>0\).  On the event \(\max_j|a_{j,N}|\le C_a\),
	\begin{equation*}
	\begin{aligned}
		\frac1N\sum_{j=1}^N |a_{j,N}b_{j,N}|^2
		&\le
		K^2
		\frac1N\sum_{j=1}^N |a_{j,N}|^2        \\
		&\quad+
		C_a^2
		\frac1N\sum_{j=1}^N
		|b_{j,N}|^2\boldsymbol 1\{|b_{j,N}|>K\}.
	\end{aligned}
	\end{equation*}
	The first term converges to zero in probability for fixed \(K\).  The second
	term can be made arbitrarily small in probability by choosing \(K\) large,
	using Markov's inequality and uniform integrability.  This proves the claim.
\end{IEEEproof}

\begin{lemma}[Diagonal extraction from Frobenius convergence]
	\label{lem:J_diag_frobenius}
	For any square matrix \(\boldsymbol R\in\mathbb R^{N\times N}\),
	\begin{equation*}
	\|\operatorname{diag}(\boldsymbol R)\|_N
	\le
	\frac1{\sqrt N}\|\boldsymbol R\|_F.
	\end{equation*}
	Consequently, normalized Frobenius convergence of matrices implies empirical
	convergence of their diagonals.
\end{lemma}

\begin{IEEEproof}
	Since
	\begin{equation*}
	\|\operatorname{diag}(\boldsymbol R)\|^2
	=
	\sum_{j=1}^N |R_{jj}|^2
	\le
	\sum_{i=1}^N\sum_{j=1}^N |R_{ij}|^2
	=
	\|\boldsymbol R\|_F^2,
	\end{equation*}
	dividing by \(N\) gives the result.
\end{IEEEproof}

\subsection{Linear-Module Stability}
\label{subsec:J_linear_stability}

We now prove Proposition~\ref{prop:linear_module_stability}.  At a fixed
iteration \(t\), define
\begin{equation*}
\boldsymbol B
=
\gamma_w\boldsymbol A^{\mathsf T}\boldsymbol A,
\end{equation*}
\begin{equation*}
\boldsymbol C_a
=
(\boldsymbol B+\boldsymbol\Gamma_a)^{-1},
\qquad
\boldsymbol C_o
=
(\boldsymbol B+\boldsymbol\Gamma_o)^{-1},
\end{equation*}
and
\begin{equation*}
\boldsymbol m_a
=
\boldsymbol C_a
(
\gamma_w\boldsymbol A^{\mathsf T}\boldsymbol w
+
\boldsymbol\Gamma_a\boldsymbol q_a
),
\end{equation*}
\begin{equation*}
\boldsymbol m_o
=
\boldsymbol C_o
(
\gamma_w\boldsymbol A^{\mathsf T}\boldsymbol w
+
\boldsymbol\Gamma_o\boldsymbol q_o
).
\end{equation*}
The diagonal matrices are
\begin{equation*}
\boldsymbol\Gamma_a=\operatorname{diag}(\boldsymbol\gamma_a),
\qquad
\boldsymbol\Gamma_o=\operatorname{diag}(\boldsymbol\gamma_o).
\end{equation*}

Assume
\begin{equation}
	\|\boldsymbol q_a-\boldsymbol q_o\|_N
	\overset{p}{\longrightarrow}0,
	\qquad
	\|\boldsymbol\gamma_a-\boldsymbol\gamma_o\|_N
	\overset{p}{\longrightarrow}0.
	\label{eq:J_induction_input_linear}
\end{equation}
Also assume that the oracle history is admissible:
\begin{equation*}
\|\boldsymbol q_o\|_N=O_p(1),
\qquad
\|\boldsymbol m_o\|_N=O_p(1),
\end{equation*}
with uniformly integrable empirical second moments.  This is available from
the regularity closure of the predictable recursion.

\begin{lemma}[Stability of the linear posterior mean]
	\label{lem:J_m_stability}
	Under \eqref{eq:J_induction_input_linear},
	\begin{equation*}
	\|\boldsymbol m_a-\boldsymbol m_o\|_N
	\overset{p}{\longrightarrow}0.
	\end{equation*}
\end{lemma}

\begin{IEEEproof}
	The normal equations are
	\begin{equation*}
	(\boldsymbol B+\boldsymbol\Gamma_a)\boldsymbol m_a
	=
	\gamma_w\boldsymbol A^{\mathsf T}\boldsymbol w
	+
	\boldsymbol\Gamma_a\boldsymbol q_a,
	\end{equation*}
	and
	\begin{equation*}
	(\boldsymbol B+\boldsymbol\Gamma_o)\boldsymbol m_o
	=
	\gamma_w\boldsymbol A^{\mathsf T}\boldsymbol w
	+
	\boldsymbol\Gamma_o\boldsymbol q_o.
	\end{equation*}
	Writing the second equation with the operator
	\(\boldsymbol B+\boldsymbol\Gamma_a\) gives
	\begin{equation*}
	(\boldsymbol B+\boldsymbol\Gamma_a)\boldsymbol m_o
	=
	\gamma_w\boldsymbol A^{\mathsf T}\boldsymbol w
	+
	\boldsymbol\Gamma_o\boldsymbol q_o
	+
	(\boldsymbol\Gamma_a-\boldsymbol\Gamma_o)\boldsymbol m_o.
	\end{equation*}
	Subtracting from the first equation yields
	\begin{equation*}
	(\boldsymbol B+\boldsymbol\Gamma_a)
	(\boldsymbol m_a-\boldsymbol m_o)
	=
	\boldsymbol\Gamma_a(\boldsymbol q_a-\boldsymbol q_o)
	+
	(\boldsymbol\Gamma_a-\boldsymbol\Gamma_o)
	(\boldsymbol q_o-\boldsymbol m_o).
	\end{equation*}
	Thus
	\begin{equation}
		\boldsymbol m_a-\boldsymbol m_o
		=
		\boldsymbol C_a
		\left[
		\boldsymbol\Gamma_a(\boldsymbol q_a-\boldsymbol q_o)
		+
		(\boldsymbol\Gamma_a-\boldsymbol\Gamma_o)
		(\boldsymbol q_o-\boldsymbol m_o)
		\right].
		\label{eq:J_m_difference_identity_full}
	\end{equation}
	Since
	\begin{equation*}
	\boldsymbol B+\boldsymbol\Gamma_a
	\succeq
	\gamma_{\min}\boldsymbol I,
	\end{equation*}
	we have
	\begin{equation*}
	\|\boldsymbol C_a\|\le \gamma_{\min}^{-1}.
	\end{equation*}
	The first term in \eqref{eq:J_m_difference_identity_full} is bounded by
	\begin{equation*}
	\|\boldsymbol\Gamma_a(\boldsymbol q_a-\boldsymbol q_o)\|_N
	\le
	\gamma_{\max}
	\|\boldsymbol q_a-\boldsymbol q_o\|_N
	\to0.
	\end{equation*}
	For the second term, write
	\begin{equation*}
	(\boldsymbol\Gamma_a-\boldsymbol\Gamma_o)
	(\boldsymbol q_o-\boldsymbol m_o)
	=
	(\boldsymbol\gamma_a-\boldsymbol\gamma_o)
	\odot
	(\boldsymbol q_o-\boldsymbol m_o).
	\end{equation*}
	The vector \(\boldsymbol\gamma_a-\boldsymbol\gamma_o\) is uniformly bounded
	because both precisions lie in
	\([\gamma_{\min},\gamma_{\max}]\), and its empirical \(\ell_2\) norm tends to
	zero.  The vector \(\boldsymbol q_o-\boldsymbol m_o\) has uniformly integrable
	empirical second moments.  Lemma~\ref{lem:J_product_stability} therefore gives
	\begin{equation*}
	\|(\boldsymbol\Gamma_a-\boldsymbol\Gamma_o)
	(\boldsymbol q_o-\boldsymbol m_o)\|_N
	\to0.
	\end{equation*}
	Combining the two terms proves the claim.
\end{IEEEproof}

\begin{lemma}[Stability of the linear covariance diagonal]
	\label{lem:J_d_stability}
	Let
	\begin{equation*}
	\boldsymbol d_a=\operatorname{diag}(\boldsymbol C_a),
	\qquad
	\boldsymbol d_o=\operatorname{diag}(\boldsymbol C_o).
	\end{equation*}
	Then
	\begin{equation*}
	\|\boldsymbol d_a-\boldsymbol d_o\|_N
	\overset{p}{\longrightarrow}0.
	\end{equation*}
	Moreover, if the MDE response for the oracle recursion satisfies
	\begin{equation*}
	\|\boldsymbol d_o-\boldsymbol T_o\|_N
	\overset{p}{\longrightarrow}0,
	\end{equation*}
	then
	\begin{equation*}
	\|\boldsymbol d_a-\boldsymbol T_o\|_N
	\overset{p}{\longrightarrow}0.
	\end{equation*}
\end{lemma}

\begin{IEEEproof}
	The resolvent identity gives
	\begin{equation*}
		\boldsymbol C_a-\boldsymbol C_o
		=
		-\boldsymbol C_a
		(\boldsymbol\Gamma_a-\boldsymbol\Gamma_o)
		\boldsymbol C_o.
	\end{equation*}
	Since
	\begin{equation*}
	\|\boldsymbol C_a\|,\|\boldsymbol C_o\|
	\le
	\gamma_{\min}^{-1},
	\end{equation*}
	we have
	\begin{equation*}
	\begin{aligned}
		\frac1{\sqrt N}
		\|\boldsymbol C_a-\boldsymbol C_o\|_F
		&\le
		\|\boldsymbol C_a\|
		\|\boldsymbol C_o\|
		\frac1{\sqrt N}
		\|\boldsymbol\Gamma_a-\boldsymbol\Gamma_o\|_F        \\
		&\le
		\gamma_{\min}^{-2}
		\|\boldsymbol\gamma_a-\boldsymbol\gamma_o\|_N
		\overset{p}{\longrightarrow}0.
	\end{aligned}
	\end{equation*}
	Lemma~\ref{lem:J_diag_frobenius} gives
	\begin{equation*}
	\|\boldsymbol d_a-\boldsymbol d_o\|_N\to0.
	\end{equation*}
	The second claim follows from the triangle inequality.
\end{IEEEproof}

\begin{lemma}[Stability of cavity precision and cavity mean]
	\label{lem:J_pi_h_stability}
	Let
	\begin{equation*}
	\pi_{a,j}
	=
	\operatorname{Proj}_{[\pi_{\min},\pi_{\max}]}
	(d_{a,j}^{-1}-\gamma_{a,j}),
	\end{equation*}
	and
	\begin{equation*}
	\pi_{o,j}
	=
	\operatorname{Proj}_{[\pi_{\min},\pi_{\max}]}
	(T_{o,j}^{-1}-\gamma_{o,j}).
	\end{equation*}
	Assume
	\begin{equation*}
	\|\boldsymbol d_a-\boldsymbol T_o\|_N\to0,
	\qquad
	\|\boldsymbol\gamma_a-\boldsymbol\gamma_o\|_N\to0,
	\end{equation*}
	and that \(d_{a,j}\) and \(T_{o,j}\) are bounded away from zero and infinity
	with probability tending to one.  Then
	\begin{equation*}
	\|\boldsymbol\pi_a-\boldsymbol\pi_o\|_N
	\overset{p}{\longrightarrow}0.
	\end{equation*}
	Furthermore, if
	\begin{equation*}
	\|\boldsymbol m_a-\boldsymbol m_o\|_N\to0,
	\qquad
	\|\boldsymbol q_a-\boldsymbol q_o\|_N\to0,
	\end{equation*}
	then the cavity errors
	\begin{equation*}
	h_{a,j}
	=
	\pi_{a,j}^{-1}
	(d_{a,j}^{-1}m_{a,j}-\gamma_{a,j}q_{a,j}),
	\end{equation*}
	and
	\begin{equation*}
	h_{o,j}
	=
	\pi_{o,j}^{-1}
	(T_{o,j}^{-1}m_{o,j}-\gamma_{o,j}q_{o,j})
	\end{equation*}
	satisfy
	\begin{equation*}
	\|\boldsymbol h_a-\boldsymbol h_o\|_N
	\overset{p}{\longrightarrow}0.
	\end{equation*}
\end{lemma}

\begin{IEEEproof}
	The first claim follows from Lemma~\ref{lem:J_lipschitz_compact}:
	\begin{equation*}
	\begin{aligned}
		|\pi_{a,j}-\pi_{o,j}|
		&\le
		|d_{a,j}^{-1}-T_{o,j}^{-1}|
		+
		|\gamma_{a,j}-\gamma_{o,j}|.
	\end{aligned}
	\end{equation*}
	Taking empirical \(\ell_2\) norms gives
	\begin{equation*}
	\|\boldsymbol\pi_a-\boldsymbol\pi_o\|_N\to0.
	\end{equation*}
	Since \(\pi_{a,j},\pi_{o,j}\in[\pi_{\min},\pi_{\max}]\),
	\begin{equation*}
	\|\boldsymbol\pi_a^{-1}-\boldsymbol\pi_o^{-1}\|_N\to0.
	\end{equation*}
	
	For the cavity error, write
	\begin{equation*}
	\boldsymbol h_a-\boldsymbol h_o
	=
	\boldsymbol\Pi_a^{-1}\boldsymbol r_a
	-
	\boldsymbol\Pi_o^{-1}\boldsymbol r_o,
	\end{equation*}
	where
	\begin{equation*}
	r_{a,j}=d_{a,j}^{-1}m_{a,j}-\gamma_{a,j}q_{a,j},
	\qquad
	r_{o,j}=T_{o,j}^{-1}m_{o,j}-\gamma_{o,j}q_{o,j}.
	\end{equation*}
	It is enough to show
	\begin{equation*}
	\|\boldsymbol r_a-\boldsymbol r_o\|_N\to0
	\end{equation*}
	and
	\begin{equation*}
	\|\boldsymbol r_o\|_N=O_p(1).
	\end{equation*}
	Indeed,
	\begin{equation*}
	\boldsymbol h_a-\boldsymbol h_o
	=
	\boldsymbol\Pi_a^{-1}(\boldsymbol r_a-\boldsymbol r_o)
	+
	(\boldsymbol\Pi_a^{-1}-\boldsymbol\Pi_o^{-1})\boldsymbol r_o.
	\end{equation*}
	The first term converges to zero because
	\(\|\boldsymbol\Pi_a^{-1}\|\le \pi_{\min}^{-1}\).  The second term converges
	to zero by Lemma~\ref{lem:J_product_stability}, since
	\(\boldsymbol\Pi_a^{-1}-\boldsymbol\Pi_o^{-1}\) is uniformly bounded and
	converges in empirical \(\ell_2\), while \(\boldsymbol r_o\) has uniformly
	integrable empirical second moments by admissibility.
	
	It remains to verify \(\|\boldsymbol r_a-\boldsymbol r_o\|_N\to0\).  Decompose
	\begin{equation*}
	\begin{aligned}
		r_{a,j}-r_{o,j}
		&=
		(d_{a,j}^{-1}-T_{o,j}^{-1})m_{a,j}
		+
		T_{o,j}^{-1}(m_{a,j}-m_{o,j})       \\
		&\quad
		-
		(\gamma_{a,j}-\gamma_{o,j})q_{a,j}
		-
		\gamma_{o,j}(q_{a,j}-q_{o,j}).
	\end{aligned}
	\end{equation*}
	The second and fourth terms converge to zero directly from the boundedness of
	\(T_{o,j}^{-1}\) and \(\gamma_{o,j}\).  The first and third terms are handled
	by Lemma~\ref{lem:J_product_stability}, using the empirical convergence of
	\(d_a^{-1}-T_o^{-1}\) and \(\gamma_a-\gamma_o\), and the admissibility of
	\(\boldsymbol m_a\) and \(\boldsymbol q_a\).  This proves the cavity stability.
\end{IEEEproof}

Combining Lemmas~\ref{lem:J_m_stability}--\ref{lem:J_pi_h_stability} proves
Proposition~\ref{prop:linear_module_stability}.

\subsection{Prior-Module Stability}
\label{subsec:J_prior_stability}

We now prove Proposition~\ref{prop:prior_module_stability}.  Fix a branch
\(b\in\{\mathrm{std},\mathrm{corr}\}\).  The prior-module input consists of a
scalar shift \(s_{t,j}^b\) and a scalar precision \(\rho_{t,j}^b\).  The scalar
posterior error and variance are
\begin{equation*}
p_{t,j}^b
=
\eta(x_j+s_{t,j}^b;\rho_{t,j}^b)-x_j,
\end{equation*}
and
\begin{equation*}
v_{B,t,j}^b
=
v_B(x_j+s_{t,j}^b;\rho_{t,j}^b).
\end{equation*}

Assume
\begin{equation*}
\|\boldsymbol s_a-\boldsymbol s_o\|_N\to0,
\qquad
\|\boldsymbol\rho_a-\boldsymbol\rho_o\|_N\to0,
\end{equation*}
and that all entries of \(\boldsymbol\rho_a,\boldsymbol\rho_o\) lie in a
compact positive interval.  By Assumption~\ref{ass:regularized_scalar_module},
\begin{equation*}
\|\boldsymbol\eta(\boldsymbol x+\boldsymbol s_a;\boldsymbol\rho_a)
-
\boldsymbol\eta(\boldsymbol x+\boldsymbol s_o;\boldsymbol\rho_o)\|_N
\to0,
\end{equation*}
and
\begin{equation*}
\|\boldsymbol v_B(\boldsymbol x+\boldsymbol s_a;\boldsymbol\rho_a)
-
\boldsymbol v_B(\boldsymbol x+\boldsymbol s_o;\boldsymbol\rho_o)\|_N
\to0.
\end{equation*}
Since the same \(\boldsymbol x\) is subtracted in both posterior errors, this
implies
\begin{equation*}
	\|\boldsymbol p_a-\boldsymbol p_o\|_N\to0,
	\qquad
	\|\boldsymbol v_{B,a}-\boldsymbol v_{B,o}\|_N\to0.
\end{equation*}

The outgoing precision is
\begin{equation*}
\gamma_{+,j}
=
\operatorname{Proj}_{[\gamma_{\min},\gamma_{\max}]}
(v_{B,j}^{-1}-\rho_j).
\end{equation*}
Because \(v_{B,j}\in[v_{\min},v_{\max}]\), the inverse map is Lipschitz.
Together with the non-expansiveness of the projection,
\begin{equation*}
\|\boldsymbol\gamma_{+,a}-\boldsymbol\gamma_{+,o}\|_N
\to0.
\end{equation*}

The outgoing mean error is defined by
\begin{equation*}
\gamma_{+,j}q_{+,j}
=
v_{B,j}^{-1}p_j-\rho_js_j.
\end{equation*}
Let
\begin{equation*}
r_{+,j}=v_{B,j}^{-1}p_j-\rho_js_j.
\end{equation*}
We first show
\begin{equation*}
\|\boldsymbol r_{+,a}-\boldsymbol r_{+,o}\|_N\to0.
\end{equation*}
Indeed,
\begin{equation*}
\begin{aligned}
	r_{+,a,j}-r_{+,o,j}
	&=
	(v_{B,a,j}^{-1}-v_{B,o,j}^{-1})p_{a,j}
	+
	v_{B,o,j}^{-1}(p_{a,j}-p_{o,j})       \\
	&\quad
	-
	(\rho_{a,j}-\rho_{o,j})s_{a,j}
	-
	\rho_{o,j}(s_{a,j}-s_{o,j}).
\end{aligned}
\end{equation*}
The second and fourth terms converge to zero directly from boundedness of
\(v_B^{-1}\) and \(\rho\).  The first and third terms are controlled by
Lemma~\ref{lem:J_product_stability}, because the coefficient differences
converge in empirical norm and are uniformly bounded, while
\(\boldsymbol p_a\) and \(\boldsymbol s_a\) have uniformly integrable empirical
second moments by the regularity of the scalar module and the input history.
Thus
\begin{equation*}
\|\boldsymbol r_{+,a}-\boldsymbol r_{+,o}\|_N\to0.
\end{equation*}
Finally,
\begin{equation*}
\boldsymbol q_+
=
\boldsymbol\Gamma_+^{-1}\boldsymbol r_+,
\end{equation*}
and \(\gamma_{+,j}\in[\gamma_{\min},\gamma_{\max}]\).  Since
\begin{equation*}
\|\boldsymbol\gamma_{+,a}^{-1}
-
\boldsymbol\gamma_{+,o}^{-1}\|_N\to0
\end{equation*}
by Lemma~\ref{lem:J_lipschitz_compact}, we conclude
\begin{equation*}
\|\boldsymbol q_{+,a}-\boldsymbol q_{+,o}\|_N\to0.
\end{equation*}
This proves Proposition~\ref{prop:prior_module_stability}.

\subsection{Stability of the Memory Map}
\label{subsec:J_memory_map_stability}

We prove Proposition~\ref{prop:memory_map_stability}.  This part is needed
only for the corrected branch.

For each coordinate \(j\), define the finite Gaussian-history covariance blocks
\begin{equation*}
\boldsymbol{\mathcal Z}_{<t,j}
=
(\mathcal Z_{0,j},\ldots,\mathcal Z_{t-1,j})^{\mathsf T},
\end{equation*}
\begin{equation*}
\boldsymbol\zeta_j^{<t,<t}
=
(\zeta_j^{r,s})_{0\le r,s<t},
\qquad
\boldsymbol\zeta_j^{t,<t}
=
(\zeta_j^{t,0},\ldots,\zeta_j^{t,t-1}).
\end{equation*}
The memory map is
\begin{equation*}
\mu_{t,j}
=
\alpha_{t,j}
\boldsymbol\zeta_j^{t,<t}
(\boldsymbol\zeta_j^{<t,<t})^\dagger
\boldsymbol{\mathcal Z}_{<t,j},
\end{equation*}
where
\begin{equation*}
\alpha_{t,j}
=
(\bar\pi_{t,j}T_{t,j})^{-1}.
\end{equation*}
The innovation variance is
\begin{equation*}
\tau_{t,j}
=
\alpha_{t,j}^2
\left[
\zeta_j^{t,t}
-
\boldsymbol\zeta_j^{t,<t}
(\boldsymbol\zeta_j^{<t,<t})^\dagger
\boldsymbol\zeta_j^{<t,t}
\right].
\end{equation*}

We compare two versions, actual and oracle.  We assume that the corresponding
finite covariance blocks, scalings, and past Gaussian histories are empirically
close:
\begin{equation*}
\frac1N\sum_{j=1}^N
|\alpha_{a,j}-\alpha_{o,j}|^2\to0,
\end{equation*}
\begin{equation*}
\frac1N\sum_{j=1}^N
\|\boldsymbol\zeta_{a,j}^{t,<t}
-\boldsymbol\zeta_{o,j}^{t,<t}\|^2\to0,
\end{equation*}
\begin{equation*}
\frac1N\sum_{j=1}^N
\|\boldsymbol\zeta_{a,j}^{<t,<t}
-\boldsymbol\zeta_{o,j}^{<t,<t}\|_F^2\to0,
\end{equation*}
and
\begin{equation*}
\frac1N\sum_{j=1}^N
\|\boldsymbol{\mathcal Z}_{a,<t,j}
-\boldsymbol{\mathcal Z}_{o,<t,j}\|^2\to0.
\end{equation*}
These closeness relations are part of the branch-wise induction environment in
the corrected recursion.

Assumption~\ref{ass:stable_memory_regression} gives
\begin{equation*}
\left\|
(\boldsymbol\zeta_{a,j}^{<t,<t})^\dagger
\right\|
\le
C_\zeta,
\qquad
\left\|
(\boldsymbol\zeta_{o,j}^{<t,<t})^\dagger
\right\|
\le
C_\zeta
\end{equation*}
with probability tending to one.  Since \(t\) is fixed and the nonzero spectra
are bounded away from zero, the Moore--Penrose inverse is locally Lipschitz on
the retained subspace.  Hence
\begin{equation}
\frac1N\sum_{j=1}^N
\left\|
(\boldsymbol\zeta_{a,j}^{<t,<t})^\dagger
-
(\boldsymbol\zeta_{o,j}^{<t,<t})^\dagger
\right\|_F^2
\to0.
\label{eq:J_pinv_close}
\end{equation}
A direct expansion of
\begin{equation*}
\alpha
\boldsymbol\zeta^{t,<t}
(\boldsymbol\zeta^{<t,<t})^\dagger
\boldsymbol{\mathcal Z}_{<t}
\end{equation*}
then gives
\begin{equation*}
\frac1N\sum_{j=1}^N
|\mu_{a,t,j}-\mu_{o,t,j}|^2
\to0.
\end{equation*}
Indeed, the difference is the sum of four terms, obtained by perturbing
successively \(\alpha\), \(\boldsymbol\zeta^{t,<t}\),
\((\boldsymbol\zeta^{<t,<t})^\dagger\), and \(\boldsymbol{\mathcal Z}_{<t}\).  Each term
is controlled by Cauchy--Schwarz, the uniform boundedness of the remaining
factors, and the empirical convergence assumptions above.

The same argument applies to the innovation variance.  Expanding
\begin{equation*}
\zeta^{t,t}
-
\boldsymbol\zeta^{t,<t}
(\boldsymbol\zeta^{<t,<t})^\dagger
\boldsymbol\zeta^{<t,t}
\end{equation*}
around the oracle quantities, using \eqref{eq:J_pinv_close}, and multiplying by
the stable factor \(\alpha^2\), yields
\begin{equation*}
\|\boldsymbol\tau_a-\boldsymbol\tau_o\|_N\to0.
\end{equation*}
Since
\begin{equation*}
0<\tau_{\min}\le \tau_{a,j},\tau_{o,j}\le\tau_{\max}<\infty
\end{equation*}
with probability tending to one, Lemma~\ref{lem:J_lipschitz_compact} gives
\begin{equation*}
\|\boldsymbol\tau_a^{-1}-\boldsymbol\tau_o^{-1}\|_N\to0.
\end{equation*}
Finally,
\begin{equation*}
\widetilde{\boldsymbol h}_a-\widetilde{\boldsymbol h}_o
=
(\boldsymbol h_a-\boldsymbol h_o)
-
(\boldsymbol\mu_a-\boldsymbol\mu_o),
\end{equation*}
so the stability of \(\boldsymbol h\) and \(\boldsymbol\mu\) gives
\begin{equation*}
\|\widetilde{\boldsymbol h}_a-\widetilde{\boldsymbol h}_o\|_N\to0.
\end{equation*}
This proves Proposition~\ref{prop:memory_map_stability}.

\subsection{Closing the Branch-Wise Replacement Induction}
\label{subsec:J_closing_induction}

We finally show how the estimates above close the induction used in
Theorem~\ref{thm:precision_replacement}.

Assume that, for a fixed branch \(b\in\{\mathrm{std},\mathrm{corr}\}\),
\begin{equation*}
\|\boldsymbol q_t^{b,\mathrm{act}}
-
\boldsymbol q_t^{b,\mathrm{orc}}\|_N\to0,
\end{equation*}
and
\begin{equation*}
\|\boldsymbol\gamma_t^{b,\mathrm{act}}
-
\boldsymbol\gamma_t^{b,\mathrm{orc}}\|_N\to0.
\end{equation*}
By Proposition~\ref{prop:linear_module_stability},
\begin{equation*}
\|\boldsymbol m_t^{b,\mathrm{act}}
-
\boldsymbol m_t^{b,\mathrm{orc}}\|_N\to0,
\end{equation*}
\begin{equation*}
\|\boldsymbol\pi_t^{b,\mathrm{act}}
-
\boldsymbol\pi_t^{b,\mathrm{orc}}\|_N\to0,
\end{equation*}
and
\begin{equation*}
\|\boldsymbol h_t^{b,\mathrm{act}}
-
\boldsymbol h_t^{b,\mathrm{orc}}\|_N\to0.
\end{equation*}

For the standard branch,
\begin{equation*}
\boldsymbol s_t^{\mathrm{std}}=\boldsymbol h_t,
\qquad
\boldsymbol\rho_t^{\mathrm{std}}=\boldsymbol\pi_t.
\end{equation*}
Therefore the input stability condition of
Proposition~\ref{prop:prior_module_stability} holds, and so
\begin{equation*}
\|\boldsymbol\gamma_{t+1}^{\mathrm{std},\mathrm{act}}
-
\boldsymbol\gamma_{t+1}^{\mathrm{std},\mathrm{orc}}\|_N\to0,
\end{equation*}
\begin{equation*}
\|\boldsymbol q_{t+1}^{\mathrm{std},\mathrm{act}}
-
\boldsymbol q_{t+1}^{\mathrm{std},\mathrm{orc}}\|_N\to0.
\end{equation*}

For the corrected branch,
\begin{equation*}
\boldsymbol s_t^{\mathrm{corr}}
=
\widetilde{\boldsymbol h}_t
=
\boldsymbol h_t-\boldsymbol\mu_t,
\qquad
\boldsymbol\rho_t^{\mathrm{corr}}
=
\boldsymbol\tau_t^{-1}.
\end{equation*}
The stability of \(\boldsymbol h_t\) follows from the linear module, and the
stability of \(\boldsymbol\mu_t\) and \(\boldsymbol\tau_t^{-1}\) follows from
Proposition~\ref{prop:memory_map_stability}.  Hence
\begin{equation*}
\|\boldsymbol s_t^{\mathrm{corr},\mathrm{act}}
-
\boldsymbol s_t^{\mathrm{corr},\mathrm{orc}}\|_N\to0,
\end{equation*}
and
\begin{equation*}
\|\boldsymbol\rho_t^{\mathrm{corr},\mathrm{act}}
-
\boldsymbol\rho_t^{\mathrm{corr},\mathrm{orc}}\|_N\to0.
\end{equation*}
Applying Proposition~\ref{prop:prior_module_stability} gives
\begin{equation*}
\|\boldsymbol\gamma_{t+1}^{\mathrm{corr},\mathrm{act}}
-
\boldsymbol\gamma_{t+1}^{\mathrm{corr},\mathrm{orc}}\|_N\to0,
\end{equation*}
and
\begin{equation*}
\|\boldsymbol q_{t+1}^{\mathrm{corr},\mathrm{act}}
-
\boldsymbol q_{t+1}^{\mathrm{corr},\mathrm{orc}}\|_N\to0.
\end{equation*}

Thus, in either branch, closeness at time \(t\) implies closeness at time
\(t+1\).  Since the initial actual and oracle recursions are identical and the
horizon \(T\) is fixed, finite induction proves the branch-wise precision
replacement theorem.

\bibliographystyle{IEEEtran}
\bibliography{main} 

@article{DonohoMalekiMontanari2009,
  author  = {D. L. Donoho and A. Maleki and A. Montanari},
  title   = {Message passing algorithms for compressed sensing},
  journal = {Proc. Natl. Acad. Sci. U.S.A.},
  volume  = {106},
  number  = {45},
  pages   = {18914--18919},
  month   = nov,
  year    = {2009},
  doi     = {10.1073/pnas.0909892106},
  note    = {\doi{10.1073/pnas.0909892106}},
}

@article{BayatiMontanari2011,
  author  = {M. Bayati and A. Montanari},
  title   = {The dynamics of message passing on dense graphs, with applications to compressed sensing},
  journal = {IEEE Trans. Inf. Theory},
  volume  = {57},
  number  = {2},
  pages   = {764--785},
  month   = feb,
  year    = {2011},
  doi     = {10.1109/TIT.2010.2094817},
  note    = {\doi{10.1109/TIT.2010.2094817}},
}

@article{BayatiLelargeMontanari2015,
  author  = {M. Bayati and M. Lelarge and A. Montanari},
  title   = {Universality in polytope phase transitions and message passing algorithms},
  journal = {Ann. Appl. Probab.},
  volume  = {25},
  number  = {2},
  pages   = {753--822},
  year    = {2015},
  doi     = {10.1214/14-AAP1010},
  note    = {\doi{10.1214/14-AAP1010}},
}

@article{JavanmardMontanari2013,
  author  = {A. Javanmard and A. Montanari},
  title   = {State evolution for general approximate message passing algorithms, with applications to spatial coupling},
  journal = {Inf. Inference},
  volume  = {2},
  number  = {2},
  pages   = {115--144},
  year    = {2013},
  doi     = {10.1093/imaiai/iat004},
  note    = {\doi{10.1093/imaiai/iat004}},
}

@article{BerthierMontanariNguyen2020,
  author  = {R. Berthier and A. Montanari and P.-M. Nguyen},
  title   = {State evolution for approximate message passing with non-separable functions},
  journal = {Inf. Inference},
  volume  = {9},
  number  = {1},
  pages   = {33--79},
  year    = {2020},
  doi     = {10.1093/imaiai/iay021},
  note    = {\doi{10.1093/imaiai/iay021}},
}

@article{FengVenkataramananRushSamworth2022,
  author  = {O. Y. Feng and R. Venkataramanan and C. G. Rush and R. J. Samworth},
  title   = {A unifying tutorial on approximate message passing},
  journal = {Found. Trends Mach. Learn.},
  volume  = {15},
  number  = {4},
  pages   = {335--536},
  year    = {2022},
  doi     = {10.1561/2200000092},
  note    = {\doi{10.1561/2200000092}},
}

@inproceedings{Rangan2011,
  author    = {S. Rangan},
  title     = {Generalized approximate message passing for estimation with random linear mixing},
  booktitle = {Proc. IEEE Int. Symp. Inf. Theory (ISIT)},
  pages     = {2168--2172},
  year      = {2011},
  doi       = {10.1109/ISIT.2011.6033942},
  note    = {\doi{10.1109/ISIT.2011.6033942}},
}

@inproceedings{Minka2001,
  author    = {T. P. Minka},
  title     = {Expectation propagation for approximate {Bayesian} inference},
  booktitle = {Proc. 17th Conf. Uncertainty Artif. Intell. (UAI)},
  pages     = {362--369},
  publisher = {Morgan Kaufmann},
  year      = {2001},
  doi     = {10.5555/647235.720257},
  note    = {\doi{10.5555/647235.720257}},
}

@article{OpperWinther2001,
  author  = {M. Opper and O. Winther},
  title   = {Tractable approximations for probabilistic models: The adaptive {TAP} mean field approach},
  journal = {Phys. Rev. Lett.},
  volume  = {86},
  number  = {17},
  pages   = {3695--3698},
  year    = {2001},
  doi     = {10.1103/PhysRevLett.86.3695},
  note    = {\doi{10.1103/PhysRevLett.86.3695}},
}

@article{OpperWinther2005,
  author  = {M. Opper and O. Winther},
  title   = {Expectation consistent approximate inference},
  journal = {J. Mach. Learn. Res.},
  volume  = {6},
  pages   = {2177--2204},
  month   = dec,
  year    = {2005},
  doi     = {10.5555/1046920.1194917},
  note    = {\doi{10.5555/1046920.1194917}},
}

@article{MaPing2017,
  author  = {J. Ma and L. Ping},
  title   = {Orthogonal {AMP}},
  journal = {IEEE Access},
  volume  = {5},
  pages   = {2020--2033},
  year    = {2017},
  doi     = {10.1109/ACCESS.2017.2653119},
  note    = {\doi{10.1109/ACCESS.2017.2653119}},
}

@article{RanganSchniterFletcher2019,
  author  = {S. Rangan and P. Schniter and A. K. Fletcher},
  title   = {Vector approximate message passing},
  journal = {IEEE Trans. Inf. Theory},
  volume  = {65},
  number  = {10},
  pages   = {6664--6684},
  month   = oct,
  year    = {2019},
  doi     = {10.1109/TIT.2019.2916359},
  note    = {\doi{10.1109/TIT.2019.2916359}},
}

@article{Takeuchi2020,
  author  = {K. Takeuchi},
  title   = {Rigorous dynamics of expectation-propagation-based signal recovery from unitarily invariant measurements},
  journal = {IEEE Trans. Inf. Theory},
  volume  = {66},
  number  = {1},
  pages   = {368--386},
  month   = jan,
  year    = {2020},
  doi     = {10.1109/TIT.2019.2947058},
  note    = {\doi{10.1109/TIT.2019.2947058}},
}

@article{Takeuchi2022LongMemory,
  author  = {K. Takeuchi},
  title   = {On the convergence of orthogonal/vector {AMP}: Long-memory message-passing strategy},
  journal = {IEEE Trans. Inf. Theory},
  volume  = {68},
  number  = {12},
  pages   = {8121--8138},
  month   = dec,
  year    = {2022},
  doi     = {10.1109/TIT.2022.3194855},
  note    = {\doi{10.1109/TIT.2022.3194855}},
}

@article{Takeuchi2024SpatialOAMP,
  author  = {K. Takeuchi},
  title   = {Orthogonal approximate message-passing for spatially coupled linear models},
  journal = {IEEE Trans. Inf. Theory},
  volume  = {70},
  number  = {1},
  pages   = {594--631},
  month   = jan,
  year    = {2024},
  doi     = {10.1109/TIT.2023.3311408},
  note    = {\doi{10.1109/TIT.2023.3311408}},
}

@article{TakahashiKabashima2022,
  author  = {T. Takahashi and Y. Kabashima},
  title   = {Macroscopic analysis of vector approximate message passing in a model-mismatched setting},
  journal = {IEEE Trans. Inf. Theory},
  volume  = {68},
  number  = {8},
  pages   = {5579--5600},
  month   = aug,
  year    = {2022},
  doi     = {10.1109/TIT.2022.3163342},
  note    = {\doi{10.1109/TIT.2022.3163342}},
}

@article{MaXuMaleki2024,
  author  = {J. Ma and J. Xu and A. Maleki},
  title   = {Toward designing optimal sensing matrices for generalized linear inverse problems},
  journal = {IEEE Trans. Inf. Theory},
  volume  = {70},
  number  = {1},
  pages   = {482--508},
  month   = jan,
  year    = {2024},
  doi     = {10.1109/TIT.2023.3307553},
  note    = {\doi{10.1109/TIT.2023.3307553}},
}

@article{CademartoriRush2024,
  author  = {C. Cademartori and C. Rush and A. De Jesus Martins da Silva},
  title   = {A non-asymptotic analysis of generalized vector approximate message passing algorithms with rotationally invariant designs},
  journal = {IEEE Trans. Inf. Theory},
  volume  = {70},
  number  = {8},
  pages   = {5811--5856},
  month   = aug,
  year    = {2024},
  doi     = {10.1109/TIT.2024.3396472},
  note    = {\doi{10.1109/TIT.2024.3396472}},
}

@article{Fan2022,
  author  = {Z. Fan},
  title   = {Approximate message passing algorithms for rotationally invariant matrices},
  journal = {Ann. Statist.},
  volume  = {50},
  number  = {1},
  pages   = {197--224},
  year    = {2022},
  doi     = {10.1214/21-AOS2101},
  note    = {\doi{10.1214/21-AOS2101}},
}

@article{ZhongWangFan2024,
  author  = {X. Zhong and T. Wang and Z. Fan},
  title   = {Approximate message passing for orthogonally invariant ensembles: Multivariate non-linearities and spectral initialization},
  journal = {Inf. Inference},
  volume  = {13},
  number  = {3},
  pages   = {iaae024},
  year    = {2024},
  doi     = {10.1093/imaiai/iaae024},
  note    = {\doi{10.1093/imaiai/iaae024}},
}

@inproceedings{VenkataramananKoglerMondelli2022,
  author    = {R. Venkataramanan and K. K{"o}gler and M. Mondelli},
  title     = {Estimation in rotationally invariant generalized linear models via approximate message passing},
  booktitle = {Proc. 39th Int. Conf. Mach. Learn. (ICML)},
  series    = {Proc. Mach. Learn. Res.},
  volume    = {162},
  pages     = {22120--22144},
  publisher = {PMLR},
  year      = {2022},
}

@article{LiuMa2024,
  author  = {S. Liu and J. Ma},
  title   = {Unifying {AMP} algorithms for rotationally-invariant models},
  journal = {arXiv preprint arXiv:2412.01574},
  year    = {2024},
  eprint  = {2412.01574},
  archivePrefix= {arXiv},
  doi     = {10.48550/arXiv.2412.01574},
  note    = {\doi{10.48550/arXiv.2412.01574}},
}

@article{DudejaLiuMa2026,
  author  = {R. Dudeja and S. Liu and J. Ma},
  title   = {Optimality of approximate message passing for spiked matrix models with rotationally invariant noise},
  journal = {Ann. Statist.},
  volume  = {54},
  number  = {1},
  year    = {2026},
  doi     = {10.1214/25-AOS2575},
  pages   = {466--489},
  note    = {\doi{10.1214/25-AOS2575}},
}

@article{LiuHuangKurkoski2022,
  author  = {L. Liu and S. Huang and B. M. Kurkoski},
  title   = {Memory {AMP}},
  journal = {IEEE Trans. Inf. Theory},
  volume  = {68},
  number  = {12},
  pages   = {8015--8039},
  month   = dec,
  year    = {2022},
  doi     = {10.1109/TIT.2022.3186166},
  note    = {\doi{10.1109/TIT.2022.3186166}},
}

@inproceedings{LiuHuangKurkoski2022SSMAMP,
  author    = {L. Liu and S. Huang and B. M. Kurkoski},
  title     = {Sufficient statistic memory approximate message passing},
  booktitle = {Proc. IEEE Int. Symp. Inf. Theory (ISIT)},
  pages     = {157--162},
  year      = {2022},
  doi     = {10.1109/ISIT50566.2022.9834568},
  note    = {\doi{10.1109/ISIT50566.2022.9834568}},
}

@article{WangZhongFan2024,
  author  = {T. Wang and X. Zhong and Z. Fan},
  title   = {Universality of approximate message passing algorithms and tensor networks},
  journal = {Ann. Appl. Probab.},
  volume  = {34},
  number  = {4},
  pages   = {3943--3994},
  year    = {2024},
  doi     = {10.1214/24-AAP2056},
  note    = {\doi{10.1214/24-AAP2056}},
}

@article{Hachem2024,
  author  = {W. Hachem},
  title   = {Approximate message passing for sparse matrices with application to the equilibria of large ecological {Lotka--Volterra} systems},
  journal = {Stochastic Process. Appl.},
  volume  = {170},
  pages   = {104276},
  year    = {2024},
  doi     = {10.1016/j.spa.2023.104276},
  note    = {\doi{10.1016/j.spa.2023.104276}},
}

@article{GueddariHachemNajim2025Elliptic,
  author  = {M.-Y. Gueddari and W. Hachem and J. Najim},
  title   = {Elliptic approximate message passing and an application to theoretical ecology},
  journal = {Random Matrices Theory Appl.},
  volume  = {14},
  number  = {4},
  year    = {2025},
  doi     = {10.1142/S2010326325500182},
  pages   = {2550018},
  note    = {\doi{10.1142/S2010326325500182}},
}

@article{GueddariHachemNajim2026General,
  author  = {M.-Y. Gueddari and W. Hachem and J. Najim},
  title   = {Approximate message passing for general non-symmetric random matrices},
  journal = {J. Theor. Probab.},
  year    = {2026},
  doi     = {10.1007/s10959-025-01476-z},
  volume  = {39},
  number  = {1},
  pages   = {1--69},
  note    = {\doi{10.1007/s10959-025-01476-z}},
}

@article{BaoHanXu2023,
  author  = {Z. Bao and Q. Han and X. Xu},
  title   = {A leave-one-out approach to approximate message passing},
  journal = {Ann. Appl. Probab.},
  year    = {2025},
  volume  = {35},
  number  = {4},
  pages   = {2716--2766},
  doi     = {10.1214/25-AAP2186},
  note    = {\doi{10.1214/25-AAP2186}},
}

@article{BarbierKrzakala2017,
  author  = {J. Barbier and F. Krzakala},
  title   = {Approximate message-passing decoder and capacity-achieving sparse superposition codes},
  journal = {IEEE Trans. Inf. Theory},
  volume  = {63},
  number  = {8},
  pages   = {4894--4927},
  month   = aug,
  year    = {2017},
  doi     = {10.1109/TIT.2017.2713833},
  note    = {\doi{10.1109/TIT.2017.2713833}},
}

@article{RushGreigVenkataramanan2017,
  author  = {C. Rush and A. Greig and R. Venkataramanan},
  title   = {Capacity-achieving sparse superposition codes via approximate message passing decoding},
  journal = {IEEE Trans. Inf. Theory},
  volume  = {63},
  number  = {3},
  pages   = {1476--1500},
  month   = mar,
  year    = {2017},
  doi     = {10.1109/TIT.2017.2649460},
  note    = {\doi{10.1109/TIT.2017.2649460}},
}

@article{MontanariVenkataramanan2021,
  author  = {A. Montanari and R. Venkataramanan},
  title   = {Estimation of low-rank matrices via approximate message passing},
  journal = {Ann. Statist.},
  volume  = {49},
  number  = {1},
  pages   = {321--345},
  year    = {2021},
  doi     = {10.1214/20-AOS1984},
  note    = {\doi{10.1214/20-AOS1984}},
}

@article{LelargeMiolane2019,
  author  = {M. Lelarge and L. Miolane},
  title   = {Fundamental limits of symmetric low-rank matrix estimation},
  journal = {Probab. Theory Relat. Fields},
  volume  = {173},
  pages   = {859--929},
  year    = {2019},
  doi     = {10.1007/s00440-018-0845-x},
  note    = {\doi{10.1007/s00440-018-0845-x}},
}

@article{GuionnetKoKrzakalaZdeborova2022,
  author  = {A. Guionnet and J. Ko and F. Krzakala and L. Zdeborov{\'a}},
  title   = {Low-rank matrix estimation with inhomogeneous noise},
  journal = {Inf. Inference},
  year    = {2025},
  volume  = {14},
  number  = {2},
  pages   = {iaaf010},
  doi     = {10.1093/imaiai/iaaf010},
  note    = {\doi{10.1093/imaiai/iaaf010}},
}

@inproceedings{PakKoKrzakala2023,
  author    = {A. Pak and J. Ko and F. Krzakala},
  title     = {Optimal algorithms for the inhomogeneous spiked {Wigner} model},
  booktitle = {Adv. Neural Inf. Process. Syst.},
  volume    = {36},
  pages     = {76409--76424},
  year      = {2023},
  doi     = {10.52202/075280-3340},
  note    = {\doi{10.52202/075280-3340}},
}

@article{AjankiErdosKruger2017,
  author  = {O. H. Ajanki and L. Erd{\H{o}}s and T. Kr{\"u}ger},
  title   = {Universality for general {Wigner}-type matrices},
  journal = {Probab. Theory Relat. Fields},
  volume  = {169},
  number  = {3--4},
  pages   = {667--727},
  year    = {2017},
  doi     = {10.1007/s00440-016-0740-2},
  note    = {\doi{10.1007/s00440-016-0740-2}},
}

@article{AjankiErdosKruger2019,
  author  = {O. H. Ajanki and L. Erd{\H{o}}s and T. Kr{\"u}ger},
  title   = {Stability of the matrix {Dyson} equation and random matrices with correlations},
  journal = {Probab. Theory Relat. Fields},
  volume  = {173},
  number  = {1--2},
  pages   = {293--373},
  year    = {2019},
  doi     = {10.1007/s00440-018-0835-z},
  note    = {\doi{10.1007/s00440-018-0835-z}},
}

@incollection{Erdos2019MDE,
  author    = {L. Erd{\H{o}}s},
  title     = {The matrix {Dyson} equation and its applications for random matrices},
  booktitle = {Random Matrices},
  series    = {IAS/Park City Math. Ser.},
  volume    = {26},
  publisher = {Amer. Math. Soc.},
  address   = {Providence, RI, USA},
  year      = {2019},
  doi     = {10.1090/pcms/026/03},
  pages   = {75--158},
  note    = {\doi{10.1090/pcms/026/03}},
}

@book{BaiSilverstein2010,
  author    = {Z. Bai and J. W. Silverstein},
  title     = {Spectral Analysis of Large Dimensional Random Matrices},
  edition   = {2},
  publisher = {Springer},
  address   = {New York, NY, USA},
  year      = {2010},
  doi     = {10.1007/978-1-4419-0661-8},
  note    = {\doi{10.1007/978-1-4419-0661-8}},
}

@book{CouilletDebbah2011,
  author    = {R. Couillet and M. Debbah},
  title     = {Random Matrix Methods for Wireless Communications},
  publisher = {Cambridge Univ. Press},
  address   = {Cambridge, U.K.},
  year      = {2011},
  doi     = {10.1017/CBO9780511994746},
  note    = {\doi{10.1017/CBO9780511994746}},
}

@book{Anderson2003,
  author    = {T. W. Anderson},
  title     = {An Introduction to Multivariate Statistical Analysis},
  edition   = {3},
  publisher = {Wiley},
  address   = {Hoboken, NJ, USA},
  year      = {2003},
}

@book{Chatterjee2014,
  author    = {S. Chatterjee},
  title     = {Superconcentration and Related Topics},
  publisher = {Springer},
  address   = {Cham, Switzerland},
  year      = {2014},
  doi       = {10.1007/978-3-319-03886-5},
  note    = {\doi{10.1007/978-3-319-03886-5}},
}

@book{MezardMontanari2009,
  author    = {M. M{\'e}zard and A. Montanari},
  title     = {Information, Physics, and Computation},
  publisher = {Oxford Univ. Press},
  address   = {Oxford, U.K.},
  year      = {2009},
  doi     = {10.1093/acprof:oso/9780198570837.001.0001},
  note    = {\doi{10.1093/acprof:oso/9780198570837.001.0001}},
}

@article{Kabashima2003,
  author  = {Y. Kabashima},
  title   = {A {CDMA} multiuser detection algorithm on the basis of belief propagation},
  journal = {J. Phys. A, Math. Gen.},
  volume  = {36},
  number  = {43},
  pages   = {11111--11121},
  year    = {2003},
  doi     = {10.1088/0305-4470/36/43/030},
  note    = {\doi{10.1088/0305-4470/36/43/030}},
}

\end{document}